\documentclass[onecolumn]{IEEEtran}

\usepackage{url}

\usepackage{amsfonts}
\usepackage{amsmath}
\usepackage{amsthm}
\usepackage{amssymb}
\newtheorem{theorem}{Theorem}
\newtheorem{corollary}{Corollary}
\newtheorem{lemma}{Lemma}
\newtheorem{proposition}{Proposition}
\newtheorem{definition}{Definition}
\newtheorem{remark}{Remark}
\newtheorem{example}{Example}
\newtheorem{claim}{Claim}
\newcommand{\comment}[1]{}
\usepackage{graphicx}
\usepackage{subfigure}
\usepackage{algorithm}
\usepackage{algorithmic}
\usepackage{multirow}

\begin{document}
\title{Completing Low-Rank Matrices with Corrupted Samples from Few Coefficients in General Basis}

\author{Hongyang Zhang$^{\dag\ddag}$,
        \and Zhouchen Lin$^{\ddag*}$,
        \and Chao Zhang$^{\ddag*}$% <-this % stops a space
        \\$^\dag$Machine Learning Department, School of Computer Science, Carnegie Mellon University
        \\$^\ddag$Key Laboratory of Machine Perception (MOE), School of EECS, Peking University
        \\$^*$Cooperative Medianet Innovation Center, Shanghai Jiao Tong University
        \\hongyanz@cs.cmu.edu, zlin@pku.edu.cn, chzhang@cis.pku.edu.cn
}

\maketitle

\begin{abstract}
Subspace recovery from corrupted and missing data is crucial for various applications in signal processing and information theory. To complete missing values and detect column corruptions, existing robust Matrix Completion (MC) methods mostly concentrate on recovering a low-rank matrix from few corrupted coefficients w.r.t. standard basis, which, however, does not apply to more general basis, e.g., Fourier basis. In this paper, we prove that the range space of an $m\times n$ matrix with rank $r$ can be exactly recovered from few coefficients w.r.t. \emph{general basis}, though $r$ and the number of corrupted samples are both as high as $O(\min\{m,n\}/\log^3 (m+n))$. Our model covers previous ones as special cases, and robust MC can recover the intrinsic matrix with a higher rank. Moreover, we suggest a universal choice of the regularization parameter, which is $\lambda=1/\sqrt{\log n}$. By our $\ell_{2,1}$ filtering algorithm, which has theoretical guarantees, we can further reduce the computational cost of our model. As an application, we also find that the solutions to extended robust Low-Rank Representation and to our extended robust MC are mutually expressible, so both our theory and algorithm can be applied to the subspace clustering problem with missing values under certain conditions. Experiments verify our theories.
\end{abstract}


% \fi
\IEEEpeerreviewmaketitle

\section{Introduction}
\label{section: introduction}
We are now in an era of big and high-dimensional data. Unfortunately, due to the storage difficulty and the computational obstacle, we can measure only a few entries from the data matrix. So restoring all of the information that the data carry through the partial measurements is of great interest in data analysis. This challenging problem is also known as the Matrix Completion (MC) problem, which is highly related to the so-called recommendation system, where one tries to predict unrevealed users' preference according to the incomplete rating feedback. Admittedly, this inverse problem is ill-posed as there should be infinite number of feasible solutions. Fortunately, most of the data are structured, e.g., face~\cite{LiuG1}, texture~\cite{Ma2}, and motion~\cite{Gear,Yan,Rao}. They typically lie around low-dimensional subspaces. Because the rank of data matrix corresponds to the dimensionality of subspace, recent work~\cite{Candes2009exact,Candes2010power,Gross2011recovering,recht2011simpler} in convex optimization demonstrates a remarkable fact: it is possible to exactly complete an $m\times n$ matrix of rank $r$, if the number of randomly selected matrix elements is no less than $O((m+n)r\log^2 (m+n))$. % The intuition behind this is about the degree of freedom, namely, if the degree of the freedom of the observed entries is no less than that of the low-rank matrix, one can hope to exactly recover the original matrix with an overwhelming probability.

Yet it is well known that the traditional MC model suffers from the robustness issue. It is even sensitive to minor corruptions, which commonly occur due to sensor failures and uncontrolled environments. In the recommendation system, for instance, malicious manipulation of even a single rater might drive the output of MC algorithm far from the ground truth. To resolve the issue, several efforts have been devoted to robustifying the MC model, among which robust MC~\cite{ICML2011Chen_469} is the one with solid theoretical analysis. Chen et al.~\cite{ICML2011Chen_469} proved that robust MC is able to exactly recover the ground truth subspace and detect the column corruptions (i.e., some entire columns are corrupted by noises), if the dimensionality of subspace is not too high and the corrupted columns are sparse compared with the input size. Most importantly, the observed expansion coefficients should be sufficient w.r.t. the standard matrix basis $\{e_ie_j^T\}_{ij}$ (please refer to Table \ref{table: notations} for explanation of
notations).

However, recent advances in theoretical physics measure quantum-state entries by tomography w.r.t. the Pauli basis, which is rather different from the standard matrix one~\cite{Gross2011recovering}. So it is not very straightforward to apply the existing theory on robust MC to such a special case. This paper tries to resolve the problem. More generally, we demonstrate the exact recoverability of an extended robust MC model in the presence of only a few coefficients w.r.t. a set of general basis, although some columns of the intrinsic matrix might be arbitrarily corrupted. By applying our $\ell_{2,1}$ filtering algorithm which has theoretical guarantees, we are able to speed up solving the model numerically. There are various applications of our results.

\subsection{Practical Applications}
\label{subsection: application}
In numerical analysis, instead of the standard polynomial basis $\{x^k\}_{k=0}^n$, the Legendre polynomials are widely used to represent smooth functions due to their orthogonality. Such expansions, however, are typically sensitive to perturbation: a small perturbation of the function might arbitrarily drive the fitting result far from its original. Moreover, to reduce the storage and the computational costs, sometimes we can record only a few expansion coefficients. To complete the missing values and get the outliers removed, this paper justifies the possibility of doing so.

In digital signal processing, one usually samples the signals, e.g., voices and feature vectors, at random in the Fourier basis. However, due to sensor failure, a group of signals that we capture may be rather unreliable. To recover the intrinsic information that the signals carry and remove the outliers simultaneously, our theoretical analysis guarantees the success of robust MC w.r.t. the Fourier basis.

In quantum information theory, to obtain a maximum likelihood estimation of a quantum state of 8 ions, one typically requires hundred of thousands of measurements w.r.t. the Pauli basis, which are unaffordable because of high experimental costs. To overcome the difficulty, Gross~\cite{Gross2011recovering} compressed the number of observations w.r.t. any basis by an MC model. However, their model is fragile to severe corruptions, which commonly occurs because of measurement errors. To robustify the model, this paper justifies the exact recoverability of robust MC w.r.t. general basis, even if the datasets are wildly corrupted.

In subspace clustering, one tries to segment the data points according to the subspaces they lie in, which can be widely applied to motion
segmentation~\cite{Gear,Yan,Rao,Costeira,Vidal2004CVPR}, face classification~\cite{LiuG1,Ho2003,Vidal2005PAMI,LiuG3}, system identification~\cite{Vidal2003CDC,ZhangC,Paoletti2007}, and image segmentation~\cite{Yang2008CVIU,Cheng}. Recently, it is of great interest to cluster the subspaces while the observations w.r.t. some coordinates are missing. To resolve the issue, as an application in this paper, our theorem relates robust MC to a certain subspace clustering model -- the so-called extended robust Low-Rank Representation (LRR). Thus one could hope to correctly recover the structure of multiple subspaces, if robust MC is able to complete the unavailable values and remove the outlier samples at an overwhelming probability. This is guaranteed by our paper.

\subsection{Related Work}
Suppose that $L_0$ is an $m\times n$ data matrix of rank $r$ whose columns are sample points, and entries are partially observed among the set $\mathcal{K}_{obs}$. The MC problem aims at exactly recovering $L_0$, or the range space of $L_0$, from the measured elements. Probably the most well-known MC model was proposed by Cand\`{e}s et al.~\cite{Candes2009exact}. To choose the lowest-rank matrix so as to fit the observed entries, the original model is formulated as
\begin{equation}
\label{equ: original MC}
\min_L \mbox{rank}(L),\quad\mbox{s.t.}\quad \langle L,e_ie_j^T\rangle=\langle L_0,e_ie_j^T\rangle,\quad(i,j)\in\mathcal{K}_{obs}.
\end{equation}
This model, however, is untractable because problem \eqref{equ: original MC} is NP-hard. Inspired by recent work in compressive sensing, Cand\`{e}s et al. replaced the rank in the objective function with the nuclear norm, which is the sum of singular values and is the convex envelope of rank on the unit ball of matrix operator norm. Namely,
\begin{equation}
\label{equ: relaxed MC}
\min_L \|L\|_*,\quad\mbox{s.t.}\quad \langle L,e_ie_j^T\rangle=\langle L_0,e_ie_j^T\rangle,\quad(i,j)\in\mathcal{K}_{obs}.
\end{equation}
It is worth noting that model \eqref{equ: relaxed MC} is only w.r.t. the standard matrix basis $\{e_ie_j^T\}_{ij}$. To extend the model to any basis $\{\omega_{ij}\}_{ij}$, Gross~\cite{Gross2011recovering} proposed a more general MC model:
\begin{equation}
\label{equ: general basis MC}
\min_L \|L\|_*,\quad\mbox{s.t.}\quad \langle L,\omega_{ij}\rangle=\langle L_0,\omega_{ij}\rangle,\quad(i,j)\in\mathcal{K}_{obs}.
\end{equation}
Models \eqref{equ: relaxed MC} and \eqref{equ: general basis MC} both have solid theoretical guarantees: recent work~\cite{Candes2009exact,Candes2010power,Gross2011recovering,recht2011simpler} showed that the models are able to exactly recover the ground truth $L_0$ by an overwhelming probability, if $\mathcal{K}_{obs}$ is uniformly distributed among all sets of cardinality $O((m+n)r\log^2 (m+n))$.
Unfortunately, these traditional MC models suffer from the robustness issue: they are even sensitive to minor corruptions, which commonly occurs due to sensor failures, uncontrolled environments, etc.

A parallel study to the MC problem is the so-called matrix recovery, namely, recovering underlying data matrix $L_0$, or the range space of $L_0$, from the corrupted data matrix $M=L_0+S_0$, where $S_0$ is the noise. Probably the most widely used one is Principal Component Analysis (PCA). However, PCA is fragile to outliers. Even a single but severe corruption may wildly degrade the performance of PCA. To resolve the issue, much work has been devoted to robustifying PCA~\cite{gnanadesikan1972robust,huber2011robust,fischler1981random,Torre,ke2005robust,mccoy2011two,zhang2014novel,lerman2014robust,hardt2012algorithms}, among which a simple yet successful model to remove column corruptions is robust PCA via Outlier Pursuit:
\begin{equation}
\label{equ: original outlier pursuit}
\min_{L,S} \mbox{rank}(L)+\lambda\|S\|_{2,1},\quad\mbox{s.t.}\quad M=L+S,
\end{equation}
and its convex relaxation
\begin{equation}
\label{equ: relaxed outlier pursuit}
\min_{L,S} \|L\|_*+\lambda\|S\|_{2,1},\quad\mbox{s.t.}\quad M=L+S.
\end{equation}
Outlier Pursuit has theoretical guarantees: Xu et al.~\cite{Xu:CSPCP} and our previous work~\cite{Zhang2015AAAI} proved that when the dimensionality of ground truth subspace is not too high and the column-wise corruptions are sparse compared with the sample size, Outlier Pursuit is able to recover the range space of $L_0$ and detect the non-zero columns of $S_0$ at an overwhelming probability. Nowadays, Outlier Pursuit has been widely applied to subspace clustering~\cite{Zhang:RobustLatLRR}, image alignment~\cite{Peng}, texture
representation~\cite{ZhangZ}, etc. Unfortunately, the model cannot handle the case of missing values, which significantly limits its working range in practice.

It is worth noting that the pros and cons of above-mentioned MC and Outlier Pursuit are mutually complementary. To remedy both of their limitations, recent work~\cite{ICML2011Chen_469} suggested combining the two models together, resulting in robust MC -- a model that could complete the missing values and detect the column corruptions simultaneously. Specifically, it is formulated as
\begin{equation}
\label{equ: original robust MC}
\begin{split}
&\min_{L,S} \mbox{rank}(L)+\lambda\|S\|_{2,1},\\&\ \mbox{s.t.}\quad \langle M,e_ie_j^*\rangle=\langle L+S,e_ie_j^*\rangle,\quad(i,j)\in\mathcal{K}_{obs}.
\end{split}
\end{equation}
Correspondingly, the relaxed form is
\begin{equation}
\label{equ: relaxed robust MC}
\begin{split}
&\min_{L,S} \|L\|_*+\lambda\|S\|_{2,1},\\&\ \mbox{s.t.}\quad \langle M,e_ie_j^*\rangle=\langle L+S,e_ie_j^*\rangle,\quad(i,j)\in\mathcal{K}_{obs}.
\end{split}
\end{equation}
Chen at al.~\cite{ICML2011Chen_469} demonstrated the recoverability of model \eqref{equ: relaxed robust MC}, namely, if the range space of $L_0$ is low-dimensional, the observed entries are sufficient, and the column corruptions are sparse compared with the input size, one can hope to exactly recover the range space of $L_0$ and detect the corrupted samples by robust MC at an overwhelming probability. It is well reported that robust MC has been widely applied to recommendation system and medical research~\cite{ICML2011Chen_469}. However, the specific basis $\{e_ie_j^T\}_{ij}$ in problem \eqref{equ: relaxed robust MC} limits its extensible applications to more challenging tasks, such as those discussed in Section \ref{subsection: application}.

\subsection{Our Contributions}
In this paper, we extend robust MC to more general cases, namely, the expansion coefficients are observed w.r.t. a set of general basis. We are particularly interested in the exact recoverability of this extended model. Our contributions are as follows:
\begin{itemize}
\item
We demonstrate that the extended robust MC model succeeds at an overwhelming probability. This result broadens the working range of traditional robust MC in three aspects: 1. the choice of basis in our model is not limited to the standard one anymore; 2. with slightly stronger yet reasonable incoherence (ambiguity) conditions, our result allows $\mbox{rank}(L_0)$ to be as high as $O(n/\log^3 n)$ even when the number of corruptions and observations are both constant fraction of the total input size. In comparison with the existing result which requires that $\mbox{rank}(L_0)=O(1)$, our analysis significantly extends the succeeding range of robust MC model; 3. we suggest that the regularization parameter be chosen as $\lambda=1/\sqrt{\log n}$, which is universal.
\item
We propose a so-called $\ell_{2,1}$ filtering algorithm to reduce the computational complexity of our model. Furthermore, we establish theoretical guarantees for our algorithm, which are elegantly relevant to the incoherence of the low-rank component.
\item
As an application, we relate the extended robust MC model to a certain subspace clustering model -- extended robust LRR. So both our theory and our algorithm on the extended robust MC can be applied to the subspace clustering problem if the extended robust MC can exactly recover the data structure.
\end{itemize}

\subsubsection{Novelty of Our Analysis Technique}
In the analysis of the exact recoverability of the model, we novelly divide the proof of Theorem \ref{theorem: exact recovery under Bernoulli sampling} into two parts: The exact recoverability of column support and the exact recoverability of column space. We are able to attack the two problems separately thanks to the idea of expanding the objective function at the well-designed points, i.e., $(\widetilde{L},\widetilde{S})$ for the recovery of column support and $(\widehat{L},\widehat{S})$ for the recovery of column space, respectively (see Sections \ref{subsubsection: dual conditions 1} and \ref{subsubsection: dual conditions 2} for details). This technique enables us to decouple the randomization of $\mathcal{I}_0$ and $\Omega_{obs}$, and so construct the dual variables easily by standard tools like the least squares and golfing scheme. We notice that our framework is general. It not only can be applied to the proof for easier model like Outlier Pursuit~\cite{Zhang2015AAAI} (though we will sacrifice a small polylog factor for the probability of outliers), but can also hopefully simplify the proof for model with more complicated formulation. That is roughly the high-level intuition why we can handle the general basis in this paper.

In the analysis for our $\ell_{2,1}$ filtering algorithm, we take advantage of the low-rank property, namely, we recover a small-sized seed matrix first and then use the linear representation to obtain the whole desired matrix. Our analysis employs tools in recent matrix concentration literature~\cite{Tropp} to bound the size of the seed matrix, which elegantly relates to the incoherence of the underlying matrix. This is definitely consistent with the fact that, for matrix with high incoherence, we typically need to sample more columns in order to fully observe the maximal linearly independent group (see Algorithm \ref{algl_21 filtering} for the procedure).

The remainder of this paper is organized as follows. Section \ref{section: problem setup} describes the problem setup. Section \ref{section: exact recoverability of the model} shows our theoretical results, i.e., the exact recoverability of our model. In Section \ref{section: proof}, we present the detailed proofs of our main results. Section \ref{section: algorithm} proposes a novel $\ell_{2,1}$ filtering algorithm for the extended robust MC model, and establishes theoretical guarantees for the algorithm. We show an application of our analysis to subspace clustering problem, and demonstrate the validity of our theory by experiments in Section \ref{section: applications and experiments}. Finally, Section \ref{section: conclusion} concludes the paper.

\section{Problem Setup}
\label{section: problem setup}
Suppose that $L_0$ is an $m\times n$ data matrix of rank $r$, whose columns are sample points. $S_0\in\mathbb{R}^{m\times n}$ is a noise matrix, whose column support is sparse compared with the input size $n$. Let $M=L_0+S_0$. Its expansion coefficients w.r.t. a set of general basis $\{\omega_{ij}\}_{ij}$, $(i,j)\in\mathcal{K}_{obs}$, are partially observed. This paper considers the exact recovery problem as defined below.
\begin{definition}[Exact Recovery Problem]
\label{def: Exact Recovery Problem} The exact recovery
problem investigates whether the range space of $L_0$ and the
column support of $S_0$ can be exactly recovered from randomly selected coefficients of $M$ w.r.t. general basis, provided that some columns of $M$ are arbitrarily corrupted.
\end{definition}
A similar problem was proposed in \cite{Zhang:Counterexample,Candes}, which recovered the whole matrix $L_0$ and $S_0$ themselves if $S_0$ has element-wise support. However, it is worth noting that one can only hope to recover the range space of $L_0$ and the column support of $S_0$ in Definition \ref{def: Exact Recovery Problem}, because a
corrupted column can be addition of any one vector in the range space
of $L_0$ and another appropriate vector~\cite{Xu:CSPCP,ICML2011Chen_469,Zhang2015AAAI}. Moreover, as existing work mostly concentrates on recovering a low-rank matrix from a sampling of \emph{matrix elements}, our exact recovery problem covers this situation as a special case.

%A similar problem to Definition \ref{def: Exact Recovery Problem} has been proposed for robust MC, where one specifies $\{\omega_{ij}\}_{ij}$ as the standard matrix basis $\{e_ie_j^T\}_{ij}$. Thus our exact recovery problem is more challenging.

%However, Definition \ref{def: Exact Recovery Problem} for Outlier
%Pursuit has its own characteristic: one can only expect to recover
%the column space of $L_0$ and the column support of $S_0$, rather
%than the whole $L_0$ and $S_0$ themselves. This is because a
%corrupted sample can be addition of any vector in the column space
%of $L_0$ and another appropriate vector.

\subsection{Model Formulations}
As our exact recovery problem defines, we study an extended robust MC model w.r.t. a set of general basis. To choose the solution $L$ with the lowest rank, the original model is formulated as
\begin{equation}
\label{equ: original robust MC wrt any basis 1}
\begin{split}
&\min_{L,S} \mbox{rank}(L)+\lambda\|S\|_{2,1},\\&\ \mbox{s.t.}\ \langle M, \omega_{ij}\rangle=\langle L+S, \omega_{ij}\rangle,\ \ (i,j)\in\mathcal{K}_{obs},
\end{split}
\end{equation}
where $\mathcal{K}_{obs}$ is the observation index and $\{\omega_{ij}\}_{i,j=1}^{m,n}$ is a set of ortho-normal bases such that
\begin{equation}
\label{equ: commutative operators}
\mbox{Span}\{\omega_{ij}, i=1,...,m\}=\mbox{Span}\{e_ie_j^*, i=1,...,m\},\ \ \forall j.
\end{equation}
Unfortunately, problem \eqref{equ: original robust MC wrt any basis 1} is NP-hard because the rank function is discrete. So we replace the rank in the objective function with the nuclear norm, resulting in the relaxed formulation:
\begin{equation}
\label{equ: robust MC wrt any basis 1}
\begin{split}
&\min_{L,S} \|L\|_*+\lambda\|S\|_{2,1},\\&\ \mbox{s.t.}\ \langle M, \omega_{ij}\rangle=\langle L+S, \omega_{ij}\rangle,\ \ (i,j)\in\mathcal{K}_{obs}.
\end{split}
\end{equation}
For brevity, we also rewrite it as
\begin{equation}
\label{equ: robust MC wrt any basis}
\min_{L,S} \|L\|_*+\lambda\|S\|_{2,1},\ \ \mbox{s.t.}\ \ \mathcal{R}(L+S)=\mathcal{R}(M),
\end{equation}
where $\mathcal{R}(\cdot)=\sum_{ij\in\mathcal{K}_{obs}}\langle\cdot,\omega_{ij}\rangle\omega_{ij}$ is an operator which projects a matrix onto the space $\Omega_{obs}=\mbox{Span}\{\omega_{ij},\ i,j\in\mathcal{K}_{obs}\}$, i.e., $\mathcal{R}=\mathcal{P}_{\Omega_{obs}}$.

In this paper, we show that problem \eqref{equ: robust MC wrt any basis 1}, or equivalently problem \eqref{equ: robust MC wrt any basis}, exactly recovers the range space of $L_0$ and the column support of $S_0$, if the rank of $L_0$ is no higher than $O(n/\log^3 n)$, and the number of corruptions and observations are (nearly) constant fractions of the total input size. In other words, the original problem \eqref{equ: original robust MC wrt any basis 1} can be well approximated by the relaxed problem \eqref{equ: robust MC wrt any basis 1}.

\subsection{Assumptions}
At first sight, it seems not always possible to successfully separate $M$ as the low-rank term plus the column-sparse one, because there seems to not be sufficient information to avoid the identifiability issues. The identifiability issues are reflected in two aspects: the true low-rank term might be sparse and the true sparse component might be low-rank, thus we cannot hopefully identify the ground truth correctly. So we require several assumptions in order to avoid such unidentifiable cases.

\subsubsection{Incoherence Conditions on the Low-Rank Term}
As an extreme example, suppose that the low-rank term has only one non-zero entry, e.g., $e_1e_1^*$. This matrix has a one in the top left corner and zeros elsewhere, thus being both low-rank and sparse. So it is impossible to identify this matrix as the low-rank term correctly. Moreover, we cannot expect to recover the range space of this matrix from a sampling of its entries, unless we pretty much observe all of the elements.

To resolve the issue, Gross~\cite{Gross2011recovering} introduced $\mu$-incoherence condition to the low-rank term $L\in\mathbb{R}^{m\times n}$ in problem \eqref{equ: general basis MC} w.r.t. the general basis $\{\omega_{ij}\}_{ij}$:
\begin{subequations}
\begin{align}
&\max_{ij} \|\mathcal{P}_{\mathcal{V}}\omega_{ij}\|_F^2\le\frac{\mu r}{n},\ \mbox{(avoid column sparsity)}\label{equ: incoherence 1}\\
&\max_{ij} \|\mathcal{P}_{\mathcal{U}}\omega_{ij}\|_F^2\le\frac{\mu r}{m},\ \mbox{(avoid row sparsity)}\label{equ: another incoherence}\\
&\max_{ij}\langle UV^*, \omega_{ij}\rangle^2\le\frac{\mu r}{mn},\label{equ: incoherence 2}
\end{align}
\end{subequations}
where $U\Sigma V^*\in\mathbb{R}^{m\times n}$ is the skinny SVD of $L$. Intuitively, as discussed in~\cite{Gross2011recovering,Candes}, conditions \eqref{equ: incoherence 1}, \eqref{equ: another incoherence}, and \eqref{equ: incoherence 2} assert that the singular vectors reasonably spread out for small $\mu$. Because problem \eqref{equ: general basis MC}, which is a noiseless version of problem \eqref{equ: robust MC wrt any basis 1}, requires conditions \eqref{equ: incoherence 1}, \eqref{equ: another incoherence}, and \eqref{equ: incoherence 2} in its theoretical guarantees~\cite{Gross2011recovering}, we will set the same incoherence conditions to analyze our model \eqref{equ: robust MC wrt any basis 1} as well.
We argue that beyond \eqref{equ: incoherence 1}, conditions \eqref{equ: another incoherence} and \eqref{equ: incoherence 2} are indispensible for the exact recovery of the target matrix in our setting. As an example, let few entries in the first row of a matrix be non-zeros while all other elements are zeros. This matrix satisfies condition \eqref{equ: incoherence 1} but does not satisfy \eqref{equ: another incoherence} and \eqref{equ: incoherence 2}. In this scenario the probability of recovering its column space is not very high, as we cannot guarantee to take a sample from those uncorrupted non-zero entries, when there are a large amount of noises.

So we assume that the low-rank part $\widetilde{L}$ satisfies conditions \eqref{equ: incoherence 1}, \eqref{equ: another incoherence}, and \eqref{equ: incoherence 2}, and the low-rank component $\widehat{L}$ satisfies condition \eqref{equ: incoherence 1}, as work \cite{Zhang2015AAAI} did (please refer to Table \ref{table: notations} for explanation of notations). Though it is more natural to assume the incoherence on $L_0$, the following example shows that the incoherence of $L_0$ does not suffice to guarantee the success of model \eqref{equ: robust MC wrt any basis 1} when the rank is relatively high:

%\paragraph{Why not assume incoherence on $L_0$}
%Given the partially observed matrix $M$, one may easily think of a low-rank and column-sparse decomposition: $M=\mathcal{P}_{\Omega_{obs}}(L_0+S_0)$. This decomposition of $M$, however, is not unique. Because we also have $M=\mathcal{P}_{\Omega_{obs}}(\widetilde{L}+\widetilde{S})$ and $M=\mathcal{P}_{\Omega_{obs}}(\widehat{L}+\widehat{S})$ (please refer to Table \ref{table: notations} for notations), both of which fulfill the requirement of the low-rank plus the sparse. So we assume $\widetilde{L}$ and $\widehat{L}$ are incoherent as recent work~\cite{Zhang2015AAAI} did. Moreover, the following example shows that the incoherence of $L_0$ does not suffice to guarantee the success of model \eqref{equ: robust MC wrt any basis 1}.
\begin{example}
\label{example: failure on L0}
Compute $L_0=XY^T$ as a product of $n\times r$
i.i.d. $\mathcal{N}(0,1)$ matrices. The column support of
$S_0$ is sampled by Bernoulli distribution with parameter $a$. Let the first entry of each non-zero column of $S_0$ be $n$ and all other entries be zeros. Also set the observation matrix as $\mathcal{P}_{\Omega_{obs}}(L_0+S_0)$, where $\Omega_{obs}$ is the set of observed index selected by i.i.d. $\mbox{Ber}(p_0)$. We adopt $n=10,000$, $r=0.1n$, $p_0=1$, and $a=10/n$, so there are around constant number of corrupted samples in this example. Note that, here, $L_0$ is incoherent fulfilling conditions \eqref{equ: incoherence 1}, \eqref{equ: another incoherence}, and \eqref{equ: incoherence 2}, while $\widetilde{L}$ and $\widehat{L}$ are not. However, the output of algorithm falsely identifies all of the corrupted samples as the clean data. So the incoherence of $L_0$ cannot guarantee the exact recoverability of our model.
\end{example}

Imposing incoherence conditions on $\widetilde{L}=L_0+\mathcal{P}_{\mathcal{I}_0}H_L$ and $\widehat{L}=L_0+\mathcal{P}_{\mathcal{U}_0}H_L$ is not so surprising: there might be multiple solutions for the optimization model, and the low-rankness/sparseness decompositions of $M$ are non-unique (depending on which solution we are considering). Since $\widetilde{L}+\widetilde{S}$ and $\widehat{L}+\widehat{S}$ are two eligible decompositions of $M$ related to a fixed optimal solution pair, it is natural to consider imposing incoherence on them. Specifically, we first assume incoherence conditions \eqref{equ: incoherence 1}, \eqref{equ: another incoherence}, and \eqref{equ: incoherence 2} on $\widetilde{L}=L_0+\mathcal{P}_{\mathcal{I}_0}H_L$. Note that these conditions guarantee that matrix $\widetilde{L}$ cannot be sparse, so we can resolve the identifiability issue for the decomposition $M=\widetilde{L}+\widetilde{S}$ and hopefully recover the index $\mathcal{I}_0$. After that, the ambiguity between the low rankness and the row sparseness is not an issue any more, i.e., even for row-sparse underlying matrix we can still expect to recover its column space. Here is an example to illustrate this: suppose the low rank matrix is $e_1\mathbf{1}^*$ which has ones in the first rows and zeros elsewhere, and we have known some of the columns are corrupted by noise. Remove the outlier columns. Even we cannot fully observe the remaining entries, we can still expect to recover the column space $\mbox{Range}(e_1)$ since the information for the range space is sufficient to us. Therefore, we only need to impose condition \eqref{equ: incoherence 1} on $\widehat{L}=L_0+\mathcal{P}_{\mathcal{U}_0}H_L$, which asserts that $\widehat{L}$ cannot be column-sparse.

\subsubsection{Ambiguity Conditions on Column-Sparse Term}
Analogously, the column-sparse term $\widehat{S}$ has the identification issue as well. Suppose that $\widehat{S}$ is a rank-1 matrix such that a constant fraction of the columns are zeros. This matrix is both low-rank and column-sparse, which cannot be correctly identified. To avoid this case, one needs the isotropic assumption~\cite{vershynin2010introduction}, or the following ambiguity condition, on the column-sparse term $\widehat{S}$, which is introduced by \cite{Zhang2015AAAI}:
\begin{equation}
\label{equ: incoherence on S_0}
\|\mathcal{B}(\widehat{S})\|\leq \mu',
\end{equation}
where $\mu'$ can be any numerical constant. Here the isotropic assumption asserts that the covariance of the noise matrix is the identity. In fact, many noise models satisfy this assumption, e.g., i.i.d. Gaussian noise. So the normalized noise vector would uniformly distribute on the surface of a unit sphere centered at the origin, thus they cannot be in a low-dimensional subspace --- in other words, not low-rank. Similarly, the ambiguity condition was proposed for the same purpose~\cite{Zhang2015AAAI}. Geometrically, the spectral norm stands for the length of the first principal direction (we use operator $\mathcal{B}$ to remove the scaling factor). So condition \eqref{equ: incoherence on S_0} asserts that the energy for each principal direction does not differ too much, namely, the data distribute around a ball (see Figure \ref{figure: ambiguity conditions}), and \eqref{equ: incoherence on S_0} holds once the directions of non-zero columns of $\widehat{S}$ scatter sufficiently randomly. Note that the isotropic assumption implies our ambiguity condition: if the columns of $\widehat{S}$ are isotropic, $\|\mathcal{B}(\widehat{S})\|$ would be a constant even though the number of column support of $\widehat{S}$ is comparable to $n$. Thus our ambiguity condition \eqref{equ: incoherence on S_0} is feasible. No matter what number of non-zero columns of $\widehat{S}$ is, the assumption guarantees matrix $\widehat{S}$ not to be low-rank.
\begin{figure}
\centering
\includegraphics[width=0.4\textwidth]{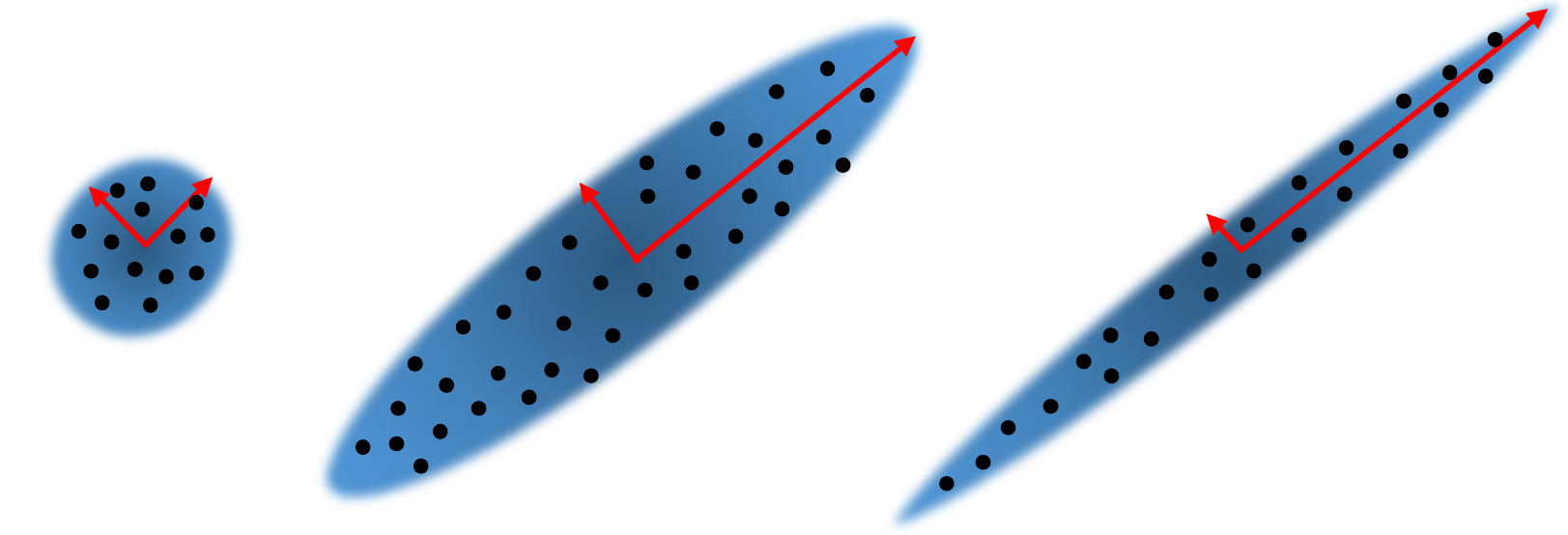}
\caption{Illustration of the ambiguity condition. From the left to the right, the $\mu'$ increases and the data tend to lie in a low-dimensional subspace.}
\label{figure: ambiguity conditions}
\end{figure}

\subsubsection{Probability Model}
Our main results assume that the column
support of $S_0$ and the entry support of measured set $\mathcal{K}_{obs}$ obey i.i.d. Bernoulli distribution with parameter $a$ and parameter $p_0$, respectively. Such assumptions are mild because we have no further information on the positions of outlier and measurement. More specifically, we assume that $[S_0]_{:j}=[\delta_0]_j[Z_0]_{:j}$ throughout our proof, where $[\delta_0]_j\sim \mbox{Ber}(p)$ determines the outlier positions and $[Z_0]_{:j}$ determines the outlier values. If an event holds with a probability at least
$1-\Theta(n^{-10})$, we say that the event happens with an overwhelming probability.

\subsubsection{Other Assumptions}
Obviously, to guarantee the exact recovery of $\mbox{Range}(L_0)$, the noiseless samples $\mathcal{P}_{\mathcal{I}_0^\perp}L_0$ should span the same space as that of $\mbox{Range}(L_0)$, i.e., $\mbox{Range}(L_0)=\mbox{Range}(\mathcal{P}_{\mathcal{I}_0^\perp}L_0)$. Otherwise, only a subspace of $\mbox{Range}(L_0)$ can be recovered, because the noises may be arbitrarily severe. So without loss of generality, we assume $L_0=\mathcal{P}_{\mathcal{I}_0^\perp}L_0$, as work \cite{Xu:CSPCP,ICML2011Chen_469} did. Moreover, the noises should be identifiable, namely, they cannot lie in the ground truth $\mbox{Range}(L_0)$.

\subsection{Summary of Main Notations}
In this paper, matrice are denoted by capital symbols. For matrix $M$, we represent $M_{:j}$ or $M^{(j)}$ as the $j$th column of $M$. We denote by $M_{ij}$ the entry at the $i$th row, $j$th column of the matrix. For matrix operators, $M^*$ and $M^\dag$ represent the conjugate transpose and the Moore-Penrose pseudo-inverse of $M$, respectively, and $|M|$ stands for the matrix whose $(i,j)$-th entry is $|M_{ij}|$.

Several norms appear in this paper, both for vector and for matrix. The only vector norm we use is $\|\cdot\|_2$, which stands for the Euclidean norm or the vector $\ell_2$ norm. For matrix norm, we denote by $\|\cdot\|_*$ the nuclear norm, which stands for the sum of singular values. The matrix norm analogous to the vector $\ell_2$ norm is the Frobenious norm, represented by $\|\cdot\|_F$. The pseudo-norms, $\|\cdot\|_0$ and $\|\cdot\|_{2,0}$, denote the number of non-zero entries and non-zero columns of a matrix, respectively; They are not real norms because the absolute homogeneity does not hold. The convex surrogates of $\|\cdot\|_0$ and $\|\cdot\|_{2,0}$ are matrix $\ell_1$ and $\ell_{2,1}$ norms, with definitions $\|M\|_1=\sum_{ij}|M_{ij}|$ and $\|M\|_{2,1}=\sqrt{\sum_j\|M_{:j}\|_2}$, respectively. The dual norms of matrix $\ell_1$ and $\ell_{2,1}$ norms are $\ell_\infty$ and $\ell_{2,\infty}$ norms, represented by $\|M\|_\infty=\max_{ij}|M_{ij}|$ and $\|M\|_{2,\infty}=\max_j\|M_{:j}\|_2$. We also denote the operator norm of operator $\mathcal{P}$ as $\|\mathcal{P}\|=\sup_{\|M\|_F=1}\|\mathcal{P}M\|_F$.

Our analysis involves linear spaces as well. For example, $\mathcal{I}$ and $\mbox{Supp}(L)$ (similarly define $\mathcal{I}_0$ for $L_0$, we will not restate that for the following notations) denotes the column support of matrix $L$. Without confusion, it forms a linear subspace. We use $\Omega$ to represent the element support of a matrix, as well as the corresponding linear subspace. The column space of a matrix is written as script $\mathcal{U}$, while the row space is written as script $\mathcal{V}$ or $\mbox{Row}(L)$. For any space $\mathcal{X}$, $\mathcal{X}^\perp$ stands for the orthogonal complement of space $\mathcal{X}$.

We also discuss some special matrices and spaces in our analysis. For example, $(L_0,S_0)$ denotes the ground truth. We represent $(L^*,S^*)=(L_0+H_L,S_0-H_S)$ as the optimal solutions of our model, where $H_L$ and $H_S$ guarantee the feasibility of the solution. We are especially interested in expanding the objective function at some particular points: For the exact recovery of the column support, we focus on $(\widetilde{L},\widetilde{S})=(L_0+\mathcal{P}_{\mathcal{I}_0}H_L,S_0-\mathcal{P}_{\mathcal{I}_0}H_S)$; for the the exact recovery of the column support, we focus on $(\widehat{L},\widehat{S})=(L_0+\mathcal{P}_{\mathcal{I}_0}\mathcal{P}_{\mathcal{U}_0}H_L,S_0-\mathcal{P}_{\Omega_{obs}}\mathcal{P}_{\mathcal{I}_0}\mathcal{P}_{\mathcal{U}_0}H_L)$. Another matrix we are interested in is $\mathcal{B}(S)$, which consists of normalized non-zero columns of $S$ and belongs to the subdifferential of $\ell_{2,1}$ norm. Similarly, the space $\mathcal{T}=\{UX^*+YV^*,\forall X,Y\in\mathbb{R}^{n\times r}\}$ is highly related to the subgradient of the nuclear norm. Namely, the subgradient of nuclear norm can be written in closed form as a term in $\mathcal{T}$ plus a term in $\mathcal{T}^\perp$. The projection operator to space $\mathcal{T}^\perp$ is denoted by $\mathcal{P}_{\mathcal{T}^\perp}$, which equals $\mathcal{P}_{\mathcal{U}^\perp}\mathcal{P}_{\mathcal{V}^\perp}$.

Table \ref{table: notations} summarizes the main notations used in this paper.
\begin{table*}
\caption{Summary of main notations used in the paper.}
\label{table: notations}
\begin{center}
\begin{tabular}{c|l||c|l}
\hline
Notations & Meanings & Notations & Meanings\\
\hline
$\mathcal{I}$ & Column Support. & $\Omega$ & Element Support.\\
$\mathcal{I}_0$ & $\mathcal{I}_0\sim\mbox{Ber}(a)$. & $\Omega_{obs}$ & $\Omega_{obs}\sim\mbox{Ber}(p_0)$.\\
$\Gamma$ & $\Gamma=\mathcal{I}_0^\perp\cap\Omega_{obs}$ & $\Pi$ & $\Pi=\mathcal{I}_0\cap\Omega_{obs}$.\\
$m$, $n$ & Size of the data matrix $M$. & $n_{(1)}$, $n_{(2)}$ & $n_{(1)}=\max\{m,n\}$, $n_{(2)}=\min\{m,n\}$.\\
$\Theta(n)$ & Grows in the same order of $n$. & $O(n)$ & Grows equal to or less than the order of $n$. \\
$\otimes$ & Tensor product. & $e_i$ & Vector whose $i$th entry is $1$ and others are $0$s.\\
Capital & A matrix. & $I$, $\boldsymbol{0}$, $\boldsymbol{1}$ & The identity matrix, all-zero matrix, and all-one vector.\\
$M_{:j}$ or $M^{(j)}$ & The $j$th column of matrix $M$. & $M_{ij}$ & The entry at the $i$th row and $j$th column of $M$.\\
$M^*$ & Conjugate transpose of matrix $M$. & $M^\dag$ & Moore-Penrose pseudo-inverse of matrix $M$.\\
$|M|$ & Matrix whose $(i,j)$-th entry is $|M_{ij}|$. & $\|\cdot\|_2$ & $\ell_2$ norm for vector, $\|v\|_2=\sqrt{\sum_{i}v_{i}^2}$.\\
$\|\cdot\|_*$ & Nuclear norm, the sum of singular values. & $\|\cdot\|_0$ & $\ell_0$ norm, number of non-zero entries.\\
$\|\cdot\|_{2,0}$ & $\ell_{2,0}$ norm, number of nonzero columns. & $\|\cdot\|_1$ & $\ell_1$ norm, $\|M\|_1=\sum_{i,j}|M_{ij}|$. \\
$\|\cdot\|_{2,1}$ & $\ell_{2,1}$ norm, $\|M\|_{2,1}=\sqrt{\sum_{j}\|M_{:j}\|_2}$. & $\|\cdot\|_{2,\infty}$ & $\ell_{2,\infty}$ norm, $\|M\|_{2,\infty}=\max_j\|M_{:j}\|_2$.\\
$\|\cdot\|_F$ & Frobenious norm, $\|M\|_F=\sqrt{\sum_{i,j}M_{ij}^2}$. & $\|\cdot\|_\infty$ & Infinity norm, $\|M\|_\infty=\max_{ij}|M_{ij}|$.\\
$\|\mathcal{P}\|$ & (Matrix) operator norm. & $L_0, S_0$ & Ground truth.\\
$L^*, S^*$ & Optimal solutions, $L^*=L_0+H_L, S^*=S_0-H_S$. & $\widetilde{L}, \widetilde{S}$ & $\widetilde{L}=L_0+\mathcal{P}_{\mathcal{I}_0}H_L$, $\widetilde{S}=S_0-\mathcal{P}_{\mathcal{I}_0}H_S$.\\
$\widehat{L}, \widehat{S}$ & $\widehat{L}=L_0+\mathcal{P}_{\mathcal{U}_0}H_L$, $\widehat{S}=S_0-\mathcal{P}_{\Omega_{obs}}\mathcal{P}_{\mathcal{U}_0}H_L$ & $\widehat{U}$, $\widehat{V}$ & Left and right singular vectors of $\widehat{L}$.\\
$\mathcal{U}_0$, $\mathcal{\widehat{U}}$, $\mathcal{U}^*$ & Column space of $L_0$, $\widehat{L}$, $L^*$. & $\mathcal{V}_0$, $\mathcal{\widehat{V}}$, $\mathcal{V}^*$ & Row space of $L_0$, $\widehat{L}$, $L^*$.\\
$\mathcal{\widehat{T}}$ & Space $\mathcal{\widehat{T}}=\{\widehat{U}X^*+Y\widehat{V}^*,\forall X,Y\in\mathbb{R}^{n\times r}\}$. & $\mathcal{X}^\perp$ & Orthogonal complement of the space $\mathcal{X}$.\\
$\mathcal{P}_\mathcal{\widehat{U}}$, $\mathcal{P}_\mathcal{\widehat{V}}$ & $\mathcal{P}_\mathcal{\widehat{U}}M=\widehat{U}\widehat{U}^*M$, $\mathcal{P}_\mathcal{\widehat{V}}M=M\widehat{V}\widehat{V}^*$. & $\mathcal{P}_\mathcal{\widehat{T}^\perp}$ & $\mathcal{P}_\mathcal{\widehat{T}^\perp}M=\mathcal{P}_\mathcal{\widehat{U}^\perp}\mathcal{P}_\mathcal{\widehat{V}^\perp}M$.\\
$\mathcal{I}_0$, $\widehat{\mathcal{I}}$, $\mathcal{I}^*$ & Index of outliers of $S_0$, $\widehat{S}$, $S^*$. & $|\mathcal{I}_0|$ & Outliers number of $S_0$.\\
$X\in\mathcal{I}$ & The column support of $X$ is a subset of $\mathcal{I}$. & $\mathcal{B}(\widehat{S})$ & $\mathcal{B}(\widehat{S})=\{H:
\mathcal{P}_{\mathcal{\widehat{I}}^\perp}(H)=0;
H_{:j}=\frac{\widehat{S}_{:j}}{\|\widehat{S}_{:j}\|_2},
j\in\mathcal{\widehat{I}}$\}.\\
$\sim\mbox{Ber}(p)$ & Obeys Bernoulli distribution with parameter
$p$. & $\mathcal{N}(a,b^2)$ & Gaussian distribution (mean $a$
and variance $b^2$).\\
Row$(M)$ & Row space of matrix $M$. & Supp$(M)$ & Column support of matrix $M$.\\
$\sigma_i(M)$ & The $i$th singular value of matrix $M$. & $\lambda_i(M)$ & The $i$th eigenvalue of matrix $M$.\\
$\omega_{ij}$ & General basis. & $\mathcal{R}$ & $\mathcal{R}(\cdot)=\sum_{ij\in\mathcal{K}_{obs}}\langle\cdot,\omega_{ij}\rangle\omega_{ij}$.\\
\hline
\end{tabular}
\end{center}
\end{table*}

\section{Exact Recoverability of the Model}
\label{section: exact recoverability of the model}

Our main results in this paper show that, surprisingly, model \eqref{equ: robust MC wrt any basis} is able to exactly recover the range space of $L_0$ and identify the column support of $S_0$ with a closed-form regularization parameter, even when only a small number of expansion coefficients are measured w.r.t. general basis and a constant fraction of columns are arbitrarily corrupted. Our theorem is as follows:
\begin{theorem}[Exact Recoverability Under Bernoulli Sampling]
\label{theorem: exact recovery under Bernoulli sampling} Any solution $(L^*,S^*)$ to the extended robust MC \eqref{equ: robust MC wrt any basis} with $\lambda=1/\sqrt{\log n}$ exactly recovers the column space of $L_0$ and the column support of $S_0$ with a probability at least $1-cn^{-10}$, if the column support
$\mathcal{I}_0$ of $S_0$ subjects to i.i.d. $\mbox{Ber}(a)$, the support $\mathcal{K}_{obs}$ subjects to i.i.d. $\mbox{Ber}(p_0)$, and
\begin{equation}
\label{equ: upper bounds under Bernoulli sampling}
\mbox{rank}(L_0)\le\rho_r\frac{n_{(2)}}{\mu(\log n_{(1)})^3},\ \ a\le \rho_a\frac{n_{(2)}}{\mu n(\log n_{(1)})^3},\ \ p_0\ge\rho_p,
\end{equation}
where $c$, $\rho_r<1$, $\rho_a<1$, and $\rho_p<1$ are all constants independent of each other, and $\mu$ is the incoherence parameter in \eqref{equ: incoherence 1}, \eqref{equ: another incoherence}, and \eqref{equ: incoherence 2}.
\end{theorem}

\begin{remark}
According to \cite{Candes}, a recovery result under the Bernoulli model with parameter $p$ automatically
implies a corresponding result for the uniform model with parameter $\Theta{(np)}$ at an overwhelming probability. So conditions \eqref{equ: upper bounds under Bernoulli sampling} are equivalent to
\begin{equation}
\label{equ: upper bounds under uniform sampling}
\mbox{rank}(L_0)\le\frac{\rho_rn_{(2)}}{\mu(\log n_{(1)})^3},\ \ s\le\frac{\rho_s'n_{(2)}}{\mu(\log n_{(1)})^3},\ \ k\ge\rho_p'n_{(1)}n_{(2)},
\end{equation}
where the column support $\mathcal{I}_0$ of $S_0$ is uniformly distributed among all sets of cardinality $s$, the support $\mathcal{K}_{obs}$ is uniformly distributed among all sets of cardinality $k$, and $\rho_r$, $\rho_s'$, and $\rho_p'$ are numerical constants.
\end{remark}

%\begin{lemma}
%A recovery result under the Bernoulli model with parameter $p$ automatically
%implies a corresponding result for the uniform model with parameter $C_0np$ at an overwhelming probability, provided that $p\ge c\log n/n$, where $n$ is the size of the samples, and $c$ and $C_0$ are numerical constants.
%\end{lemma}
%
%
%
%
%\begin{theorem}[Exact Recoverability Under Uniform Sampling without Replacement]
%\label{theorem: exact recovery under uniform sampling} Suppose $\mbox{Range}(L_0)=\mbox{Range}(\mathcal{P}_{\mathcal{I}_0^\perp}L_0)$, and $[\mathcal{R}(S_0)]_{:j}\not\in\mbox{Range}(\mathcal{R}_j'(L_0))$ for $\forall j\in\mathcal{I}_0$. Then any solution $(L^*,S^*)=(L_0+H_L,S_0'-H_S)$ to modified Outlier Pursuit \eqref{equ: Column
%Sparse RPCA} with $\lambda=1/\sqrt{\log n}$ exactly recovers the column space of $L_0$ and the column support of $S_0$ with a probability at least $1-cn^{-10}$, if the column support
%$\mathcal{I}_0$ of $S_0$ is uniformly distributed among all sets
%of cardinality $s$, the support $\mathcal{K}_{obs}$ is uniformly distributed among all sets of cardinality $k$ and
%\begin{equation}
%\label{equ: upper bounds under uniform sampling}
%\mbox{rank}(L_0)\le\rho_r\frac{n_{(2)}}{\mu(\log n_{(1)})^3},\ \ \ \ \ s\le\rho_s\frac{n_{(2)}}{\mu(\log n_{(1)})^3},\ \ \ \ \
%\mbox{and}\ \ \ \ \ k\ge\rho_kn_{(1)}n_{(2)},
%\end{equation}
%where $c$, $\rho_r$, $\rho_s$, and $\rho_k$ are all constants.
%\end{theorem}

\subsection{Comparison to Previous Results}
In the traditional low-rank MC problem, one seeks to complete a low-rank matrix from only a few measurements \emph{without corruptions}. Recently, it has been shown that a constant fraction of the entries are allowed to be missing, even if the rank of intrinsic matrix is as high as $O(n/\log^2 n)$. So compared with the result, our bound in Theorem \ref{theorem: exact recovery under Bernoulli sampling} is tight up to a polylog factor. Note that the polylog gap comes from the consideration of arbitrary corruptions in our analysis. When $a=0$, our theorem partially recovers the results of \cite{Gross2011recovering}.

In the traditional low-rank matrix recovery problem, one tries to recover a low-rank matrix, or the range space of matrix, from \emph{fully observed} corrupted data.
To this end, our previous work~\cite{Zhang2015AAAI} demonstrated that a constant fraction of the columns can be corrupted, even if the rank of intrinsic matrix is as high as $O(n/\log n)$. Compared with the result, our bound in Theorem \ref{theorem: exact recovery under Bernoulli sampling} is tight up to a polylog factor as well, where the polylog gap comes from the consideration of missing values in our analysis. When $p_0=1$, our theorem partially recovers the results of \cite{Zhang2015AAAI}.

Probably the only low-rank model that can simultaneously complete the missing values, recover the ground truth subspace, and detect the corrupted samples is robust MC~\cite{ICML2011Chen_469}. As a corollary, Chen et al.~\cite{ICML2011Chen_469} showed that a constant fraction of columns and entries can be corrupted and missing, respectively, if the rank of $L_0$ is of order $O(1)$. Compared with this, though with stronger incoherence (ambiguity) conditions, our work extends the working range of robust MC model to the rank of order $O(n/\log^3 n)$. Moreover, our results consider a set of more general basis, i.e., when $\omega_{ij}=e_ie_j^T$, our theorem partially recovers the results of \cite{ICML2011Chen_469}.

Wright et al.~\cite{wright2013compressive} produced a certificate of optimality for $(L_0,S_0)$ for the Compressive Principal Component Pursuit, given that $(L_0,S_0)$ is the optimal solution for Principal Component Pursuit. There are significant differences between their work and ours: 1. Their analysis assumed that certain entries are corrupted by noise, while our paper assumes that some whole columns are noisy. In some sense, theoretical analysis on column noise is more difficult than that on Principal Component Pursuit~\cite{Zhang2015AAAI}. The most distinct difference
is that we cannot expect our model to exactly recover $L_0$
and $S_0$. Rather, only the \emph{column space} of $L_0$ and the
\emph{column support} of $S_0$ can be exactly
recovered~\cite{ICML2011Chen_469,Xu:CSPCP}. 2. Wright et al.'s analysis is based on the assumption that $(L_0,S_0)$ can be recovered by Principal Component Pursuit, while our analysis is independent of this requirement.

\section{Complete Proofs of Theorem \ref{theorem: exact recovery under Bernoulli sampling}}
\label{section: proof}

Theorem \ref{theorem: exact recovery under Bernoulli sampling} shows the exact recoverability of our extended robust MC model w.r.t. general basis. This section is devoted to proving this result.

\subsection{Proof Sketch}
We argue that it is not very straightforward to apply the existing proofs on Robust PCA/Matrix Completion to the case of general basis, since these proofs essentially require the observed entries and the outliers to be represented under the same basis~\cite{Candes}. To resolve the issue, generally speaking, we novelly divide the proof of Theorem \ref{theorem: exact recovery under Bernoulli sampling} into two parts: The exact recoverability of column support and the exact recoverability of column space. We are able to attack the two problems separately thanks to the idea of expanding the objective function at the well-designed points, i.e., $(\widetilde{L},\widetilde{S})$ for the recovery of column support and $(\widehat{L},\widehat{S})$ for the recovery of column space, respectively (see Sections \ref{subsubsection: dual conditions 1} and \ref{subsubsection: dual conditions 2} for details). This technique enables us to decouple the randomization of $\mathcal{I}_0$ and $\Omega_{obs}$, and so construct the dual variables easily by standard tools like the least squares and golfing scheme. We notice that our framework is general. It not only can be applied to the proof for easier model like Outlier Pursuit~\cite{Zhang2015AAAI} (though we will sacrifice a small polylog factor for the probability $a$ of outliers), but can also hopefully simplify the proof for model with more complicated formulation, e.g., decomposing the data matrix $M$ into more than two structural components~\cite{wright2013compressive}. That is roughly the high-level intuition why we can handle the general basis and improve over the previous work in this paper.

Specifically, for the exact recoverability of column support, we expand the objective function at $(\widetilde{L},\widetilde{S})$ to establish our first class of dual conditions. Though it is standard to construct dual variables by golfing scheme, many lemmas need to be generalized in the standard setting because of the existence of both $\mathcal{I}_0$ and $\Omega_{obs}$. All the preliminary work is done in Appendix \ref{subsection: preliminary}. When $p_0=1$ or $a=0$, we claim that our lemmas return to the ones in~\cite{Candes,ICML2011Chen_469}, thus being more general. The idea behind the proofs is to fix $\mathcal{I}_0$ first and use the randomized argument for $\Omega_{obs}$ to have a one-step result, and then allow $\mathcal{I}_0$ to be randomized to get our desired lemmas.

For the exact recoverability of column support, similarly, we expand the objective function at $(\widehat{L},\widehat{S})$ to establish our second class of dual conditions. We construct the dual variables by the least squares, and prove the correctness of our construction by using generalized lemmas as well. To this end, we also utilize the ambiguity condition, which guarantees that the outlier matrix cannot be low-rank. This enables us to improve the upper bound for the rankness of the ground truth matrix from $O(1)$ to our $O(n/\log^3 n)$.

In summary, our proof proceeds in two parallel lines. The steps are as follows.
\begin{itemize}
\item
We prove the exact recoverability of column support:
\begin{itemize}
\item
Section \ref{subsubsection: dual conditions 1} proves the correctness of dual condition, as shown in Lemma \ref{lemma: dual conditions for exact column support}. In particular, in the proof we focus on the subgradient of objective function at $(\widetilde{L},\widetilde{S})$.
\item
Section \ref{subsubsection: certification by golfing scheme} shows the construction of dual variables \eqref{equ: W^L}, and Section \ref{subsubsection: proofs of dual conditions 1} proves its correction in Lemma \ref{lemma: correction for column support}.
\end{itemize}
\item
We then prove the exact recoverability of column space:
\begin{itemize}
\item
Section \ref{subsubsection: dual conditions 2} proves the correctness of dual condition, as shown in Lemma \ref{lemma: dual conditions for exact column space}. In particular, in the proof we focus on the subgradient of objective function at $(\widehat{L},\widehat{S})$.
\item
Section \ref{subsubsection: certification by least squares} shows the construction of dual variables \eqref{equ: W}, and Section \ref{subsubsection: proofs of dual conditions 2} proves its correction in Lemma \ref{lemma: correction for column space}.
\end{itemize}
\end{itemize}

\subsection{Exact Recovery of Column Support}
\label{section: exact recovery of column support}
\subsubsection{Dual Conditions}
\label{subsubsection: dual conditions 1}
We first establish dual conditions for the exact recovery of the column support. The following lemma shows that once we can construct dual variables satisfying certain conditions (a.k.a. dual conditions), the column support of the outliers can be exactly recovered with a high probability by solving our robust MC model \eqref{equ: robust MC wrt any basis}. Basically, the proof is to find conditions which implies that $0$ belongs to the subdifferential of the objective function at the desired low-rank and column-sparse solution.
\begin{lemma}
\label{lemma: dual conditions for exact column support}
Let $(L^*,S^*)=(L_0+H_L,S_0'-H_S)$ be any solution to the extended robust MC \eqref{equ: robust MC wrt any basis}, $\widetilde{L}=L_0+\mathcal{P}_{\mathcal{I}_0}H_L$, and $\widetilde{S}=S_0-\mathcal{P}_{\mathcal{I}_0}H_S$. Assume that $\|\mathcal{P}_{\Omega_{obs}^\perp}\mathcal{P}_\mathcal{\widetilde{T}}\|\le1/2$, and
\begin{equation*}
\widetilde{U}\widetilde{V}^*+\widetilde{W}=\lambda(\widetilde{F}+\mathcal{P}_{\Omega_{obs}^\perp}\widetilde{D}),
\end{equation*}
where $\mathcal{P}_{\mathcal{\widetilde{T}}}\widetilde{W}=0$, $\|\widetilde{W}\|\le1/2$, $\mathcal{P}_{\Omega_{obs}^\perp}\widetilde{F}=0$, $\|\widetilde{F}\|_{2,\infty}\le1/2$, and $\|\mathcal{P}_{\Omega_{obs}^\perp}\widetilde{D}\|_F\le1/4$. Then $S^*$ exactly recovers the column support of $S_0$, i.e., $H_L\in\mathcal{I}_0$.
\end{lemma}

\begin{proof}
We first recall that the subgradients of nuclear norm
and $\ell_{2,1}$ norm are as follows:
\begin{equation*}
\partial_{\widetilde{L}}\|\widetilde{L}\|_*=\{\widetilde{U}\widetilde{V}^*+\widetilde{Q}: \widetilde{Q}\in\mathcal{\widetilde{T}^\perp},\|\widetilde{Q}\|\le1\},
\end{equation*}
\begin{equation*}
\partial_{\widetilde{S}}\|\widetilde{S}\|_{2,1}=\{\mathcal{B}(\widetilde{S})+\widetilde{E}: \widetilde{E}\in\mathcal{\widetilde{I}^\perp},\|\widetilde{E}\|_{2,\infty}\le1\}.
\end{equation*}
According to Lemma \ref{lemma: inside Omega} and the feasibility of $(L^*,S^*)$, $\mathcal{P}_{\Omega_{obs}}H_L=\mathcal{P}_{\Omega_{obs}}H_S=H_S$. Let $\widetilde{S}=S_0'-H_S+\mathcal{P}_{\Omega_{obs}}\mathcal{P}_{\mathcal{I}_0^\perp}H_L=S_0'-\mathcal{P}_{\Omega_{obs}}\mathcal{P}_{\mathcal{I}_0}H_L\in\mathcal{I}_0$. Thus the pair $(\widetilde{L},\widetilde{S})$ is feasible to problem \eqref{equ: robust MC wrt any basis}. Then we have
\begin{equation*}
\begin{split}
&\ \ \ \ \ \|L_0+H_L\|_*+\lambda\|S_0'-H_S\|_{2,1}\\
&\ge\|\widetilde{L}\|_*+\lambda\|\widetilde{S}\|_{2,1}+\langle\widetilde{U}\widetilde{V}^*+\widetilde{Q},\mathcal{P}_{\mathcal{I}_0^\perp}H_L\rangle-\lambda\langle\mathcal{B}(\widetilde{S})+\widetilde{E},\mathcal{P}_{\Omega_{obs}}\mathcal{P}_{\mathcal{I}_0^\perp}H_L\rangle.
\end{split}
\end{equation*}
Now adopt $\widetilde{Q}$ such that $\langle\widetilde{Q},\mathcal{P}_{\mathcal{I}_0^\perp}H_L\rangle=\|\mathcal{P}_{\mathcal{\widetilde{T}}^\perp}\mathcal{P}_{\mathcal{I}_0^\perp}H_L\|_*$ and $\langle\widetilde{E},\mathcal{P}_{\Omega_{obs}}\mathcal{P}_{\mathcal{I}_0^\perp}H_L\rangle=\|\mathcal{P}_{\Omega_{obs}}\mathcal{P}_{\mathcal{I}_0^\perp}H_L\|_{2,1}$\footnote{By the duality between the nuclear norm and the operator norm, there exists a $Q$ such that $\langle Q, \mathcal{P}_{\widetilde{\mathcal{T}}^\perp}\mathcal{P}_{\mathcal{I}_0^\perp}H\rangle=\|\mathcal{P}_{\widetilde{\mathcal{T}}^\perp}\mathcal{P}_{\mathcal{I}_0^\perp}H\|_*$ and $\|Q\|\le1$. Thus we take $\widetilde{Q}=\mathcal{P}_{\widetilde{\mathcal{T}}^\perp}Q\in\mathcal{\widetilde{T}^\perp}$. It holds similarly for $\widetilde{E}$.}, and note that $\langle\mathcal{B}(\widetilde{S}),\mathcal{P}_{\Omega_{obs}}\mathcal{P}_{\mathcal{I}_0^\perp}H_L\rangle=0$. So we have
\begin{equation*}
\begin{split}
&\ \ \ \ \ \|L_0+H_L\|_*+\lambda\|S_0'-H_S\|_{2,1}\\
&\ge\|\widetilde{L}\|_*+\lambda\|\widetilde{S}\|_{2,1}+\|\mathcal{P}_{\mathcal{\widetilde{T}}^\perp}\mathcal{P}_{\mathcal{I}_0^\perp}H_L\|_*+\lambda\|\mathcal{P}_{\Omega_{obs}}\mathcal{P}_{\mathcal{I}_0^\perp}H_L\|_{2,1}+\langle\widetilde{U}\widetilde{V}^*,\mathcal{P}_{\mathcal{I}_0^\perp}H_L\rangle.
\end{split}
\end{equation*}
Notice that
\begin{equation*}
\begin{split}
&\ \ \ \ \ |\langle\widetilde{U}\widetilde{V}^*,\mathcal{P}_{\mathcal{I}_0^\perp}H_L\rangle|\\
&=|\langle\widetilde{W}-\lambda\widetilde{F}-\lambda\mathcal{P}_{\Omega_{obs}^\perp}\widetilde{D},\mathcal{P}_{\mathcal{I}_0^\perp}H_L\rangle|\\
&\le\frac{1}{2}\|\mathcal{P}_{\mathcal{\widetilde{T}^\perp}}\mathcal{P}_{\mathcal{I}_0^\perp}H_L\|_*+\frac{\lambda}{2}\|\mathcal{P}_{\Omega_{obs}}\mathcal{P}_{\mathcal{I}_0^\perp}H_L\|_{2,1}+\frac{\lambda}{4}\|\mathcal{P}_{\Omega_{obs}^\perp}\mathcal{P}_{\mathcal{I}_0^\perp}H_L\|_F.
\end{split}
\end{equation*}
So we have
\begin{equation*}
\begin{split}
&\ \ \ \ \ \|L_0+H_L\|_*+\lambda\|S_0'-H_S\|_{2,1}\\
&\ge\|\widetilde{L}\|_*\hspace{-0.05cm}+\hspace{-0.05cm}\lambda\|\widetilde{S}\|_{2,1}\hspace{-0.05cm}+\hspace{-0.05cm}\frac{1}{2}\|\mathcal{P}_{\mathcal{\widetilde{T}}^\perp}\mathcal{P}_{\mathcal{I}_0^\perp}H_L\|_*\hspace{-0.05cm}+\hspace{-0.05cm}\frac{\lambda}{2}\|\mathcal{P}_{\Omega_{obs}}\mathcal{P}_{\mathcal{I}_0^\perp}H_L\|_{2,1}-\frac{\lambda}{4}\|\mathcal{P}_{\Omega_{obs}^\perp}\mathcal{P}_{\mathcal{I}_0^\perp}H_L\|_F.
\end{split}
\end{equation*}
Also, note that
\begin{equation*}
\begin{split}
&\ \ \ \|\mathcal{P}_{\Omega_{obs}^\perp}\mathcal{P}_{\mathcal{I}_0^\perp}H_L\|_F\\&\le\|\mathcal{P}_{\Omega_{obs}^\perp}\mathcal{P}_\mathcal{\widetilde{T}^\perp}\mathcal{P}_{\mathcal{I}_0^\perp}H_L\|_F+\|\mathcal{P}_{\Omega_{obs}^\perp}\mathcal{P}_\mathcal{\widetilde{T}}\mathcal{P}_{\mathcal{I}_0^\perp}H_L\|_F\\
&\le\|\mathcal{P}_\mathcal{\widetilde{T}^\perp}\mathcal{P}_{\mathcal{I}_0^\perp}H_L\|_F+\frac{1}{2}\|\mathcal{P}_{\mathcal{I}_0^\perp}H_L\|_F\\
&\le\|\mathcal{P}_\mathcal{\widetilde{T}^\perp}\mathcal{P}_{\mathcal{I}_0^\perp}H_L\|_F\hspace{-0.1cm}+\hspace{-0.1cm}\frac{1}{2}\|\mathcal{P}_{\Omega_{obs}}\mathcal{P}_{\mathcal{I}_0^\perp}H_L\|_F\hspace{-0.1cm}+\hspace{-0.1cm}\frac{1}{2}\|\mathcal{P}_{\Omega_{obs}^\perp}\mathcal{P}_{\mathcal{I}_0^\perp}H_L\|_F.
\end{split}
\end{equation*}
That is
\begin{equation*}
\|\mathcal{P}_{\Omega_{obs}^\perp}\mathcal{P}_{\mathcal{I}_0^\perp}H_L\|_F\le2\|\mathcal{P}_\mathcal{\widetilde{T}^\perp}\mathcal{P}_{\mathcal{I}_0^\perp}H_L\|_F+\|\mathcal{P}_{\Omega_{obs}}\mathcal{P}_{\mathcal{I}_0^\perp}H_L\|_F.
\end{equation*}
Therefore, we have
\begin{equation*}
\begin{split}
&\ \ \ \ \ \|L_0+H_L\|_*+\lambda\|S_0'-H_S\|_{2,1}\\
&\ge\|\widetilde{L}\|_*+\lambda\|\widetilde{S}\|_{2,1}+\frac{1-\lambda}{2}\|\mathcal{P}_{\mathcal{\widetilde{T}}^\perp}\mathcal{P}_{\mathcal{I}_0^\perp}H_L\|_*+\frac{\lambda}{4}\|\mathcal{P}_{\Omega_{obs}}\mathcal{P}_{\mathcal{I}_0^\perp}H_L\|_{2,1}.
\end{split}
\end{equation*}
Since the pair $(L_0+H_L,S_0'-H_S)$ is optimal to problem \eqref{equ: robust MC wrt any basis}, we have
\begin{equation*}
\mathcal{P}_{\mathcal{\widetilde{T}}^\perp}\mathcal{P}_{\mathcal{I}_0^\perp}H_L=0\ \ \mbox{and}\ \ \mathcal{P}_{\Omega_{obs}}\mathcal{P}_{\mathcal{I}_0^\perp}H_L=0,
\end{equation*}
i.e., $\mathcal{P}_{\mathcal{I}_0^\perp}H_L\in\mathcal{\widetilde{T}}\cap\Omega_{obs}^\perp=\{0\}$. So $H_L\in\mathcal{I}_0$.
\end{proof}

By Lemma \ref{lemma: dual conditions for exact column support}, to prove the exact recovery of column support, it suffices to show a dual certificate $\widetilde{W}$ such that
\begin{equation}
\label{equ: dual conditions for exact support}
\begin{cases}
\mbox{(a)}\ \ \widetilde{W}\in \mathcal{\widetilde{T}^\perp},\\
\mbox{(b)}\ \ \|\widetilde{W}\|\le1/2,\\
\mbox{(c)}\ \ \|\mathcal{P}_{\Omega_{obs}^\perp}(\widetilde{U}\widetilde{V}^*+\widetilde{W})\|_F\le\lambda/4,\\
\mbox{(d)}\ \ \|\mathcal{P}_{\Omega_{obs}}(\widetilde{U}\widetilde{V}^*+\widetilde{W})\|_{2,\infty}\le\lambda/2.
\end{cases}
\end{equation}

\subsubsection{Certification by Golfing Scheme}
\label{subsubsection: certification by golfing scheme}
The remainder of the proofs is to construct $\widetilde{W}$ which satisfies
dual conditions \eqref{equ: dual conditions for exact support}.
Before introducing our construction, we assume that $\mathcal{K}_{obs}\sim\mbox{Ber}(p_0)$
(For brevity, we also write it as $\Omega_{obs}\sim\mbox{Ber}(p_0)$), or equivalently
$\Omega_{obs}^\perp\sim\mbox{Ber}(1-p_0)$. Note that
$\Omega_{obs}$ has the same distribution as that of
$\Omega_1\cup\Omega_2\cup...\cup\Omega_{j_0}$,
where each $\Omega_j$ is drawn from $\mbox{Ber}(q)$, $j_0=\lceil\log
n_{(1)}\rceil$,
and $q$ fulfills $1-p_0=(1-q)^{j_0}$ (Note that $q=\Theta(1/\log n_{(1)})$ implies $p_0=\Theta(1)$). We construct $\widetilde{W}$ based on such a
distribution.

To construct $\widetilde{W}$, we use the golfing scheme introduced by
\cite{Gross2011recovering} and \cite{Candes}. Let
$
Z_{j-1}=\mathcal{P}_\mathcal{\widetilde{T}}(\widetilde{U}\widetilde{V}^*-Y_{j-1}).
$
We construct $\widetilde{W}$ by an inductive procedure:
\begin{equation}
\label{equ: W^L}
\begin{split}
Y_j=Y_{j-1}+&q^{-1}\mathcal{P}_{\Omega_j}Z_{j-1}=q^{-1}\sum_{k=1}^j\mathcal{P}_{\Omega_k}Z_{k-1},\\
&\widetilde{W}=\mathcal{P}_{\mathcal{\widetilde{T}}^\perp}Y_{j_0}.
\end{split}
\end{equation}
Also, we have the inductive equation:
\begin{equation}
\label{equ: inductive equation on Z}
Z_j=Z_{j-1}-q^{-1}\mathcal{P}_{\mathcal{\widetilde{T}}}\mathcal{P}_{\Omega_j}Z_{j-1}.
\end{equation}

\subsubsection{Proofs of Dual Conditions}
\label{subsubsection: proofs of dual conditions 1}
We now prove that the dual variables satisfy our dual conditions. The proof basically uses the recursiveness of the dual variables that we construct.
\begin{lemma}
\label{lemma: correction for column support}
Assume that $\Omega_{obs}\sim\mbox{Ber}(p_0)$ and
$j_0=\lceil\log n\rceil$. Then under the other assumptions of
Theorem \ref{theorem: exact recovery under Bernoulli sampling}, $W^L$ given by \eqref{equ:
W^L} obeys the dual conditions \eqref{equ: dual conditions for exact support}.
\end{lemma}

\begin{proof}
By Lemma \ref{lemma: Omega 2 norm}, Lemma \ref{lemma: Omega infty infty norm} and the inductive equation \eqref{equ: inductive equation on Z}, when $q\ge c'\mu r\log
n_{(1)}/\varepsilon^2n_{(2)}$ for some $c'$, the following inequalities hold with an
overwhelming probability:
\begin{equation*}
\|Z_j\|_F<\varepsilon^j\|Z_0\|_F=\varepsilon^j\|\widetilde{U}\widetilde{V}^*\|_F,
\end{equation*}
\begin{equation*}
\max_{ab}|\langle Z_j,\omega_{ab}\rangle|<\varepsilon^j\max_{ab}|\langle Z_0,\omega_{ab}\rangle|=\varepsilon^j\max_{ab}|\langle \widetilde{U}\widetilde{V}^*,\omega_{ab}\rangle|.
\end{equation*}
Now we check the three conditions in \eqref{equ: dual conditions for exact support}.

(a) The construction \eqref{equ: W^L} implies the condition (a) holds.

(b) It holds that
\begin{equation*}
\begin{split} \|\widetilde{W}\|&=\|\mathcal{P}_\mathcal{\widetilde{T}^\perp}Y_{j_0}\|\\&\le\sum_{k=1}^{j_0}\|q^{-1}\mathcal{P}_\mathcal{\widetilde{T}^\perp}\mathcal{P}_{\Omega_k}Z_{k-1}\|\\
&=\sum_{k=1}^{j_0}\|\mathcal{P}_\mathcal{\widetilde{T}^\perp}(q^{-1}\mathcal{P}_{\Omega_k}Z_{k-1}-Z_{k-1})\|\\
&\le\sum_{k=1}^{j_0}\|q^{-1}\mathcal{P}_{\Omega_k}Z_{k-1}-Z_{k-1}\|\\
&\le C_0\sqrt{\frac{n_{(1)}\log n_{(1)}}{q}}\sum_{k=1}^{j_0}\max_{ab}|\langle Z_{k-1},\omega_{ab}\rangle|\\
&\le C_0\frac{1}{1-\varepsilon}\sqrt{\frac{n_{(1)}\log n_{(1)}}{q}}\sqrt{\frac{\mu r}{mn}}\\
&=C_0\frac{1}{1-\varepsilon}\sqrt{\frac{\rho_r}{q(\log n_{(1)})^2}}\\
&\le\frac{1}{4},
\end{split}
\end{equation*}
where the third inequality holds due to Lemma \ref{lemma: Omega 2 infty norm} and the last inequality holds once $q\ge\Theta(1/\log n_{(1)})$.

(c) Notice that $Y_{j_0}\in\Omega_{obs}$, i.e.,
$\mathcal{P}_{\Omega_{obs}^\perp}Y_{j_0}=\mathbf{0}$. Then the
following inequalities follow
\begin{equation}
\label{equ: proof of W^L 1}
\begin{split}
\|\mathcal{P}_{\Omega_{obs}^\perp}(\widetilde{U}\widetilde{V}^*+\widetilde{W})\|_F&=\|\mathcal{P}_{\Omega_{obs}^\perp}(\widetilde{U}\widetilde{V}^*+\mathcal{P}_\mathcal{\widetilde{T}^\perp}Y_{j_0})\|_F\\
&=\|\mathcal{P}_{\Omega_{obs}^\perp}(\widetilde{U}\widetilde{V}^*+Y_{j_0}-\mathcal{P}_\mathcal{\widetilde{T}}Y_{j_0})\|_F\\
&=\|\mathcal{P}_{\Omega_{obs}^\perp}(\widetilde{U}\widetilde{V}^*-\mathcal{P}_\mathcal{\widetilde{T}}Y_{j_0})\|_F\\
&=\|\mathcal{P}_{\Omega_{obs}^\perp}Z_{j_0}\|_F\\
&\le\varepsilon^{j_0}\sqrt{r}\ \ \ \ (j_0=\lceil\log n_{(1)}\rceil\ge\log n_{(1)})\\
&\le\frac{\sqrt{r}}{n_{(1)}}\ \ \ \ (\varepsilon<e^{-1})\\
&\le\frac{\sqrt{\rho_r\mu}}{\sqrt{n_{(1)}}(\log n_{(1)})^{3/2}}\\&\le\frac{\lambda}{4}.
\end{split}
\end{equation}

(d) We first note that $\widetilde{U}\widetilde{V}^*+\widetilde{W}=Z_{j_0}+Y_{j_0}$. It
follows from \eqref{equ: proof of W^L 1} that
\vspace{+0.3cm}
\begin{equation*}
\begin{split}
\|\mathcal{P}_{\Omega_{obs}}Z_{j_0}\|_{2,\infty}&\le\|\mathcal{P}_{\Omega_{obs}}Z_{j_0}\|_F\le\varepsilon^{j_0}\sqrt{r}\le\frac{\lambda}{8}.
\end{split}
\end{equation*}
Moreover, we have
\begin{equation*}
\begin{split}
\|\mathcal{P}_{\Omega_{obs}}Y_{j_0}\|_{2,\infty}&=\|Y_{j_0}\|_{2,\infty}\\&\le\sum_{k=1}^{j_0}q^{-1}\|\mathcal{P}_{\Omega_k}Z_{k-1}\|_{2,\infty}\\
&\le q^{-1}\sqrt{m}\sum_{k=1}^{j_0}\max_{ab}|\langle Z_{k-1},\omega_{ab}\rangle|\\&\le\frac{1}{q}\sqrt{\frac{\mu r}{n}}\sum_{k=1}^{j_0}\varepsilon^{k-1}\\
&\le\frac{1}{q}\sqrt{\frac{\mu r}{n_{(2)}}}\sum_{k=1}^{j_0}\varepsilon^{k-1}\\
&\le c\log n_{(1)}\frac{1}{(\log n_{(1)})^{3/2}}\\
&\le\frac{\lambda}{8},
\end{split}
\end{equation*}
where the fifth inequality holds once $q\ge\Theta(1/\log n_{(1)})$.
Thus $\|\mathcal{P}_{\Omega_{obs}}(\widetilde{U}\widetilde{V}^*+\widetilde{W})\|_{2,\infty}\le\lambda/4$.
\end{proof}

\subsection{Exact Recovery of Column Space}
\label{section: exact recovery of column space}
\subsubsection{Dual Conditions}
\label{subsubsection: dual conditions 2}
We then establish dual conditions for the exact recovery of the column space. The following lemma shows that if we can construct dual variables satisfying certain conditions, the column space of the underlying matrix can be exactly recovered with a high probability by solving model \eqref{equ: robust MC wrt any basis}.
\begin{lemma}[Dual Conditions for Exact Column Space]
\label{lemma: dual conditions for exact column space}
Let $(L^*,S^*)=(L_0+H_L,S_0'-H_S)$ be any solution to the extended robust MC \eqref{equ: robust MC wrt any basis}, $\widehat{L}=L_0+\mathcal{P}_{\mathcal{U}_0}H_L$, and $\widehat{S}=S_0'-\mathcal{P}_{\Omega_{obs}}\mathcal{P}_{\mathcal{U}_0}H_L$. Suppose that $\mathcal{\widehat{V}}\cap\Gamma^\perp=\{0\}$ and
\begin{equation*}
\widehat{W}=\lambda(\mathcal{B}(\widehat{S})+\widehat{F}),
\end{equation*}
where $\widehat{W}\in\mathcal{\widehat{V}^\perp}\cap\Omega_{obs}$, $\|\widehat{W}\|\le1/2$, $\mathcal{P}_{\Gamma^\perp}\widehat{F}=0$, and $\|\widehat{F}\|_{2,\infty}\le1/2$. Then $L^*$ exactly recovers the column support of $L_0$, i.e., $H_L\in\mathcal{U}_0$.
\end{lemma}

\begin{proof}
We first recall that the subgradients of nuclear norm
and $\ell_{2,1}$ norm are as follows:
\begin{equation*}
\partial_{\widehat{L}}\|\widehat{L}\|_*=\{\widehat{U}\widehat{V}^*+\widehat{Q}: \widehat{Q}\in\mathcal{\widehat{T}^\perp},\|\widehat{Q}\|\le1\},
\end{equation*}
\begin{equation*}
\partial_{\widehat{S}}\|\widehat{S}\|_{2,1}=\{\mathcal{B}(\widehat{S})+\widehat{E}: \widehat{E}\in\mathcal{\widehat{I}^\perp},\|\widehat{E}\|_{2,\infty}\le1\}.
\end{equation*}

By the definition of subgradient, the inequality follows
\begin{equation*}
\begin{split}
&\ \ \ \ \ \|L_0+H_L\|_*+\lambda\|S_0'-H_S\|_{2,1}\\
&\ge\|\widehat{L}\|_*+\lambda\|\widehat{S}\|_{2,1}+\langle\widehat{U}\widehat{V}^*+\widehat{Q}, \mathcal{P}_{\mathcal{U}_0^\perp}H_L\rangle-\lambda\langle\mathcal{B}(\widehat{S})+\widehat{E},\mathcal{P}_{\Omega_{obs}}\mathcal{P}_{\mathcal{U}_0^\perp}H_L\rangle\\
&\ge\|\widehat{L}\|_*+\lambda\|\widehat{S}\|_{2,1}+\langle\widehat{U}\widehat{V}^*, \mathcal{P}_{\mathcal{U}_0^\perp}H_L\rangle+\langle\widehat{Q}, \mathcal{P}_{\mathcal{U}_0^\perp}H_L\rangle-\lambda\langle\mathcal{B}(\widehat{S}), \mathcal{P}_{\mathcal{U}_0^\perp}H_L\rangle-\lambda\langle\widehat{E}, \mathcal{P}_{\Omega_{obs}}\mathcal{P}_{\mathcal{U}_0^\perp}H_L\rangle.\\
\end{split}
\end{equation*}
Now adopt $\widehat{Q}$ such that $\langle \widehat{Q},\mathcal{P}_{\mathcal{U}_0^\perp}H_L\rangle=\|\mathcal{P}_{\widehat{\mathcal{V}}^\perp}\mathcal{P}_{\mathcal{U}_0^\perp}H_L\|_*$ and $\langle\widehat{E}, \mathcal{P}_{\Omega_{obs}}\mathcal{P}_{\mathcal{U}_0^\perp}H_L\rangle=-\|\mathcal{P}_\Gamma\mathcal{P}_{\mathcal{U}_0^\perp}H_L\|_{2,1}$\footnote{By the duality between the nuclear norm and the operator norm, there exists a $Q$ such that $\langle Q, \mathcal{P}_{\widehat{\mathcal{V}}^\perp}\mathcal{P}_{\mathcal{U}_0^\perp}H\rangle=\|\mathcal{P}_{\widehat{\mathcal{V}}^\perp}\mathcal{P}_{\mathcal{U}_0^\perp}H\|_*$ and $\|Q\|\le1$. Thus we take $\widehat{Q}=\mathcal{P}_{\mathcal{U}_0^\perp}\mathcal{P}_{\widehat{\mathcal{V}}^\perp}Q\in\mathcal{\widehat{T}^\perp}$. It holds similarly for $\widehat{E}$.}. We have
\begin{equation*}
\begin{split}
&\ \ \ \ \ \|L_0+H_L\|_*+\lambda\|S_0'-H_S\|_{2,1}\\
&\ge\|\widehat{L}\|_*+\lambda\|\widehat{S}\|_{2,1}+\|\mathcal{P}_{\widehat{\mathcal{V}}^\perp}\mathcal{P}_{\mathcal{U}_0^\perp}H_L\|_*+\lambda\|\mathcal{P}_\Gamma\mathcal{P}_{\mathcal{U}_0^\perp}H_L\|_{2,1}\\&\ \ \ -\lambda\langle\mathcal{B}(\widehat{S}), \mathcal{P}_{\mathcal{U}_0^\perp}H_L\rangle.
\end{split}
\end{equation*}
Notice that
\begin{equation*}
\begin{split}
|\langle-\lambda\mathcal{B}(\widehat{S}),& \mathcal{P}_{\mathcal{U}_0^\perp}H_L\rangle|=|\langle\lambda \widehat{F}-\widehat{W}, \mathcal{P}_{\mathcal{U}_0^\perp}H_L\rangle|\\
&\le|\langle \widehat{W}, \mathcal{P}_{\mathcal{U}_0^\perp}H_L\rangle|+\lambda|\langle \widehat{F}, \mathcal{P}_{\mathcal{U}_0^\perp}H_L\rangle|\\
&\le\frac{1}{2}\|\mathcal{P}_{\mathcal{\widehat{V}}^\perp}\mathcal{P}_{\mathcal{U}_0^\perp}H_L\|_*+\frac{\lambda}{2}\|\mathcal{P}_\Gamma\mathcal{P}_{\mathcal{U}_0^\perp}H_L\|_{2,1}.
\end{split}
\end{equation*}
Hence
\begin{equation*}
\begin{split}
&\ \ \ \ \|\widehat{L}\|_*+\lambda\|\widehat{S}\|_{2,1}\\
&\ge\|L_0+H_L\|_*+\lambda\|S_0'-H_S\|_{2,1}\\
&\ge\|\widehat{L}\|_*+\lambda\|\widehat{S}\|_{2,1}+\frac{1}{2}\|\mathcal{P}_{\widehat{\mathcal{V}}^\perp}\mathcal{P}_{\mathcal{U}_0^\perp}H_L\|_*+ \frac{\lambda}{2}\|\mathcal{P}_\Gamma\mathcal{P}_{\mathcal{U}_0^\perp}H_L\|_{2,1}.
\end{split}
\end{equation*}
So $\mathcal{P}_{\mathcal{U}_0^\perp}H_L\in\widehat{\mathcal{V}}\cap\Gamma^\perp=\{0\}$, i.e., $H_L\in\mathcal{U}_0$.

\end{proof}

The following lemma shows that one of the conditions in Lemma \ref{lemma: dual conditions for exact column space} holds true.
\begin{lemma}
Under the assumption of Theorem \ref{theorem: exact recovery under Bernoulli sampling}, $\widehat{\mathcal{V}}\cap\Gamma^\perp=\{0\}$.
\end{lemma}

\begin{proof}
We first prove $p(1-p_0)\|\mathcal{P}_{\mathcal{\widehat{V}}}\mathcal{P}_{\Gamma^\perp}M\|_F\le2\|\mathcal{P}_{\mathcal{\widehat{V}}^\perp}\mathcal{P}_{\Gamma^\perp}M\|_F$ for any matrix $M$. Let $M'=\mathcal{P}_{\Gamma^\perp}M$. Because $\mathcal{P}_\Gamma\mathcal{P}_\mathcal{\widehat{V}}M'+\mathcal{P}_\Gamma\mathcal{P}_\mathcal{\widehat{V}^\perp}M'=0$, we have $\|\mathcal{P}_\Gamma\mathcal{P}_\mathcal{\widehat{V}}M'\|_F=\|\mathcal{P}_\Gamma\mathcal{P}_\mathcal{\widehat{V}^\perp}M'\|_F\le\|\mathcal{P}_\mathcal{\widehat{V}^\perp}M'\|_F$. Note that
\begin{equation*}
\begin{split}
&\ \ \ \ (p(1-p_0))^{-1}\|\mathcal{P}_\Gamma\mathcal{P}_\mathcal{\widehat{V}}M'\|_F\\&=(p(1-p_0))^{-1}\langle\mathcal{P}_\Gamma\mathcal{P}_\mathcal{\widehat{V}}M',\mathcal{P}_\Gamma\mathcal{P}_\mathcal{\widehat{V}}M'\rangle\\
&=\langle\mathcal{P}_\mathcal{\widehat{V}}M',(p(1-p_0))^{-1}\mathcal{P}_\mathcal{\widehat{V}}\mathcal{P}_\Gamma\mathcal{P}_\mathcal{\widehat{V}}M'\rangle\\
&=\langle \mathcal{P}_\mathcal{\widehat{V}}M',((p(1-p_0))^{-1}\mathcal{P}_\mathcal{\widehat{V}}\mathcal{P}_\Gamma\mathcal{P}_\mathcal{\widehat{V}}-\mathcal{P}_\mathcal{\widehat{V}})\mathcal{P}_\mathcal{\widehat{V}}M'\rangle\\&\ \ \ +\langle\mathcal{P}_\mathcal{\widehat{V}}M',\mathcal{P}_\mathcal{\widehat{V}}M'\rangle\\
&\ge\|\mathcal{P}_\mathcal{\widehat{V}}M'\|_F-\frac{1}{2}\|\mathcal{P}_\mathcal{\widehat{V}}M'\|_F\\
&=\frac{1}{2}\|\mathcal{P}_\mathcal{\widehat{V}}M'\|_F,
\end{split}
\end{equation*}
where the first inequality holds due to Corollary \ref{corollary: sum 2 norm}.
So we have
\begin{equation*}
\|\mathcal{P}_\mathcal{\widehat{V}^\perp}M'\|_F\ge\|\mathcal{P}_\Gamma\mathcal{P}_\mathcal{\widehat{V}}M'\|_F\ge\frac{p(1-p_0)}{2}\|\mathcal{P}_\mathcal{\widehat{V}}M'\|_F,
\end{equation*}
i.e., $p(1-p_0)\|\mathcal{P}_{\mathcal{\widehat{V}}}\mathcal{P}_{\Gamma^\perp}M\|_F\le2\|\mathcal{P}_{\mathcal{\widehat{V}}^\perp}\mathcal{P}_{\Gamma^\perp}M\|_F$.

Now let $M\in\widehat{\mathcal{V}}\cap\Gamma^\perp$. Then $\mathcal{P}_{\mathcal{\widehat{V}}^\perp}\mathcal{P}_{\Gamma^\perp}M=0$ while $\mathcal{P}_{\mathcal{\widehat{V}}}\mathcal{P}_{\Gamma^\perp}M=M$. So $p(1-p_0)\|M\|_F\le0$, i.e., $M=0$. Therefore, $\widehat{\mathcal{V}}\cap\Gamma^\perp=\{0\}$.
\end{proof}

By Lemma \ref{lemma: dual conditions for exact column space}, to prove the exact recovery of column space, it suffices to show a dual certificate $\widetilde{W}$ such that
\begin{equation}
\label{equ: dual conditions for exact space}
\begin{cases}
\widehat{W}\in \mathcal{\widehat{V}^\perp}\cap\Omega_{obs},\\
\|\widehat{W}\|\le1/2,\\
\mathcal{P}_\Pi\widehat{W}=\lambda\mathcal{B(\widehat{S})},\ \ \Pi=\mathcal{I}_0\cap\Omega_{obs},\\
\|\mathcal{P}_\Gamma\widehat{W}\|_{2,\infty}\le\lambda/2,\ \ \Gamma=\mathcal{I}_0^\perp\cap\Omega_{obs}.
\end{cases}
\end{equation}

\subsubsection{Certification by Least Squares}
\label{subsubsection: certification by least squares}
The remainder of proofs is to construct $\widehat{W}$ which satisfies
the dual conditions \eqref{equ: dual conditions for exact space}. Note that $\mathcal{\widehat{I}}=\mathcal{I}_0\sim\mbox{Ber}(p)$.
To construct $W$, we consider the method of least squares, which is
\begin{equation}
\label{equ: W}
\widehat{W}=\lambda\mathcal{P}_{\mathcal{\widehat{V}^\perp}\cap\Omega_{obs}}\sum_{k\ge0}(\mathcal{P}_\Pi\mathcal{P}_{\mathcal{\widehat{V}}+\Omega_{obs}^\perp}\mathcal{P}_\Pi)^k\mathcal{B}(\widehat{S}),
\end{equation}
where the Neumann series is well defined due to $\|\mathcal{P}_\Pi\mathcal{P}_{\mathcal{\widehat{V}}+\Omega_{obs}^\perp}\mathcal{P}_\Pi\|<1$. Indeed, note that $\Pi\subseteq\Omega_{obs}$. So we have the identity:
\begin{equation*}
\begin{split}
&\ \ \ \ \mathcal{P}_\Pi\mathcal{P}_{\mathcal{\widehat{V}}+\Omega_{obs}^\perp}\mathcal{P}_\Pi\\&=\mathcal{P}_\Pi(\mathcal{P}_{\mathcal{\widehat{V}}}+\mathcal{P}_{\Omega_{obs}^\perp}-\mathcal{P}_{\mathcal{\widehat{V}}}\mathcal{P}_{\Omega_{obs}^\perp}-\mathcal{P}_{\Omega_{obs}^\perp}\mathcal{P}_{\mathcal{\widehat{V}}}+...)\mathcal{P}_\Pi\\
&=\mathcal{P}_\Pi\mathcal{P}_{\mathcal{\widehat{V}}}(\mathcal{P}_{\mathcal{\widehat{V}}}+\mathcal{P}_{\mathcal{\widehat{V}}}\mathcal{P}_{\Omega_{obs}^\perp}\mathcal{P}_{\mathcal{\widehat{V}}}+...)\mathcal{P}_{\mathcal{\widehat{V}}}\mathcal{P}_\Pi\\
&=\mathcal{P}_\Pi\mathcal{P}_{\mathcal{\widehat{V}}}(\mathcal{P}_{\mathcal{\widehat{V}}}-\mathcal{P}_{\mathcal{\widehat{V}}}\mathcal{P}_{\Omega_{obs}^\perp}\mathcal{P}_{\mathcal{\widehat{V}}})^{-1}\mathcal{P}_{\mathcal{\widehat{V}}}\mathcal{P}_\Pi\\
&=\mathcal{P}_\Pi\mathcal{P}_{\mathcal{\widehat{V}}}(\mathcal{P}_{\mathcal{\widehat{V}}}\mathcal{P}_{\Omega_{obs}}\mathcal{P}_{\mathcal{\widehat{V}}})^{-1}\mathcal{P}_{\mathcal{\widehat{V}}}\mathcal{P}_\Pi
.
\end{split}
\end{equation*}
By Lemma \ref{lemma: Omega 2 norm} and the triangle inequality, we have that $1-(1-p_0)^{-1}\|\mathcal{P}_{\mathcal{\widehat{V}}}\mathcal{P}_{\Omega_{obs}}\mathcal{P}_{\mathcal{\widehat{V}}}\|<1/2$, i.e., $\|(\mathcal{P}_{\mathcal{\widehat{V}}}\mathcal{P}_{\Omega_{obs}}\mathcal{P}_{\mathcal{\widehat{V}}})^{-1}\|<2/(1-p_0)$. Therefore,
\begin{equation}
\label{equ: operator norm}
\begin{split}
\|\mathcal{P}_\Pi\mathcal{P}_{\mathcal{\widehat{V}}+\Omega_{obs}^\perp}\|^2&=\|\mathcal{P}_\Pi\mathcal{P}_{\mathcal{\widehat{V}}+\Omega_{obs}^\perp}\mathcal{P}_\Pi\|\\&\le2(1-p_0)^{-1}\|\mathcal{P}_{\mathcal{\widehat{V}}}\mathcal{P}_\Pi\|^2\\&\le2(1-p_0)^{-1}\sigma^2\\&<1,
\end{split}
\end{equation}
where the second inequality holds due to Corollary \ref{corollary: gamma v bound}.
Note that $\mathcal{P}_\Omega \widehat{W}=\lambda\mathcal{B}(\mathcal{\widehat{S}})$ and $\widehat{W}\in\mathcal{\widehat{V}^\perp}\cap\Omega_{obs}$. So to prove the dual conditions \eqref{equ: dual conditions for exact space}, it suffices to show that
\begin{flalign}
\label{equ: dual conditions for W}
\begin{cases}
\mbox{(a)}\ \ \|\widehat{W}\|\le1/2,\\
\mbox{(b)}\ \ \|\mathcal{P}_\Gamma\widehat{W}\|_{2,\infty}\le\lambda/2.
\end{cases}
\end{flalign}

\subsubsection{Proofs of Dual Conditions}
\label{subsubsection: proofs of dual conditions 2}
We now prove that the dual variables that we construct above satisfy our dual conditions.
\begin{lemma}
\label{lemma: correction for column space}
Under the assumptions of Theorem \ref{theorem: exact recovery under Bernoulli sampling}, $\widehat{W}$ given by \eqref{equ: W} obeys dual
conditions \eqref{equ: dual conditions for W}.
\end{lemma}

\begin{proof}
Let
$\mathcal{H}=\sum_{k\ge1}(\mathcal{P}_\Pi\mathcal{P}_{\mathcal{\widehat{V}}+\Omega_{obs}^\perp}\mathcal{P}_\Pi)^k$.
Then
\begin{equation}
\label{equ: proof of W 1}
\begin{split}
\widehat{W}&=\lambda\mathcal{P}_{\mathcal{\widehat{V}^\perp}\cap\Omega_{obs}}\sum_{k\ge0}(\mathcal{P}_\Pi\mathcal{P}_{\mathcal{\widehat{V}}+\Omega_{obs}^\perp}\mathcal{P}_\Pi)^k\mathcal{B}(\widehat{S})\\
&=\lambda\mathcal{P}_{\mathcal{\widehat{V}^\perp}\cap\Omega_{obs}}\mathcal{B}(\widehat{S})+\lambda\mathcal{P}_{\mathcal{\widehat{V}^\perp}\cap\Omega_{obs}}\mathcal{H}(\mathcal{B}(\widehat{S})),
\end{split}
\end{equation}
Now we check the two conditions in \eqref{equ: dual conditions for
W}.

(a) By the assumption, we have $\|\mathcal{B}(\widehat{S})\|\leq\mu'$.
%\begin{equation}
%\begin{split}
%\|\mathcal{B}(\widehat{S})\|&=\|A\|=\frac{1}{\sqrt{n}}\|\sqrt{n}A\|\\
%&\le\frac{2+C}{\sqrt{n}}\sqrt{n}\\&
%=2+C.
%\end{split}
%\end{equation}
Thus the first term in \eqref{equ: proof of W 1} obeys
\begin{equation}
\label{equ:W_first_term}
\lambda\left\Vert\mathcal{P}_{\mathcal{\widehat{V}^\perp}\cap\Omega_{obs}}\mathcal{B}(\widehat{S})\right\Vert\le\lambda\left\Vert\mathcal{B}(\widehat{S})\right\Vert\le\frac{1}{4}.
\end{equation}
For the second term, we have
\begin{equation*}
\lambda\|\mathcal{P}_{\mathcal{\widehat{V}^\perp}\cap\Omega_{obs}}\mathcal{H}(\mathcal{B}(\widehat{S}))\|\le\lambda\|\mathcal{H}\|\left\Vert\mathcal{B}(\widehat{S})\right\Vert.
\end{equation*}
Then according to \eqref{equ: operator norm} which states that $\|\mathcal{P}_{\mathcal{\widehat{V}}+\Omega_{obs}^\perp}\mathcal{P}_\Pi\|^2\le2\sigma^2/(1-p_0)\triangleq\sigma_0^2$ with high probability,
\begin{equation*}
\|\mathcal{H}\|\le\sum_{k\ge1}\sigma_0^{2k}=\frac{\sigma_0^2}{1-\sigma_0^2}\le1.
\end{equation*}
So
\begin{equation*}
\lambda\|\mathcal{P}_{\mathcal{\widehat{V}^\perp}\cap\Omega_{obs}}\mathcal{H}(\mathcal{B}(\widehat{S}))\|\le\frac{1}{4}.
\end{equation*}
That is
\begin{equation*}
\|\widehat{W}\|\le\frac{1}{2}.
\end{equation*}

(b) Let $\mathcal{G}$ stand for
$\mathcal{G}=\sum_{k\ge0}(\mathcal{P}_\Pi\mathcal{P}_{\mathcal{\widehat{V}}+\Omega_{obs}^\perp}\mathcal{P}_\Pi)^k$. Then $\widehat{W}=\lambda\mathcal{P}_{\mathcal{\widehat{V}^\perp}\cap\Omega_{obs}}\mathcal{G}(\mathcal{B}(\widehat{S}))$.
Notice that $\mathcal{G}(\mathcal{B}(\widehat{S}))\in\mathcal{I}_0$. Thus
\begin{equation*}
\begin{split}
\mathcal{P}_\Gamma \widehat{W}&=\lambda\mathcal{P}_{\mathcal{I}_0^\perp}\mathcal{P}_{\mathcal{\widehat{V}^\perp}\cap\Omega_{obs}}\mathcal{G}(\mathcal{B}(\widehat{S}))\\
&=\lambda\mathcal{P}_{\mathcal{I}_0^\perp}\mathcal{G}(\mathcal{B}(\widehat{S}))-\lambda\mathcal{P}_{\mathcal{I}_0^\perp}\mathcal{P}_{\mathcal{\widehat{V}}+\Omega_{obs}^\perp}\mathcal{G}(\mathcal{B}(\widehat{S}))\\
&=-\lambda\mathcal{P}_{\mathcal{I}_0^\perp}\mathcal{P}_{\mathcal{\widehat{V}}+\Omega_{obs}^\perp}\mathcal{G}(\mathcal{B}(\widehat{S})).
\end{split}
\end{equation*}
Now denote $Q\triangleq\mathcal{P}_{\mathcal{\widehat{V}}+\Omega_{obs}^\perp}\mathcal{G}(\mathcal{B}(\widehat{S}))$. Note that

\begin{equation*}
\begin{split}
\|Q_{:j}\|^2&=\sum_iQ_{ij}^2=\sum_i\left\langle\mathcal{P}_{\mathcal{\widehat{V}}+\Omega_{obs}^\perp}\mathcal{G}(\mathcal{B}(\widehat{S})), \omega_{ij}\right\rangle^2\\
&=\sum_i\left\langle\mathcal{B}(\widehat{S}), \mathcal{G}\mathcal{P}_\Pi\mathcal{P}_{\mathcal{\widehat{V}}+\Omega_{obs}^\perp}(\omega_{ij})\right\rangle^2\\&=\sum_i\sum_{j_0}\left\langle[\mathcal{B}(\widehat{S})]_{:j_0}, \mathcal{G}\mathcal{P}_\Pi\mathcal{P}_{\mathcal{\widehat{V}}+\Omega_{obs}^\perp}(\omega_{ij})e_{j_0}\right\rangle^2\\
&=\sum_{j_0}\hspace{-0.1cm}\sum_i\hspace{-0.1cm}\left(\hspace{-0.1cm}(G_je_i)^*[\mathcal{B}(\widehat{S})]_{:j_0}\right)^2\hspace{-0.15cm}\left(\hspace{-0.1cm}(G_je_i)^*\mathcal{G}\mathcal{P}_\Pi\mathcal{P}_{\mathcal{\widehat{V}}+\Omega_{obs}^\perp}(\omega_{ij})e_{j_0}\hspace{-0.1cm}\right)^2\\
&\le\sum_{j_0}\left((G_je_{i_m})^*\mathcal{G}\mathcal{P}_\Pi\mathcal{P}_{\mathcal{\widehat{V}}+\Omega_{obs}^\perp}(\omega_{i_mj})e_{j_0}\right)^2\\
&=\left\Vert(G_je_{i_m})^*\mathcal{G}\mathcal{P}_\Pi\mathcal{P}_{\mathcal{\widehat{V}}+\Omega_{obs}^\perp}(\omega_{i_mj})\right\Vert_2^2\\&\le\|\mathcal{G}\|\ \left\Vert\mathcal{P}_\Pi\mathcal{P}_{\mathcal{\widehat{V}}+\Omega_{obs}^\perp}\right\Vert\\&\le\frac{1}{4},\quad\forall j,
\end{split}
\end{equation*}
where $i_m=\arg \max_i |e_i^*\mathcal{G}\mathcal{P}_\Omega\mathcal{P}_{\mathcal{\widehat{V}}+\Omega_{obs}^\perp}(e_ie_j^*)e_{j_0}|$, $G_j$ is a unitary matrix, and the second inequality holds because of fact \eqref{equ: operator norm}.
Thus $\|\mathcal{P}_\Gamma \widehat{W}\|_{2,\infty}=\lambda\|\mathcal{P}_{\mathcal{I}_0^\perp} Q\|_{2,\infty}\le\lambda\|Q\|_{2,\infty}\le\lambda/2$. The proofs are completed.
\end{proof}

\section{Algorithm}
\label{section: algorithm}
It is well known that robust MC can be efficiently solved by Alternating Direction Method of Multipliers (ADMM)~\cite{Lin2}, which is probably the most widely used method for solving nuclear norm minimization problems. In this section, we develop a faster algorithm, termed $\ell_{2,1}$ filtering algorithm, to solve the same problem.

\subsection{$\ell_{2,1}$ Filtering Algorithm}
\label{section: l{2,1} filtering algorithm}
Briefly speaking, our $\ell_{2,1}$ filtering algorithm consists of two steps:
recovering the ground truth subspace from a randomly selected sub-column matrix, and then processing the remaining columns via
$\ell_{2,1}$ norm based linear regression, which turns out to be a
least square problem.

\subsubsection{Recovering Subspace from a Seed Matrix}
To speed up the algorithm, our strategy is to focus on a small-scaled subproblem from which we can recover the same subspace as solving the whole original problem~\cite{zhang2015relations}. To this end, we partition the
whole matrix into two blocks. Suppose that $r=\mbox{rank}(L)\ll
\min\{m,n\}$. We randomly sample $k$ columns from $M$ by i.i.d. Ber$(d/n)$ (our Theorem \ref{theorem: success of steps 2} suggests choosing $d$ as $\Theta(r\log^3 n)$), forming a submatrix
$M_{l}$ (for brevity, we assume that $M_{l}$ is the leftmost submatrix
of $M$). Then we can partition $M$, $L$, and $S$ accordingly:
$$M=[M_{l},M_{r}],\quad S=[S_{l},S_{r}],\quad L=[L_{l},L_{r}].$$
To recover the desired subspace $\mbox{Range}(L_0)$ from $M_{l}$, we solve a small-scaled problem:
\begin{equation}
\label{equrecovery of the seed matrix}
\begin{split}
&\min_{L_{l},S_l}
\|L_{l}\|_*+\frac{1}{\sqrt{\log k}}\|S_{l}\|_{2,1},\ \ \mbox{s.t.} \ \
\mathcal{R}'(M_{l})=\mathcal{R}'(L_{l}+S_{l})\in\mathbb{R}^{m\times k},
\end{split}
\end{equation}
where $\mathcal{R}'(\cdot)$ is a linear mapping restricting $\mathcal{R}(\cdot)$ on the column index of $M_l$. As we will show in Section \ref{section: theoretical guarantee for algorithm}, when the Bernoulli parameter $d$ is no less than a lower bound, problem \eqref{equrecovery of the seed matrix} exactly recovers the correct subspace $\mbox{Range}(L_0)$ and the column support of $[S_0]_{l}$ with an overwhelming probability.

\subsubsection{$\ell_{2,1}$ Filtering Step}
Since $\mbox{Range}(L_{l})=\mbox{Range}(L_0)$ at an overwhelming probability, each column of
$L_{r}$ can be represented as the linear combinations of $L_{l}$. Namely, there exists a
representation matrix $Q\in\mathbb{R}^{k\times (n-k)}$ such that
\begin{equation*}
L_{r}=L_{l}Q.
\end{equation*}
Note that the part $S_{r}$ should have very sparse columns, so we use the following $\ell_{2,1}$ norm based linear
regression problem to explore the column supports of $S_r$:
\begin{equation}
\label{equl21 filtering}
\min_{Q,S_{r}} \|S_{r}\|_{2,1}, \ \ \mbox{s.t.} \ \ \mathcal{R'}(M_{r})=\mathcal{R'}(L_{l}Q+S_{r}).
\end{equation}

If we solve problem \eqref{equl21 filtering} directly by using
ADMM~\cite{LiuR2}, the complexity of our algorithm will be nearly
the same as that of solving the whole original problem.
Fortunately, we can solve \eqref{equl21 filtering} column-wise
due to the separability of $\ell_{2,1}$ norms.
Let $M_{r}^{(i)}$, $q^{(i)}$, and $S_{r}^{(i)}$ represent the $i$th
column of $M_{r}$, $Q$, and $S_{r}$, respectively
($i=1,...,n-sr$). Then problem \eqref{equl21 filtering}
could be decomposed into $n-k$ subproblems:
\begin{equation}
\label{equleast-square problem}
\begin{split}
&\min_{q^{(i)},S_{r}^{(i)}}
\|S_{r}^{(i)}\|_2,\ \ \ \mbox{s.t.} \ \
\mathcal{R}_i'(M_{r})^{(i)}=\mathcal{R}_i'(L_{l}q+S_{r})^{(i)}\in\mathbb{R}^m,\ \ i=1,...,n-k.
\end{split}
\end{equation}
Equivalently,
\begin{equation}
\label{equleast-square problem 2}
\begin{split}
&\min_{q^{(i)},\mathcal{Z}_i'(S_{r}^{(i)})}
\|\mathcal{Z}_i'(S_{r}^{(i)})\|_2,\ \ \mbox{s.t.} \ \
\mathcal{Z}_i'(M_{r}^{(i)})=\mathcal{Y}_i'(L_{l})q^{(i)}+\mathcal{Z}_i'(S_{r}^{(i)})\in\mathbb{R}^{h_i},\ \ i=1,...,n-k,
\end{split}
\end{equation}
where $\mathcal{Z}_i'$ is an operator functioning on a vector which wipes out the unobserved elements, $\mathcal{Y}_i'$ is a matrix operator which wipes out the corresponding rows of a matrix, and $h_i$ is the number of observed elements in the $i$th column.
As least square problems, \eqref{equleast-square problem} admits
closed-form solutions $q^{(i)}=\mathcal{Y}_i'(L_{l})^\dag \mathcal{Z}_i'(M_{r}^{(i)}),\ \mathcal{Z}_i'(S_{r}^{(i)})=\mathcal{Z}_i'(M_{r}^{(i)})-\mathcal{Y}_i'(L_{l})\mathcal{Y}_i'(L_{l})^\dag \mathcal{Z}_i'(M_{r}^{(i)}),
i=1,...,n-k$. If $\mathcal{Z}_i'(S_{r}^{(i)})\not=\textbf{0}$, we infer that the column $M_{r}^{(i)}$ is corrupted by noises.

We summarize our $\ell_{2,1}$ filtering algorithm in Algorithm \ref{algl_21 filtering}.
\begin{algorithm}
\caption{$\ell_{2,1}$ Filtering Algorithm for Exact Recovery of Subspace and Support}
\label{algl_21 filtering}
\begin{algorithmic}
\STATE {\bfseries Input:} Observed data matrix $\mathcal{R}(M)$ and estimated
rank $r$ (see Section \ref{Target Rank Estimation}).
\label{alg: 1}\STATE {\bfseries 1.} Randomly sample columns from $\mathcal{R}(M)\in\mathbb{R}^{m\times n}$ by Ber$(d/n)$ to form $\mathcal{R}(M_{l})\in\mathbb{R}^{m\times k}$;
\label{alg: 12}\STATE {\bfseries 2.} // Line 3 recovers the subspace from a seed matrix.
%\STATE {\bfseries 3.} \ \ \ \ \ \ Randomly sample other $cr(\log n)^3-w$ columns from unselected part and combine them with $\mathcal{R}(M_{l})$.
\label{alg: 6}\STATE {\bfseries 3.} Solve small-scaled $m\times k$ problem \eqref{equrecovery of the seed matrix} by ADMM and
obtain $L_{l}$, $Range(L_0)$, and column support of $S_{l}$;
\label{alg: 7}\STATE {\bfseries 4.} \textbf{For} $i$ from $1$ to $n-k$
\label{alg: 8}\STATE {\bfseries 5.} \ \ \ \ \ Conduct QR factorization on the matrix $\mathcal{Y}_i'(L_{l})$ as $\mathcal{Y}_i'(L_{l})=Q_iR_i$;
\label{alg: 13}\STATE {\bfseries 6.} \ \ \ \ \ // Line 7 implements $\ell_{2,1}$ filtering to the remaining columns.
\label{alg: 9}\STATE {\bfseries 7.} \ \ \ \ \ Recover $\mathcal{Z}_i'(S_{r}^{(i)})\in\mathbb{R}^{h_i}$ by solving \eqref{equleast-square problem 2}, which is $\mathcal{Z}_i'(S_{r}^{(i)})=\mathcal{Y}_i'(M_{r}^{(i)})-Q_i(Q_i^*\mathcal{Y}_i'(M_{r}^{(i)}))$;
\label{alg: 10}\STATE {\bfseries 8.} \ \ \ \ \ \textbf{If} $\mathcal{Y}_i'(S_{r}^{(i)})\not=\textbf{0}$
\label{alg: 11}\STATE {\bfseries 9.} \ \ \ \ \ \ \ \ \ \ Output ``$M_{r}^{(i)}$ is an outlier'';
\label{alg: 14}\STATE {\bfseries 10.} \ \ \ \ \textbf{End If}
\label{alg: 15}\STATE {\bfseries 11.} \textbf{End For}
\STATE {\bfseries Output:} Low-dimensional subspace Range$(L_0)$ and column support of matrix $S_0$.
\end{algorithmic}
\end{algorithm}

\subsection{Target Rank Estimation}
\label{Target Rank Estimation}
As we mentioned above, our algorithm requires the rank estimation $r$ as an input. For some specific applications, e.g., background modeling~\cite{Candes} and photometric stereo~\cite{wu2011robust}, the rank of the underlying matrix is known to us due to their physical properties. However, it is not always clear how to estimate the rank for some other cases. Here we provide a heuristic strategy for rank estimation.

Our strategy is based on the multiple trials of solving subproblem \eqref{equrecovery of the seed matrix}. Namely, starting from a small $r$ estimation, we solve subproblem \eqref{equrecovery of the seed matrix} by subsampling. If the optimal solution $L_l^*$ is such that $k/\mbox{rank}(L_l^*)\ge \Theta(\log^3 n)$, we accept the $r$ and output; Otherwise, we increase $r$ by a fixed step (and so increase $k$) and repeat the procedure until $k/n\ge0.5$. We require $k/n<0.5$ because the speed advantage of our $\ell_{2,1}$ algorithm vanishes if the low-rank assumption does not hold.

\subsection{Theoretical Guarantees}
\label{section: theoretical guarantee for algorithm}
In this section, we establish theoretical guarantees for our $\ell_{2,1}$ filtering algorithm. Namely, Algorithm \ref{algl_21 filtering} is able to exactly recover the range space of $L_0$ and the column support of $S_0$ with a high probability. To this end, we show that the two steps in Section \ref{section: l{2,1} filtering algorithm} succeed at overwhelming probabilities, respectively:
\begin{itemize}
\item
To guarantee the exact recovery of Range($L_0$) from the seed matrix, we prove that the sampled columns in Line \ref{alg: 1} exactly span the desired subspace Range($L_0$) when the columns are restricted to the set $\mathcal{I}_0^\perp$, i.e., Range$(\mathcal{P}_{\mathcal{I}_0^\perp}M_{l})$=Range($L_0$) (see Theorem \ref{theorem: complete bases of low-rank part}); Otherwise, only a subspace of Range($L_0$) can be recovered by Line \ref{alg: 6}.
Applying Theorem \ref{theorem: exact recovery under Bernoulli sampling}, we justify that Line \ref{alg: 6} recovers the ground truth subspace from the seed matrix with an overwhelming probability (see Theorem \ref{theorem: success of steps 2}).
\item
For $\ell_{2,1}$ filtering step, we demonstrate that, though operator $\mathcal{Y}_i'$ randomly wipes out several rows of $L_{l}$, the columns of $\mathcal{Y}_i'(L_{l})$ exactly span Range$(\mathcal{Y}_i'(L_0))$ with an overwhelming probability. So by checking whether the $i$th column belongs to Range$(\mathcal{Y}_i'(L_0))$, the least squares problem \eqref{equleast-square problem 2} suffices to examine whether a specific column $M_{r}^{(i)}$ is an outlier (see Theorem \ref{theorem: complete bases of elements observation}).
\end{itemize}

\subsubsection{Analysis for Recovering Subspace from a Seed Matrix}
To guarantee the recovery of Range$(L_0)$ by Line \ref{alg: 6}, the sampled columns in Line \ref{alg: 1} should be informative. In other words, $\mbox{Range}(L_0)=\mbox{Range}(\mathcal{P}_{\mathcal{I}_0^\perp}M_{l})$. To select the smallest number of columns in Line \ref{alg: 1}, we estimate the lower bound for the Bernoulli parameter $d/n$. Intuitively, this problem is highly connected to the property of $\mathcal{P}_{\mathcal{I}_0^\perp}M$. For instance, suppose that in the worst case $\mathcal{P}_{\mathcal{I}_0^\perp}M$ is a matrix whose elements in the first column are ones while all other elements equal zeros. By this time, Line \ref{alg: 1} will select the first column (the only complete basis) at a high probability if and only if $d=n$. But for $\mathcal{P}_{\mathcal{I}_0^\perp}M$ whose elements are all equal to ones, a much smaller $d$ suffices to guarantee the success of sampling. Thus, to identify the two cases, we involve the incoherence in our analysis.

We now estimate the smallest Bernoulli parameter $d$ in Line \ref{alg: 1} which ensures that $\mbox{Range}(L_0)\subseteq\mbox{Range}(M_{l})$, or equivalently $\mbox{Range}(L_0)=\mbox{Range}(\mathcal{P}_{\mathcal{I}_0^\perp}M_{l})$, at an overwhelming probability. The following theorem illustrates the result:
\begin{theorem}[Sampling a Set of Complete Basis by Line \ref{alg: 1}]
\label{theorem: complete bases of low-rank part}
Suppose that each column of the incoherent $L_0$ is sampled i.i.d. by Bernoulli distribution with parameter $d/n$. Let $[L_0]_l$ be the selected columns from $L_0$, i.e., $[L_0]_l=\sum_j\delta_j[L_0]_{:j} e_j^*$, where $\delta_j\sim \mbox{Ber}(d/n)$. Then with probability at least $1-\delta$, we have $\mbox{Range}([L_0]_l)=\mbox{Range}(L_0)$, provided that
$
d\ge 2\mu r\log\frac{r}{\delta},
$
where $\mu$ is the incoherence parameter on the row space of matrix $L_0$.
\end{theorem}
\begin{proof}
The proof of Theorem \ref{theorem: complete bases of low-rank part} can be found in the Appendices.
\end{proof}

\begin{remark}
Note that a large incoherence parameter on the row space implies that slightly perturbing $L_0$ tremendously changes its column space. So we will need more samples in order to capture enough information about the column space of $L_0$.
\end{remark}

To guarantee the exact recovery of desired subspace from the seed matrix, the rank $r$ of intrinsic matrix should be low enough compared with the input size (see Theorem \ref{theorem: exact recovery under Bernoulli sampling}). Note that Line \ref{alg: 1}, however, selects the columns by i.i.d. Ber$(d/n)$, so that the number $k$ of sampled columns is a random variable. Roughly, $k$ should be around $d$ due to the fact $\mathbb{E}(k)=d$. The following lemma implies that the magnitude of $k$ typically has the same order as that of parameter $d$ with an overwhelming probability.
\begin{lemma}
\label{lemma: number of sampling by Bernoulli}
Let $n$ be the number of Bernoulli trials and suppose that $\Omega\sim\mbox{Ber}(d/n)$. Then with an overwhelming probability, $|\Omega|=\Theta(d)$, provided that $d\ge c\log n$ for a numerical constant $c$.
\end{lemma}
\begin{proof}
Take a perturbation $\epsilon$ such that $d/n=m/n+\epsilon$. By scalar Chernoff bound which states that
\begin{equation*}
\mathbb{P}(|\Omega|\le m)\le e^{-\epsilon^2n^2/2d},
\end{equation*}
if taking $m=d/2$, $\epsilon=d/2n$ and $d\ge c_1\log n$ for an appropriate constant $c_1$, we have
\begin{equation}
\label{equ: up}
\mathbb{P}(|\Omega|\le d/2)\le e^{-d/4}\le n^{-10}.
\end{equation}

In the other direction, by scalar Chernoff bound again which states that
\begin{equation*}
\mathbb{P}(|\Omega|\ge m)\le e^{-\epsilon^2n^2/3d},
\end{equation*}
if taking $m=2d$, $\epsilon=-d/n$ and $d\ge c_2\log n$ for an appropriate constant $c_2$, we obtain
\begin{equation}
\label{equ: lower}
\mathbb{P}(|\Omega|\ge 2d)\le e^{-d/3}\le n^{-10}.
\end{equation}

Finally, according to \eqref{equ: up} and \eqref{equ: lower}, we conclude that $d/2<|\Omega|<2d$ with an overwhelming probability, provided that $d\ge c\log n$ for some constant $c$.
\end{proof}

By Theorems \ref{theorem: exact recovery under Bernoulli sampling} and \ref{theorem: complete bases of low-rank part} and Lemma \ref{lemma: number of sampling by Bernoulli}, the following theorem justifies the success of Line \ref{alg: 6} in Algorithm \ref{algl_21 filtering}.
\begin{theorem}[Exact Recovery of Ground Truth Subspace from Seed Matrix]
\label{theorem: success of steps 2}
Suppose that all the conditions in Theorem \ref{theorem: exact recovery under Bernoulli sampling} are fulfilled for the pair $([L_0]_{l},[S_0]_{l})$. Then Line \ref{alg: 6} of Algorithm \ref{algl_21 filtering} exactly recovers the column space of the incoherent $L_0$ and the column support of $[S_0]_{l}$ with an overwhelming probability $1-cn^{-10}$, provided that
$
d\ge C_0\mu r\log^3 n,
$
where $c$ and $C_0$ are numerical constants, and $\mu$ is the incoherence parameter on the row space of matrix $L_0$.
\end{theorem}
%\begin{proof}
%We decompose $M_{l}$ as $M_{l}=\mathcal{P}_{\mathcal{I}_0^\perp}M_{l}+\mathcal{P}_{\mathcal{I}_0}M_{l}\in\mathbb{R}^{m\times k}$. By Theorem \ref{theorem: complete bases of low-rank part}, $\mbox{Range}(\mathcal{P}_{\mathcal{I}_0^\perp}(M_{l}))=\mbox{Range}(L_0)$ whose dimension is $r$. Thus to use Theorem \ref{theorem: exact recovery under Bernoulli sampling}, we need to show that $r$ is relatively small and $\mathcal{P}_{\mathcal{I}_0}M_{l}$ has sparse columns.
%
%According to Lemma \ref{lemma: number of sampling by Bernoulli}, Line \ref{alg: 1} requires $k=c_0d$ for a constant $c$ with an overwhelming probability. Thus after Line \ref{alg: 5}, we must have
%\begin{equation}
%r\le\max\left\{\rho_r\frac{cr(\log n)^3}{\mu(\log (r(\log n)^3))^3},\rho_r\frac{c_0d_0}{\mu(\log c_0d_0)^3}\right\}=\rho_r\frac{k}{\mu(\log k)^3},
%\end{equation}
%where the first inequality holds due to $r\le\rho_rn/\mu(\log n)^3$ by assumption. Furthermore, conditioning on the column index set $l$, the column support of $\mathcal{P}_{\mathcal{I}_0}M_{l}$ subjects to i.i.d. $\mbox{Ber}(a)$, where $a\le\rho_a$ by assumption. Applying Theorem \ref{theorem: exact recovery under Bernoulli sampling} to $M_{l}$, we obtain the result.
%\end{proof}

\subsubsection{Analysis for $\ell_{2,1}$ Filtering}
To justify the outlier identifiability of model \eqref{equleast-square problem 2}, it suffices to show that $\mbox{Range}(\mathcal{Y}_i'(L_{l}))$ is complete, i.e., $\mbox{Range}(\mathcal{Y}_i'(L_{l}))=\mbox{Range}(\mathcal{Y}_i'(L_0))$.
Actually, this can be proved by the following theorem:
\begin{theorem}[Outlier Identifiability of $\ell_{2,1}$ Filtering]
\label{theorem: complete bases of elements observation}
Suppose that each row of $L_{l}$ is sampled i.i.d. by Bernoulli distribution with parameter $p_0$. Let $\mathcal{Y}_i'(L_{l})$ be the selected rows from matrix $L_l$, i.e., $\mathcal{Y}_i'(L_{l})=\sum_j\delta_jL_{l_{j:}} e_j$, where $\delta_j\sim \mbox{Ber}(p_0)$. Then with probability at least $1-\delta$, we have rank$(\mathcal{Y}_i'(L_{l}))=r$, or equivalently $\mbox{Range}(\mathcal{Y}_i'(L_{l}))=\mbox{Range}(\mathcal{Y}_i'(L_0))$, provided that
$
p_0\ge 2\mu\frac{r}{m}\log\frac{r}{\delta},
$
where $\mu$ is the incoherence parameter on the column space of matrix $L_0$.
\end{theorem}
\begin{proof}
The proof is similar to that of Theorem \ref{theorem: complete bases of low-rank part}, where we use a property that $\mu(L_l)=\mu(L_0)$ since $\text{Range}(L_l)=\text{Range}(L_0)$, by Theorem \ref{theorem: complete bases of low-rank part}.
\end{proof}
It is worth noting that, when the matrix is fully observed, model \eqref{equleast-square problem 2} exactly identifies the outliers even without Theorem \ref{theorem: complete bases of elements observation}.

\subsection{Complexity Analysis}
In this section, we consider the time complexity of our randomized $\ell_{2,1}$ filtering algorithm.
We analyze our algorithm in the case where $d=\Theta(r\log^3 n)$. In Algorithm \ref{algl_21 filtering}, Line \ref{alg: 1} requires $O(n)$ time. For Line \ref{alg: 6} which recovers $m\times cr(\log n)^3$ seed matrix, this step requires $O(r^2m\log^6 n)$ time. Line \ref{alg: 8} requires at most $6r^2m\log^3 n$ time due to the QR factorization~\cite{golub2012matrix}, and Line \ref{alg: 9} needs $(2r+1)m$ time due to matrix-matrix multiplication. Thus the \emph{overall} complexity of
our $\ell_{2,1}$ filtering algorithm is at most $O(r^2m\log^6 n)+(2r+1)mn+6r^2mn\log^3 n\approx6r^2mn\log^3 n$. As ADMM algorithm requires
$O(mn\min\{m,n\})$ time to run for our model due to SVD or matrix-matrix
multiplication \emph{in every iteration}, and require many iterations in order to converge, our algorithm is significantly
faster than the state-of-the-art methods.

%For the more general case where
%\begin{equation}
%d=20 n\log n\ f(L_0),
%\end{equation}
%by Lemma \ref{lemma: number of sampling by Bernoulli}, we have
%\begin{equation}
%k=\max\left\{C_0n\log n\ f(L_0),cr\log^3 n\right\}
%\end{equation}
%for constants $c$ and $C_0$ at an overwhelming probability. So the overall complexity of our $\ell_{2,1}$ filtering algorithm is at most
%\begin{equation}
%O(mk^2)+O(rmn)+6rmnk.
%\end{equation}

\section{Applications and Experiments}
\label{section: applications and experiments}

As we discuss in Section \ref{section: introduction}, our model and algorithm have various applications. To show that, this section first relates our model to the subspace clustering task with missing values, and then demonstrates the validity of our theory and applications by synthetic and real experiments.
\subsection{Applications to Robust Subspace Clustering with Missing Values}
Subspace clustering aims at clustering data according to the subspaces they lie in. It is well known that many datasets, e.g., face~\cite{LiuG1} and motion~\cite{Gear,Yan,Rao}, can be well separated by their different subspaces. So subspace clustering has been successfully applied to face recognition~\cite{LiuG1}, motion segmentation~\cite{Elhamifar}, etc.

Probably one of the most effective subspace clustering models is robust LRR~\cite{Favaro,Wei}. Suppose that the data matrix $M$ contains columns that are from a union of independent subspaces with outliers. The idea of robust LRR is to self-express the data, namely, using the clean data themselves as the dictionary, and then find the representation matrix with the lowest rank. Mathematically, it is formulated as
\begin{equation}
\label{equ: R-LRR}
\min_{Z,L,S} \|Z\|_*+\lambda\|S\|_{2,1},\quad \mbox{s.t.}\quad L=LZ,\quad M=L+S.
\end{equation}
After obtaining the optimal solution $Z^*$, we can apply spectral clustering algorithms, such as Normalized Cut, to cluster each data point
according to the subspaces they lie in.

Although robust LRR \eqref{equ: R-LRR} has been widely applied to many computer vision tasks~\cite{Favaro,Wei}, it cannot handle missing values, i.e., only a few entries of $M$ are observed. Such a situation commonly occurs because of sensor failures, uncontrolled environments, etc. To resolve the issue, in this paper we extend robust LRR by slightly modifying the second constraint:
\begin{equation}
\label{equ: MD-R-LRR}
\begin{split}
\min_{Z,L,S} \|Z\|_*+\lambda\|S\|_{2,1},\ \ \mbox{s.t.}\quad L=LZ,\ \langle M,\omega_{ij}\rangle=\langle L+S,\omega_{ij}\rangle,\ (i,j)\in\mathcal{K}_{obs}.
\end{split}
\end{equation}
A similar model has been proposed by Shi et al.~\cite{shi2014low}, which is
\begin{equation}
\label{equ: LRR for incomplete data}
\begin{split}
\min_{Z,D,S} \|Z\|_*+\lambda\|S\|_{2,1},\ \ \mbox{s.t.}\quad D=DZ+S,\ \langle M,e_ie_j^*\rangle=\langle D,e_ie_j^*\rangle,\ (i,j)\in\mathcal{K}_{obs}.
\end{split}
\end{equation}
However, there are two main differences between their model and ours: 1. Our model does not require $\omega_{ij}$ to be standard basis, thus being more general; 2. Unlike \eqref{equ: LRR for incomplete data}, we use clean data as the dictionary to represent themselves. Such a modification robustifies the model significantly, as discussed in \cite{zhang2015relations,Wei}.

The extended robust LRR \eqref{equ: MD-R-LRR} is NP-hard due to its non-convexity, which incurs great difficulty in efficient solution. As an application of this paper, we show that the solutions to \eqref{equ: MD-R-LRR} and to \eqref{equ: original robust MC wrt any basis 1} are mutually expressible in closed forms:
\begin{claim}
\label{theorem: relations}
The pair $(L^*(L^*)^\dag,L^*,S^*)$ is optimal to the extended robust LRR problem \eqref{equ: MD-R-LRR}, if $(L^*,S^*)$ is a solution to the extended robust MC problem \eqref{equ: original robust MC wrt any basis 1}. Conversely, suppose that $(Z^*,L^*,S^*)$ is a solution to the extended robust LRR problem \eqref{equ: MD-R-LRR}, then $(L^*,S^*)$ will be optimal to the extended robust MC problem \eqref{equ: original robust MC wrt any basis 1}.
\end{claim}
\begin{proof}
The proof can be found in the Appendices.
\end{proof}
Using relaxed form \eqref{equ: robust MC wrt any basis 1} to well approximate original problem \eqref{equ: original robust MC wrt any basis 1} according to Theorem \ref{theorem: exact recovery under Bernoulli sampling}, and then applying Claim \ref{theorem: relations} to obtain a solution to the extended robust LRR problem model \eqref{equ: MD-R-LRR}, we are able to robustly cluster subspaces even though a constant fraction of values are unobserved. This is true once the conditions in Theorem \ref{theorem: exact recovery under Bernoulli sampling} can be satisfied:
\begin{itemize}
\item
The low-rankness condition holds if the sum of the subspaces is low-dimensional (less than $O(n/\log^3 n)$);
\item
The incoherence condition holds if the number of the subspaces is not too large (less than an absolute constant)~\cite{liu2014recovery}.
\end{itemize}
The computational cost can be further cut by applying our $\ell_{2,1}$ filtering approach, i.e., Algorithm \ref{algl_21 filtering}.

\begin{remark}
Intuitively, Claim \ref{theorem: relations} is equivalent to a two-step procedure: First completing the data matrix and identifying the outlier by extended robust MC, and then clustering the data by LRR.
\end{remark}

\subsection{Simulations and Experiments}
In this section, we conduct a series of experiments to demonstrate the validity of our theorems, and show possible applications of our model and algorithm.

\subsubsection{Validity of Regularization Parameter}
We first verify the validity of our regularization
parameter $\lambda=1/\sqrt{\log n}$ by simulations. The toy data are designed as
follows. We compute $L_0=XY^T$ as a product of two $n\times r$
i.i.d. $\mathcal{N}(0,1)$ matrices. The non-zero columns of
$S_0$ are sampled by Bernoulli distribution with parameter $a$, whose entries obey i.i.d.
$\mathcal{N}(0,1)$. Finally, we construct our observation matrix as
$\mathcal{P}_{\Omega_{obs}}(L_0+S_0)$, where $\Omega_{obs}$ is the observed index selected by i.i.d. $\mbox{Ber}(p_0)$. We solve model \eqref{equ: robust MC wrt any basis} to obtain an optimal solution $(L^*,S^*)$, and then
compare it with $(L_0,S_0)$. The distance between the range spaces of $L^*$ and $L_0$
is defined by $\|\mathcal{P}_{\mathcal{U}^*}-\mathcal{P}_{\mathcal{U}_0}\|_F$ and the distance between the column supports of $S^*$ and $S_0$
is given by the Hamming distance. The experiment is run by 10 times and we report the average outputs. Table \ref{table: Exact Recovery for Random
Problem of Varying Size} illustrates that our choice of the regularization parameter enables
model \eqref{equ: robust MC wrt any basis} to exactly recover the range space of
$L_0$ and the column support of $S_0$ at a high probability.

\begin{table}[ht]
\caption{Exact recovery on problems with different sizes. Here
rank$(L_0)=0.05n$, $a=0.1$, $p_0=0.8$, and
$\lambda=1/\sqrt{\log n}$.} \label{table: Exact Recovery
for Random Problem of Varying Size}
\begin{center}
\begin{tabular}{c||c|c}
\hline
$n$ & dist(Range($L^*$),Range($L_0$)) & dist$(\mathcal{I}^*,\mathcal{I}_0)$\\
\hline\hline
100 & $1.09\times10^{-14}$ & 0\\
\hline
200 & $1.70\times10^{-14}$ & 0\\
\hline
500 & $4.02\times10^{-14}$ & 0\\
\hline
1,000 & $6.10\times10^{-13}$ & 0\\
\hline
\end{tabular}
\end{center}
\end{table}
\vspace{+0.5cm}
Theorem \ref{theorem: exact recovery under Bernoulli sampling} shows that the exact recoverability of model \eqref{equ: robust MC wrt any basis} is independent of the magnitudes of noises. To verify this, Table \ref{table: Exact Recovery
for Random Problem of Varying Noise Magnitudes} records the differences between the ground truth $(L_0,S_0)$ and the output $(L^*,S^*)$ of model \eqref{equ: robust MC wrt any basis} under varying noise magnitudes $\mathcal{N}(0,1/n)$, $\mathcal{N}(0,1)$, and $\mathcal{N}(0,n)$. It seems that our model always succeeds, no matter what magnitudes the noises are.

\begin{table}[ht]
\caption{Exact recovery on problems with different noise magnitudes. Here
$n=200$, rank$(L_0)=0.05n$, $a=0.1$, $p_0=0.8$, and
$\lambda=1/\sqrt{\log n}$.} \label{table: Exact Recovery
for Random Problem of Varying Noise Magnitudes}
\begin{center}
\begin{tabular}{c||c|c}
\hline
Magnitude & dist(Range($L^*$),Range($L_0$)) & dist$(\mathcal{I}^*,\mathcal{I}_0)$\\
\hline\hline
$\mathcal{N}(0,1/n)$ & $1.98\times10^{-14}$ & 0\\
\hline
$\mathcal{N}(0,1)$ & $1.50\times10^{-14}$ & 0\\
\hline
$\mathcal{N}(0,n)$ & $3.20\times10^{-14}$ & 0\\
\hline
\end{tabular}
\end{center}
\end{table}

\subsubsection{Exact Recovery from Varying Fractions of Corruptions and Observations}
We then test the exact recoverability of our model under varying fractions of corruptions and observations. The data are generated as the above-mentioned experiments, where the data size $n=200$. We repeat the experiments by decreasing the number of observations. Each simulation is run by 10 times, and Figure \ref{figure: tightness of
our bounds} plots the fraction of correct recoveries: white region represents the exact recovery in 10 experiments,
and black region denotes the failures in all of the experiments. It seems that model \eqref{equ: robust MC wrt any basis} succeeds even when the rank of intrinsic matrix is comparable to $O(n)$, which is consistent with our forecasted order $O(n/\log^3 n)$. But with the decreasing number of observations, the working range of model \eqref{equ: robust MC wrt any basis} shrinks.

\begin{figure}
\centering
\subfigure{
\includegraphics[width=0.2\textwidth]{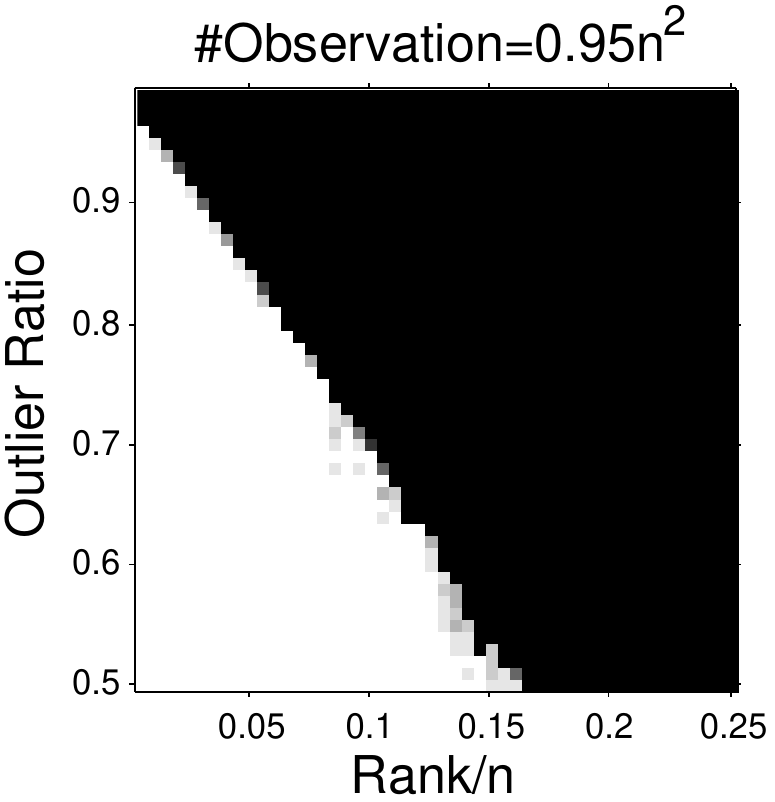}}
\subfigure{
\includegraphics[width=0.2\textwidth]{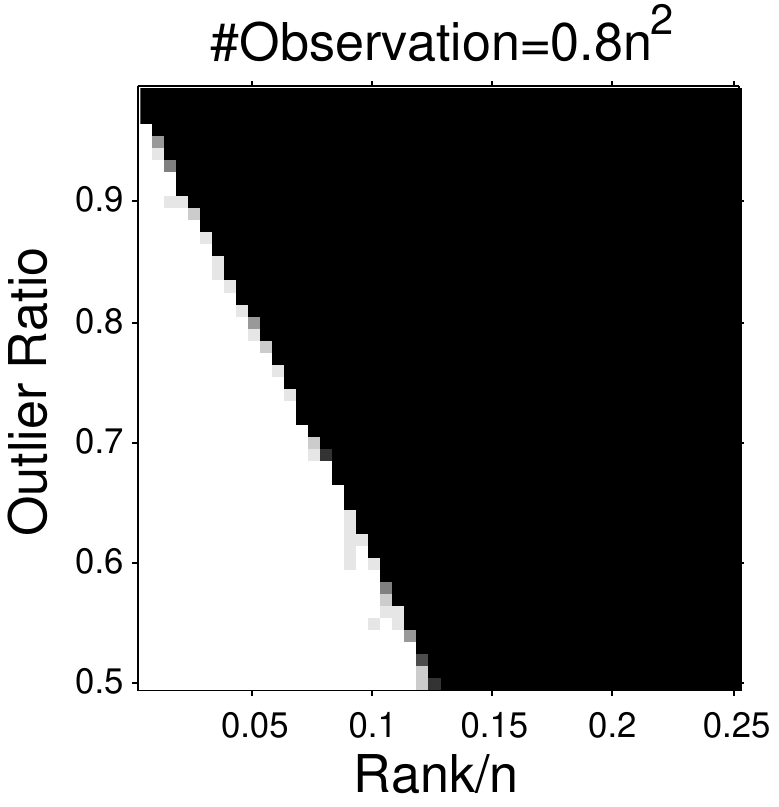}}\\
\subfigure{
\includegraphics[width=0.2\textwidth]{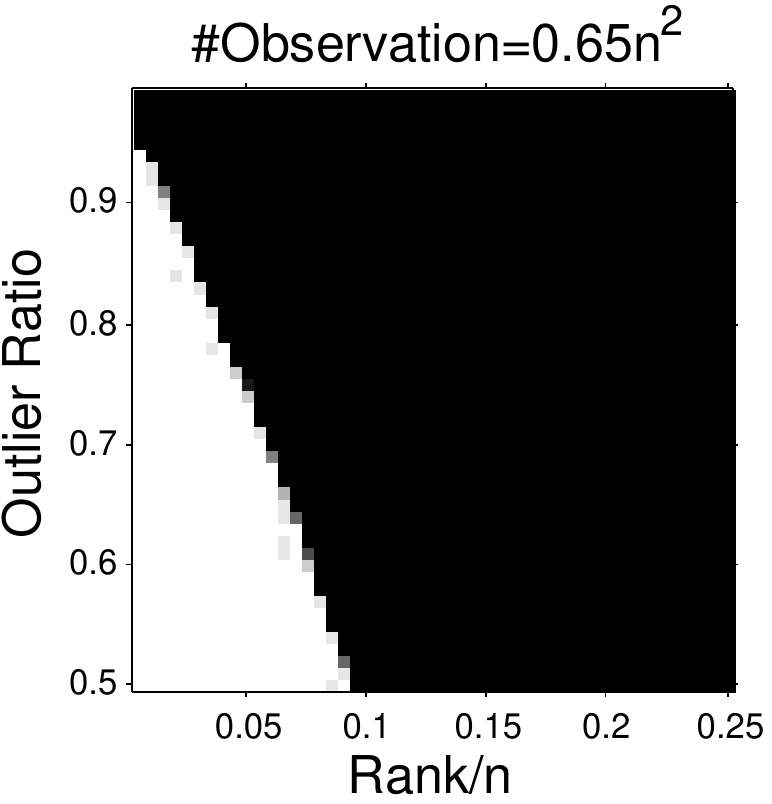}}
\subfigure{
\includegraphics[width=0.2\textwidth]{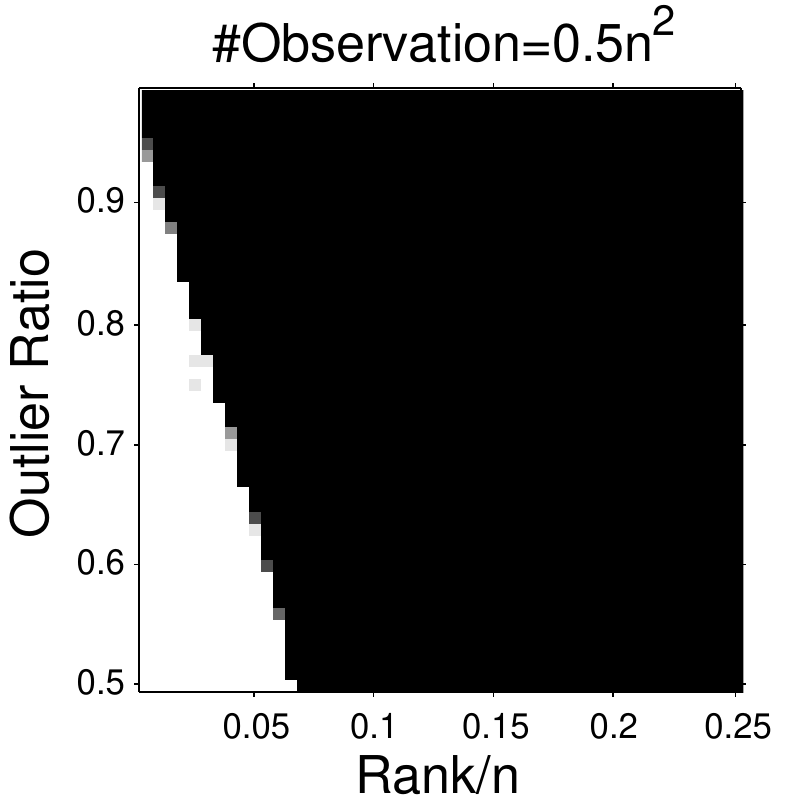}}
\caption{Exact recovery of the extended robust MC on random problems of
varying sizes. The white region represents the exact recovery in 10 experiments,
and black region denotes the failures in all of the experiments.} \label{figure: tightness of
our bounds}
\end{figure}

\subsubsection{Speed Advantage of $\ell_{2,1}$ Filtering Algorithm}
To test the speed advantage of our $\ell_{2,1}$ filtering algorithm, we compare the running time of ADMM and our filtering Algorithm \ref{algl_21 filtering} on the synthetic data. The data are generated as the above-mentioned simulations, where we change one variable among the set $(n, r, p_0, a)$ each time and fix others. Table \ref{table:comparison of speed} lists the CPU times, the distance between Range($L^*$) and Range($L_0$), and the Hamming distance between $\mathcal{I}^*$ and $\mathcal{I}_0$ by the two algorithms. It is easy to see that our $\ell_{2,1}$ filtering approach is significantly faster than ADMM under a comparable precision.
\begin{table*}
\caption{Comparison of the speed between ADMM and our $\ell_{2,1}$ filtering algorithm under varying parameter settings.} \label{table:comparison of speed}
\begin{center}
\begin{tabular}{c||c|c|c|c}
\hline
Parameter $(n, r, p_0, a)$  & Method  & Time (s) & dist(Range($L^*$),Range($L_0$)) & dist$(\mathcal{I}^*,\mathcal{I}_0)$\\
 \hline\hline
\multirow{2}{2.5cm}{(1,000, 1, 0.95, 0.1)} & ADMM & 102.89 & $2.34\times10^{-9}$ & 0\\
& $\ell_{2,1}$ Filtering & \textbf{5.86} & $5.41\times10^{-9}$ & 0\\
\hline
\multirow{2}{2.5cm}{(2,000, 1, 0.95, 0.1)} & ADMM & 587.53 & $6.03\times10^{-9}$ & 0\\
& $\ell_{2,1}$ Filtering & \textbf{35.49} & $9.91\times10^{-9}$ & 0\\
\hline
\multirow{2}{2.5cm}{(1,000, 10, 0.95, 0.1)} & ADMM & 104.12 & $2.71\times10^{-8}$ & 0\\
& $\ell_{2,1}$ Filtering & \textbf{25.75} & $4.06\times10^{-8}$ & 0\\
\hline
\multirow{2}{2.5cm}{(1,000, 1, 0.8, 0.1)} & ADMM & 100.10 & $2.76\times10^{-9}$ & 0\\
& $\ell_{2,1}$ Filtering & \textbf{4.33} & $6.17\times10^{-9}$ & 0\\
\hline
\multirow{2}{2.5cm}{(1,000, 1, 0.95, 0.2)} & ADMM & 92.95 & $4.01\times10^{-9}$ & 0\\
& $\ell_{2,1}$ Filtering & \textbf{5.09} & $8.34\times10^{-9}$ & 0\\
\hline
\end{tabular}
\end{center}
\end{table*}

\subsubsection{Applications to Subspace Clustering with Missing Coefficients}
To apply our model to the subspace clustering tasks with a fraction of missing values, we conduct experiments on the real Hopkins 155 database\footnote{\url{http://www.vision.jhu.edu/data/hopkins155}}. The Hopkins 155 database consists of 155 sequences, each of which contains multiple key points drawn from two or three motion objects. Because the motion trajectory of each rigid body lies in a single subspace, so we are able to cluster the points according to the subspaces they lie in. To make the problem more challenging, we randomly corrupt 5\% columns and remove 5\% observed coefficients. Table \ref{table:hopkins 155} lists the clustering accuracies of our algorithm on the first 5 sequences in comparison with other approaches~\cite{vidal2015sparse}. We can see that our approach always achieves high clustering accuracy, even though we cannot observe all of the data values. In addition, the experiments show that the Robust MC based methods are better than MC based methods. So our model is more robust.

%\begin{figure}
%\centering
%\subfigure{
%\includegraphics[width=0.2\textwidth]{hopkins1.jpg}}
%\subfigure{
%\includegraphics[width=0.2\textwidth]{hopkins2.jpg}}
%\label{figure: example of hopkins 155}
%\end{figure}

\begin{table*}
\caption{Clustering accuracies of our algorithm on the first 5 sequences in Hopkins 155 database, where there are 5\% missing entries.} \label{table:hopkins 155}
\begin{center}
\begin{tabular}{c||c|c|c|c|c}
\hline
\#Sequence & Motion Number & MC+SSC & MC+LRR & Robust MC+SSC & Robust MC+LRR (Ours)\\
\hline
\hline
\#1 & 2 & 60.48\% & 80.00\% &  82.38\% & \textbf{98.10\%}\\
\hline
\#2 & 3 & 67.43\% & 79.46\% & 73.24\% & \textbf{91.91\%}\\
\hline
\#3 & 2 & 75.16\% & 98.00\% & 83.01\% & \textbf{98.04\%}\\
\hline
\#4 & 2 & 73.09\% & 67.55\% &  86.54\% & \textbf{99.21\%}\\
\hline
\#5 & 2 & 75.69\% & 82.41\% &  75.69\% & \textbf{92.59\%}\\
\hline
\end{tabular}
\end{center}
\end{table*}

\section{Conclusions}
\label{section: conclusion}
In this paper, we investigate the theory, the algorithm, and the applications of our extended robust MC model. In particular, we study the exact recoverability of our model from few observed coefficients w.r.t. general basis, which partially covers the existing results as special cases. With slightly stronger incoherence (ambiguity) conditions, we are able to push the upper bound on the allowed rank from $O(1)$ to $O(n/\log^3 n)$, even when there are around a constant fraction of unobserved coefficients and column corruptions, where $n$ is the sample size. We further suggest a universal choice of the regularization parameter, which is $\lambda=1/\sqrt{\log n}$. This result waives the necessity of tuning regularization parameter, so it significantly extends the working range of robust MC. Moreover, we propose $\ell_{2,1}$ filtering algorithm so as to speed up solving our model numerically, and establish corresponding theoretical guarantees. As an application, we also relate our model to the subspace clustering tasks with missing values so that our theory and algorithm can be immediately applied to the subspace segmentation problem. Our experiments on the synthetic and real data testify to our theories.

\section{Acknowledgements}
Zhouchen Lin is supported by National Basic Research Program (973 Program) of China (grant no. 2015CB352502), National Natural Science Foundation (NSF) of China (grant nos. 61272341 and 61231002), and Microsoft Research Asia Collaborative Research Program. Chao Zhang is supported in part by 973 Program of China (grant nos. 2015CB352303 and 2011CB302400) and in part by NSF of China (grant nos. 61071156 and 61131003).

\bibliography{reference}
\bibliographystyle{IEEEtran}

\begin{appendix}
%\section{Appendices}
\label{Appendices}

\subsection{Preliminary Lemmas}
\label{subsection: preliminary}

We present several preliminary lemmas here which are critical for our proofs. For those readers who are interested in the main body of the proofs, please refer to Sections \ref{section: exact recovery of column support} and \ref{section: exact recovery of column space} directly.

\begin{lemma}
\label{lemma: inside Omega}
The optimal solution $(L^*, S^*)$ to the extended robust MC \eqref{equ: robust MC wrt any basis} satisfies $S^*\in\Omega_{obs}$.
\end{lemma}
\begin{proof}
Suppose that $S^*\not\in\Omega_{obs}$. We have $\|L^*\|_*+\|\mathcal{P}_{\Omega_{obs}}S^*\|_{2,1}<\|L^*\|_*+\|S^*\|_{2,1}$. Also, notice that the pair $(L^*, \mathcal{P}_{\Omega_{obs}}S^*)$ is feasible to problem \eqref{equ: robust MC wrt any basis}. Thus we have a contradiction to the optimality of $(L^*, S^*)$.
\end{proof}

\begin{lemma}[Elimination Lemma on Observed Elements]
\label{theorem: elimination on observed elements}
Suppose that any solution $(L^*,S^*)$ to the extended robust MC \eqref{equ: robust MC wrt any basis} with observation set $\mathcal{K}_{obs}$ exactly recovers the column space of $L_0$ and the column support of $S_0$, i.e., $\mbox{Range}(L^*)=\mbox{Range}(L_0)$ and $\{j: S^*_{:j}\not\in\mbox{Range}(L^*)\}=\mathcal{I}_0$. Then any solution $(L'^*,S'^*)$ to \eqref{equ: robust MC wrt any basis} with observation set $\mathcal{K}'_{obs}$ succeeds as well, where $\mathcal{K}_{obs}\subseteq\mathcal{K}'_{obs}$.
\end{lemma}
\begin{proof}
The conclusion holds because the constraints in problem \eqref{equ: robust MC wrt any basis} with observation set $\mathcal{K}_{obs}'$ are stronger than the constraints in problem \eqref{equ: robust MC wrt any basis} with observation set $\mathcal{K}_{obs}$.
\end{proof}

\begin{lemma}[Elimination Lemma on Column Support]
\label{theorem: elimination on column support}
Suppose that any solution $(L^*,S^*)$ to the extended robust MC \eqref{equ: robust MC wrt any basis} with input $\mathcal{R}(M)=\mathcal{R}(L^*)+\mathcal{R}(S^*)$ exactly recovers the column space of $L_0$ and the column support of $S_0$, i.e., $\mbox{Range}(L^*)=\mbox{Range}(L_0)$ and $\{j: S^*_{:j}\not\in\mbox{Range}(L^*)\}=\mathcal{I}_0$. Then any solution $(L'^*,S'^*)$ to \eqref{equ: robust MC wrt any basis} with input $\mathcal{R}(M')=\mathcal{R}(L^*)+\mathcal{R}\mathcal{P}_\mathcal{I}(S^*)$ succeeds as well, where $\mathcal{I}\subseteq\mathcal{I}^*=\mathcal{I}_0$.
\end{lemma}
\begin{proof}
Since $(L'^*,S'^*)$ is the solution of \eqref{equ: robust MC wrt any basis} with input matrix $\mathcal{P}_{\Omega_{obs}}M'$, we have
\begin{equation*}
\|L'^*\|_*+\lambda\|S'^*\|_{2,1}\le\|L^*\|_*+\lambda\|\mathcal{P}_\mathcal{I}S^*\|_{2,1}.
\end{equation*}
Therefore
\begin{equation*}
\begin{split}
&\ \ \ \ \ \|L'^*\|_*+\lambda\|S'^*+\mathcal{P}_{{\mathcal{I}^\perp}\cap\mathcal{I}_0}S^*\|_{2,1}\\
&\le\|L'^*\|_*+\lambda\|S'^*\|_{2,1}+\lambda\|\mathcal{P}_{{\mathcal{I}^\perp}\cap\mathcal{I}_0}S^*\|_{2,1}\\
&\le\|L^*\|_*+\lambda\|\mathcal{P}_{\mathcal{I}}S^*\|_{2,1}+\lambda\|\mathcal{P}_{{\mathcal{I}^\perp}\cap\mathcal{I}_0}S^*\|_{2,1}\\
&=\|L^*\|_*+\lambda\|S^*\|_{2,1}.
\end{split}
\end{equation*}
Note that
\begin{equation*}
\mathcal{R}(L'^*+S'^*+\mathcal{P}_{{\mathcal{I}^\perp}\cap\mathcal{I}_0}S^*)=\mathcal{R}(M'+\mathcal{P}_{{\mathcal{I}^\perp}\cap\mathcal{I}_0}S^*)=\mathcal{R}(M).
\end{equation*}
Thus $(L'^*,S'^*+\mathcal{P}_{{\mathcal{I}^\perp}\cap\mathcal{I}_0}S^*)$ is optimal to problem with input $\mathcal{P}_{\Omega_{obs}}M$ and by assumption we have
\begin{equation*}
\mbox{Range}(L'^*)=\mbox{Range}(L^*)=\mbox{Range}(L_0),
\end{equation*}
\begin{equation*}
\{j: [S'^*+\mathcal{P}_{{\mathcal{I}^\perp}\cap\mathcal{I}_0}S^*]_{:j}\not\in\mbox{Range}(L_0)\}=\mbox{Supp}(S_0).
\end{equation*}
The second equation implies $\mathcal{I}\subseteq\{j: S'^*_{:j}\not\in\mbox{Range}(L_0)\}$. Suppose that $\mathcal{I}\not=\{j: S'^*_{:j}\not\in\mbox{Range}(L_0)\}$. Then there exists an index $k$ such that $S'^*_{:k}\not\in\mbox{Range}(L_0)$ and $k\not\in\mathcal{I}$, i.e., $M'_{:k}=L^*_{:k}\in\mbox{Range}(L_0)$. Note that $L'^*_{:j}\in\mbox{Range}(L_0)$. Thus $S'^*_{:k}\in\mbox{Range}(L_0)$ and we have a contradiction. Thus $\mathcal{I}=\{j: S'^*_{:j}\not\in\mbox{Range}(L_0)\}=\{j: S'^*_{:j}\not\in\mbox{Range}(L'^*)\}$ and the algorithm succeeds.
\end{proof}

\begin{lemma}[Matrix (Operator) Bernstein Inequality~\cite{Tropp}]
\label{theorem: Operator-Bernstein}
Let $X_i\in\mathbb{R}^{m\times n},\ i=1,...,s$, be independent, zero-mean, matrix-valued random variables. Assume that $V,L\in\mathbb{R}$ are such that $\max\left\{\left\Vert\sum_{i=1}^s\mathbb{E}\left[X_iX_i^*\right]\right\Vert,\left\Vert\sum_{i=1}^s\mathbb{E}\left[X_i^*X_i\right]\right\Vert\right\}\le M$ and $\|X_i\|\le L$. Then
\begin{equation*}
\mathbb{P}\left[\left\Vert\sum_{i=1}^sX_i\right\Vert>t\right]\le (m+n)\exp\left(-\frac{3t^2}{8M}\right)
\end{equation*}
for $t\le M/L$, and
\begin{equation*}
\mathbb{P}\left[\left\Vert\sum_{i=1}^sX_i\right\Vert>t\right]\le (m+n)\exp\left(-\frac{3t}{8L}\right)
\end{equation*}
for $t> M/L$.
\end{lemma}

Lemma \ref{theorem: elimination on column support} shows that the success of algorithm is monotone on $|\mathcal{I}_0|$.
Thus by standard arguments in \cite{Candes}, \cite{candes2006robust}, and \cite{Candes2010power}, any guarantee proved
for the Bernoulli distribution equivalently holds for the uniform
distribution.

\begin{lemma}
\label{lemma: Omega 2 norm}
For any $\mathcal{K}\sim\mbox{Ber}(p)$, with high probability,
\begin{equation*}
\left\Vert\mathcal{P}_{\mathcal{\widetilde{T}}}-p^{-1}\mathcal{P}_{\mathcal{\widetilde{T}}}\mathcal{R}'\mathcal{P}_{\mathcal{\widetilde{T}}}\right\Vert<\varepsilon\quad\mbox{and}\quad \left\Vert\mathcal{P}_{\mathcal{\widehat{V}}}-p^{-1}\mathcal{P}_{\mathcal{\widehat{V}}}\mathcal{R}'\mathcal{P}_{\mathcal{\widehat{V}}}\right\Vert<\varepsilon,
\end{equation*}
provided that $p\ge C_0\varepsilon^{-2}(\mu r\log n_{(1)})/n_{(2)}$ for some numerical constant $C_0>0$, where $\mathcal{R}'(\cdot)=\sum_{ij\in\mathcal{K}}\langle\cdot,\omega_{ij}\rangle\omega_{ij}$.
\end{lemma}
\begin{proof}
The proof is in Appendix \ref{Sec: Proofs of Lemma}.
\end{proof}

\begin{corollary}[\cite{Candes}]
Assume that $\mathcal{K}_{obs}\sim\mbox{Ber}(p_0)$. Then with an
overwhelming probability,
$\|\mathcal{P}_{\Omega_{obs}^\perp}\mathcal{P}_{\mathcal{\widetilde{T}}}\|^2\le \varepsilon+1-p_0$,
provided that $p_0\ge C_0\varepsilon^{-2}(\mu r\log n)/n$ for some
numerical constant $C_0>0$.
\end{corollary}

\begin{lemma}
\label{lemma: Omega infty infty norm}
Suppose that $Z\in\mathcal{\widetilde{T}}$ and $\mathcal{K}\sim\mbox{Ber}(p)$. Let $\mathcal{R'}(\cdot)=\sum_{ij\in\mathcal{K}}\langle\cdot, \omega_{ij}\rangle\omega_{ij}$. Then with high probability
\begin{equation*}
\max_{ab}\left|\langle Z-p^{-1}\mathcal{P}_{\mathcal{\widetilde{T}}}\mathcal{R'}Z,\omega_{ab}\rangle\right|<\varepsilon\max_{ab}\left|\langle Z,\omega_{ab}\rangle\right|,
\end{equation*}
provided that $p\ge C_0\varepsilon^{-2}(\mu r\log n_{(1)})/n_{(2)}$ for some numerical constant $C_0>0$.
\end{lemma}
\begin{proof}
The proof is in Appendix \ref{Sec: Proofs of Lemma 12}.
\end{proof}

\begin{lemma}
\label{lemma: Omega 2 infty norm}
Suppose that $Z$ is a fixed matrix and $\mathcal{K}\sim\mbox{Ber}(p)$. Let $\mathcal{R'}(\cdot)=\sum_{ij\in\mathcal{K}}\langle\cdot, \omega_{ij}\rangle\omega_{ij}$. Then with high probability
\begin{equation*}
\|Z-p^{-1}\mathcal{R'}Z\|<C_0'\sqrt{\frac{n_{(1)}\log n_{(1)}}{p}}\max_{ij}\left|\langle Z,\omega_{ij}\rangle\right|,
\end{equation*}
provided that $p\ge C_0'(\mu\log n_{(1)})/n_{(1)}$ for some small numerical constant $C_0'>0$.
\end{lemma}
\begin{proof}
The proof is in Appendix \ref{Sec: Proofs of Lemma 13}.
\end{proof}

\begin{lemma}
\label{lemma: relation 2 norm}
Let $\mathcal{R'}$ be the projection operator onto space $\Omega=\mbox{Span}\{\omega_{ij},\ i,j\in\mathcal{K}\}$ with any $\mathcal{K}$, the space $\mathcal{I}=\mbox{Span}\{\omega_{ij},\ j\in\mathcal{J}\}$, and $\Psi=\mathcal{I}\cap\Omega$. Let $\mathcal{J}\sim\mbox{Ber}(a)$. Then with high probability
\begin{equation*}
\begin{split}
\left\Vert a^{-1}\mathcal{P}_\mathcal{\widehat{V}}\mathcal{R'}\mathcal{P}_\mathcal{I}\mathcal{R'}\mathcal{P}_\mathcal{\widehat{V}}-\mathcal{P}_\mathcal{\widehat{V}}\mathcal{R'}\mathcal{P}_\mathcal{\widehat{V}}\right\Vert=\left\Vert a^{-1}\mathcal{P}_\mathcal{\widehat{V}}\mathcal{P}_\Psi\mathcal{P}_\mathcal{\widehat{V}}-\mathcal{P}_\mathcal{\widehat{V}}\mathcal{R'}\mathcal{P}_\mathcal{\widehat{V}}\right\Vert<\varepsilon,
\end{split}
\end{equation*}
provided that $a\ge C_0\varepsilon^{-2}(\mu r\log n_{(1)})/n$ for some numerical constant $C_0>0$.
\end{lemma}
\begin{proof}
The proof is in Appendix \ref{Sec: Proofs of Lemma 14}.
\end{proof}

\begin{corollary}
\label{corollary: sum 2 norm}
Assume that $\Psi=\mathcal{I}\cap\Omega$. Then for any $\mathcal{I}\sim\mbox{Ber}(a)$ and $\Omega\sim\mbox{Ber}(p)$, with high probability
\begin{equation*}
\|(pa)^{-1}\mathcal{P}_\mathcal{\widehat{V}}\mathcal{P}_\Psi\mathcal{P}_\mathcal{\widehat{V}}-\mathcal{P}_\mathcal{\widehat{V}}\|<(p^{-1}+1)\varepsilon,
\end{equation*}
provided that $a,p\ge C_0\varepsilon^{-2}(\mu r\log n_{(1)})/n$ for some numerical constant $C_0>0$.
\end{corollary}
\begin{proof}
By Lemma \ref{lemma: Omega 2 norm} and Lemma \ref{lemma: relation 2 norm}, we have
\begin{equation*}
\|\mathcal{P}_{\mathcal{\widehat{V}}}\mathcal{P}_{\Omega}\mathcal{P}_{\mathcal{\widehat{V}}}-p\mathcal{P}_{\mathcal{\widehat{V}}}\|<p\varepsilon,
\end{equation*}
and
\begin{equation*}
\|a^{-1}\mathcal{P}_\mathcal{\widehat{V}}\mathcal{P}_\Gamma\mathcal{P}_\mathcal{\widehat{V}}-\mathcal{P}_\mathcal{\widehat{V}}\mathcal{P}_\Omega\mathcal{P}_\mathcal{\widehat{V}}\|<\varepsilon.
\end{equation*}
So by triangle inequality, we have
\begin{equation*}
\begin{split}
&\ \ \ \ \|a^{-1}\mathcal{P}_\mathcal{\widehat{V}}\mathcal{P}_\Gamma\mathcal{P}_\mathcal{\widehat{V}}-p\mathcal{P}_\mathcal{\widehat{V}}\|\\&\le\|\mathcal{P}_{\mathcal{\widehat{V}}}\mathcal{P}_{\Omega}\mathcal{P}_{\mathcal{\widehat{V}}}-p\mathcal{P}_{\mathcal{\widehat{V}}}\|+\|a^{-1}\mathcal{P}_\mathcal{\widehat{V}}\mathcal{P}_\Gamma\mathcal{P}_\mathcal{\widehat{V}}-\mathcal{P}_\mathcal{\widehat{V}}\mathcal{P}_\Omega\mathcal{P}_\mathcal{\widehat{V}}\|\\
&<(p+1)\varepsilon.
\end{split}
\end{equation*}
That is
\begin{equation*}
\|(pa)^{-1}\mathcal{P}_\mathcal{\widehat{V}}\mathcal{P}_\Gamma\mathcal{P}_\mathcal{\widehat{V}}-\mathcal{P}_\mathcal{\widehat{V}}\|<(p^{-1}+1)\varepsilon.
\end{equation*}
\end{proof}

\begin{corollary}
\label{corollary: gamma v bound}
Let $\Pi=\mathcal{I}_0\cap\Omega_{obs}$, where $\mathcal{I}_0\sim\mbox{Ber}(p_1)$. Then with an
overwhelming probability
$\|\mathcal{P}_\Pi\mathcal{P}_{\mathcal{\widehat{V}}}\|^2\le(1-p_1)\varepsilon+p_1$,
provided that $1-p_1\ge C_0\varepsilon^{-2}(\mu r\log n_{(1)})/n$ for some
numerical constant $C_0>0$.
\end{corollary}

\begin{proof}
Let $\Gamma=\mathcal{I}_0^\perp\cap\Omega_{obs}$. Note that $\mathcal{I}_0^\perp\sim\mbox{Ber}(1-p_1)$. By Lemma \ref{lemma: relation 2 norm}, we have
$\|(1-p_1)^{-1}\mathcal{P}_\mathcal{\widehat{V}}\mathcal{P}_\Gamma\mathcal{P}_\mathcal{\widehat{V}}-\mathcal{P}_\mathcal{\widehat{V}}\mathcal{P}_{\Omega_{obs}}\mathcal{P}_\mathcal{\widehat{V}}\|<\varepsilon$,
or equivalently
\begin{equation*}
\begin{split}
&\ \ \ \
\|(1-p_1)^{-1}\mathcal{P}_\mathcal{\widehat{V}}\mathcal{P}_\Gamma\mathcal{P}_\mathcal{\widehat{V}}-\mathcal{P}_\mathcal{\widehat{V}}\mathcal{P}_{\Omega_{obs}}\mathcal{P}_\mathcal{\widehat{V}}\|\\
&=(1-p_1)^{-1}\|\mathcal{P}_\mathcal{\widehat{V}}\mathcal{P}_\Gamma\mathcal{P}_\mathcal{\widehat{V}}-(1-p_1)\mathcal{P}_\mathcal{\widehat{V}}\mathcal{P}_{\Omega_{obs}}\mathcal{P}_\mathcal{\widehat{V}}\|\\
&=(1-p_1)^{-1}\|\mathcal{P}_\mathcal{\widehat{V}}\mathcal{P}_{\Omega_{obs}}\mathcal{P}_\mathcal{\widehat{V}}-\mathcal{P}_\mathcal{\widehat{V}}\mathcal{P}_{(\mathcal{I}_0^\perp\cap\Omega_{obs})}\mathcal{P}_\mathcal{\widehat{V}}-p_1\mathcal{P}_\mathcal{\widehat{V}}\mathcal{P}_{\Omega_{obs}}\mathcal{P}_\mathcal{\widehat{V}}\|\\
&=(1-p_1)^{-1}\|\mathcal{P}_\mathcal{\widehat{V}}\mathcal{P}_{(\mathcal{I}_0\cap\Omega_{obs})}\mathcal{P}_\mathcal{\widehat{V}}-p_1\mathcal{P}_\mathcal{\widehat{V}}\mathcal{P}_{\Omega_{obs}}\mathcal{P}_\mathcal{\widehat{V}}\|\\
&=(1-p_1)^{-1}\|\mathcal{P}_\mathcal{\widehat{V}}\mathcal{P}_\Pi\mathcal{P}_\mathcal{\widehat{V}}-p_1\mathcal{P}_\mathcal{\widehat{V}}\mathcal{P}_{\Omega_{obs}}\mathcal{P}_\mathcal{\widehat{V}}\|\\
&<\varepsilon.
\end{split}
\end{equation*}
Therefore, by the triangle inequality
\begin{equation*}
\begin{split}
&\ \ \ \ \|\mathcal{P}_\Pi\mathcal{P}_{\mathcal{\widehat{V}}}\|^2=\|\mathcal{P}_{\mathcal{\widehat{V}}}\mathcal{P}_\Pi\mathcal{P}_{\mathcal{\widehat{V}}}\|\le\|\mathcal{P}_\mathcal{\widehat{V}}\mathcal{P}_\Pi\mathcal{P}_\mathcal{\widehat{V}}-p_1\mathcal{P}_\mathcal{\widehat{V}}\mathcal{P}_{\Omega_{obs}}\mathcal{P}_\mathcal{\widehat{V}}\|+p_1\|\mathcal{P}_\mathcal{\widehat{V}}\mathcal{P}_{\Omega_{obs}}\mathcal{P}_\mathcal{\widehat{V}}\|\le(1-p_1)\varepsilon+p_1.
\end{split}
\end{equation*}
\end{proof}

%\begin{lemma}
%Assume that $\mathcal{\widehat{I}}\sim\mbox{Ber}(p)$. Then with an
%overwhelming probability
%$\|\mathcal{P}_{\mathcal{\widehat{I}}}\mathcal{P}_{\mathcal{\widehat{V}}}\|^2\le(1-p)\varepsilon+p$,
%provided that $1-p\ge C_0\varepsilon^{-2}(\mu r\log n)/n$ for some
%numerical constant $C_0>0$.
%\end{lemma}

%\begin{lemma}
%\label{theorem: operator and infty norm}
%For any $\mathcal{I}\sim\mbox{Ber}(a)$, with a probability at least $1-n^{-\beta}$,\vspace{-0.2cm}
%\begin{equation}
%\|Z-a^{-1}\mathcal{P}_\mathcal{I}Z\|\le C_0\sqrt{\frac{\beta n\log n}{a^2}}\|Z\|_\infty,\vspace{-0.15cm}
%\end{equation}
%provided that $a\ge c_0\mu\log n/n$, where $Z$ is any fixed matrix and $C_0$, $c_0$ are numerical constants.
%\end{lemma}

\subsection{Proofs of Theorem \ref{theorem: complete bases of low-rank part}}
To prove Theorem \ref{theorem: complete bases of low-rank part}, the following matrix Chernoff Bound is invoked in our proof:
\begin{theorem}[Matrix Chernoff Bound~\cite{gittens2011tail}]
\label{theorem: matrix chernoff bound}
Consider a finite sequence $\{X_k\}\in\mathbb{R}^{d\times d}$ of independent, random, Hermitian matrices. Assume that
\begin{equation*}
0\le\lambda_{min}(X_k)\le\lambda_{max}(X_k)\le L.
\end{equation*}
Define $Y=\sum_k X_k$, and $\mu_r$ as the $r$th largest eigenvalue of the expectation $\mathbb{E}Y$, i.e., $\mu_r=\lambda_r(\mathbb{E}Y)$. Then
\begin{equation*}
\begin{split}
\mathbb{P}\left\{\lambda_r(Y)>(1-\epsilon)\mu_r\right\}&\ge 1-r\left[\frac{e^{-\epsilon}}{(1-\epsilon)^{1-\epsilon}}\right]^{\frac{\mu_r}{L}}\\&\ge 1-r e^{-\frac{\mu_r\epsilon^2}{2L}},
\end{split}
\end{equation*}
for $\epsilon\in[0,1)$.
\end{theorem}

\comment{
Now let
\begin{equation}
\label{equ: f(B)}
f(L)=\frac{\|L\|_{2,\infty}^2}{\sigma_r(L)^2}.
\end{equation}
As preliminary, we also need the following bounds on $f(L)$:
\begin{lemma}
\label{lemma: f(B)}
For $f(L)$ defined as Eqn. \eqref{equ: f(B)}, we have
\begin{equation}
\label{equ: bound for f(L)}
\frac{r}{n}\le f(L)\le\frac{\sigma_1^2(L)}{\sigma_r^2(L)}\max_{ij}\|\mathcal{P}_{\mathcal{V}_L}e_ie_j^*\|_F^2,
\end{equation}
where $r$ and $\mathcal{V}_L$ are the rank and the right singular space, respectively.
\end{lemma}
\begin{proof}
It can be seen that
\begin{equation*}
f(L)=\frac{\|L\|_{2,\infty}^2}{\sigma_r(L)^2}\ge\frac{\|L\|_F^2/n}{\|L\|_F^2/r}=\frac{r}{n}.
\end{equation*}

For the upper bound, we have
\begin{equation*}
\begin{split}
f(L)&=\frac{\|L\|_{2,\infty}^2}{\sigma_r(L)^2}=\frac{\|U_L\Sigma_LV_L^Te_{i_0}e_1^*\|_F^2}{\sigma_r(L)^2}\\&\le\frac{\|U_L\Sigma_L\|^2\|V_L^Te_{i_0}e_1^*\|_F^2}{\sigma_r(L)^2}\\&\le\frac{\sigma_1^2(L)}{\sigma_r^2(L)}\max_{ij}\|\mathcal{P}_{\mathcal{V}_L}e_ie_j^*\|_F^2,
\end{split}
\end{equation*}
where the $i_0$th column of $L$ has the maximal $\ell_2$ norm among all columns.
\end{proof}
}

\begin{lemma}
\label{lemma: orth rank}
Let $X=U\Sigma V^T$ be the skinny SVD of matrix $X$.
For any set of coordinates $\Omega$ and any matrix $X\in\mathbb{R}^{m\times n}$, we have
$\mbox{rank}(X_{\Omega:})=\mbox{rank}(U_{\Omega:})$
and
$\mbox{rank}(X_{:\Omega})=\mbox{rank}(V_{\Omega:})$.
\end{lemma}
\begin{proof}
On one hand,
\begin{equation*}
X_{\Omega:}=I_{\Omega:}X=I_{\Omega:}U\Sigma V^T=U_{\Omega:}\Sigma V^T.
\end{equation*}
So $\mbox{rank}(X_{\Omega:})\le \mbox{rank}(U_{\Omega:})$. On the other hand, we have
\begin{equation*}
X_{\Omega:}V\Sigma^{-1}=U_{\Omega:}.
\end{equation*}
Thus $\mbox{rank}(U_{\Omega:})\le \mbox{rank}(X_{\Omega:})$. So $\mbox{rank}(X_{\Omega:})=\mbox{rank}(U_{\Omega:})$.

The second part of the argument can be proved similarly. Indeed, $X_{:\Omega}=U\Sigma V^TI_{:\Omega}=U\Sigma[V^T]_{:\Omega}$ and $\Sigma^{-1}U^TX_{:\Omega}=[V^T]_{:\Omega}$. So $\mbox{rank}(X_{:\Omega})=\mbox{rank}([V^T]_{:\Omega})=\mbox{rank}(V_{\Omega:})$, as desired.
\end{proof}

Now we are ready to prove Theorem \ref{theorem: complete bases of low-rank part}.
\begin{proof}
We investigate the smallest sampling parameter $d$ such that the sampled columns from $L_0=\mathcal{P}_{\mathcal{I}_0^\perp}M$ exactly span $\mbox{Range}(L_0)$ with an overwhelming probability.

Denote by $L_0=U\Sigma V^T$ the skinny SVD of $L_0$. Let $X=\sum_i\delta_i[V^T]_{:i}e_i^*$ be the random sampling of columns from matrix $V^T$, where $\delta_i\sim\mbox{Ber}(d/n)$. Define a positive semi-definite matrix
\begin{equation*}
Y=XX^*=\sum_{i=1}^n\delta_i[V^T]_{:i}[V^T]_{:i}^*.
\end{equation*}
Obviously, $\sigma_r(X)^2=\lambda_r(Y)$. To invoke the matrix Chernoff bound, we need to estimate $L$ and $\mu_r$ in Theorem \ref{theorem: matrix chernoff bound}. Specifically, since
\begin{equation*}
\begin{split}
\mathbb{E}Y&=\sum_{i=1}^n\mathbb{E}\delta_i[V^T]_{:i}[V^T]_{:i}^*\\&=\frac{d}{n}\sum_{i=1}^n[V^T]_{:i}[V^T]_{:i}^*\\&=\frac{d}{n}[V^T][V^T]^*,
\end{split}
\end{equation*}
we have $\mu_r=\lambda_r(\mathbb{E}Y)=d/n$. Furthermore, we also have
\begin{equation*}
\begin{split}
\lambda_{max}(\delta_i[V^T]_{:i}[V^T]_{:i}^*)=\|\delta_i[V^T]_{:i}\|_2^2\le\max_i\|V_{i:}\|_{2,\infty}^2=\frac{\mu r}{n}\triangleq L.
\end{split}
\end{equation*}
By matrix Chernoff bound,
\begin{equation*}
\begin{split}
\mathbb{P}\left\{\sigma_r(X)>0\right\}&\ge 1-r e^{-\frac{\mu_r}{2L}}\\
&=1-re^{-d/(2\mu r)}\\
&\ge 1-\delta,
\end{split}
\end{equation*}
we obtain
\begin{equation*}
d\ge2\mu r\log\left(\frac{r}{\delta}\right).
\end{equation*}
Note that $\sigma_r(X)>0$ implies that $\mbox{rank}([V^T]_{:l})=\mbox{rank}([L_0]_{:l})=r$, where the equality holds due to Lemma \ref{lemma: orth rank}. Also, $\mbox{Range}([L_0]_{:l})\subseteq\mbox{Range}(L_0)$. Thus $\mbox{Range}([L_0]_{:l})=\mbox{Range}(L_0)$.
\end{proof}

\subsection{Proofs of Claim \ref{theorem: relations}}
To prove Theorem \ref{theorem: relations}, the following proposition is crucial throughout our proof.
\begin{proposition}
\label{prop:solution to relaxed noiseless LRR}
The solution to the optimization problem:
\begin{equation}
\label{equ:relaxed noiseless LRR} \min_Z \|Z\|_*, \ \ \mbox{s.t.}
\ \ L=LZ,
\end{equation}
is unique and given by $Z^*=V_LV_L^T$, where $U_L\Sigma_L V_L^T$
is the skinny SVD of $L$.
\end{proposition}
\begin{proof}
We only prove the former part of the theorem. The proofs for the latter part of the theorem are similar. Suppose that $(L^*,S^*)$ is a solution to problem \eqref{equ: original robust MC wrt any basis 1}, while $(L^*(L^*)^\dag,L^*,S^*)$ is not optimal to problem \eqref{equ: MD-R-LRR}. So there exists an optimal solution to \eqref{equ: MD-R-LRR}, termed $(Z_*,L_*,S_*)$, which is strictly better than $(L^*(L^*)^\dag,L^*,S^*)$. Namely,
\begin{equation*}
\|Z_*\|_*+\lambda\|S_*\|_{2,1}<\|L^*(L^*)^\dag\|_*+\lambda\|S^*\|_{2,1},
\end{equation*}
\begin{equation*}
L_*=L_*Z_*,\quad \mathcal{R}(M)=\mathcal{R}(L_*+S_*).
\end{equation*}
Fixing $L$ and $S$ as $L_*$ and $S_*$ in \eqref{equ: MD-R-LRR}, respectively, and by Proposition \ref{prop:solution to relaxed noiseless LRR}, we have
\begin{equation*}
\begin{split}
\|Z_*\|_*+\lambda\|S_*\|_{2,1}&=\|V_{L_*}V_{L_*}^T\|_*+\lambda\|S_*\|_{2,1}\\&=\mbox{rank}(L_*)+\lambda\|S_*\|_{2,1}.
\end{split}
\end{equation*}
Furthermore, by the property of Moore-Penrose pseudo-inverse,
\begin{equation*}
\|L^*(L^*)^\dag\|_*+\lambda\|S^*\|_{2,1}=\mbox{rank}(L^*)+\lambda\|S^*\|_{2,1}.
\end{equation*}
Thus
\begin{equation*}
\begin{split}
\mbox{rank}(L_*)+\lambda\|S_*\|_{2,1}<\mbox{rank}(L^*)+\lambda\|S^*\|_{2,1},\\ \ \ 
L_*=L_*Z_*,\quad \mathcal{R}(M)=\mathcal{R}(L_*+S_*),
\end{split}
\end{equation*}
which is contradictory to the optimality of $(L^*,S^*)$ to problem \eqref{equ: original robust MC wrt any basis 1}. So $(L^*(L^*)^\dag,L^*,S^*)$ is optimal to problem \eqref{equ: MD-R-LRR}.
\end{proof}

\subsection{Proofs of Lemma \ref{lemma: Omega 2 norm}}
\label{Sec: Proofs of Lemma}
Now we are prepared to prove Lemma \ref{lemma: Omega 2 norm}.
\begin{proof}
For any matrix $X$, we have
\begin{equation*}
\mathcal{P}_{\mathcal{X}}X=\sum_{ij}\langle \mathcal{P}_{\mathcal{X}}X,\omega_{ij}\rangle \omega_{ij},
\end{equation*}
where $\mathcal{X}$ is $\mathcal{\widehat{V}}$ or $\mathcal{\widetilde{T}}$.
Thus $\mathcal{R}'\mathcal{P}_{\mathcal{X}}X=\sum_{ij}\kappa_{ij}\langle \mathcal{P}_{\mathcal{X}}X,\omega_{ij}\rangle \omega_{ij}$, where $\kappa_{ij}$s are i.i.d. Bernoulli variables with parameter $p$. Then
\begin{equation*}
\begin{split}
\mathcal{P}_{\mathcal{X}}\mathcal{R}'\mathcal{P}_{\mathcal{X}}X=\sum_{ij}\kappa_{ij}\langle \mathcal{P}_{\mathcal{X}}X,\omega_{ij}\rangle \mathcal{P}_{\mathcal{X}}(\omega_{ij})=\sum_{ij}\kappa_{ij}\langle X,\mathcal{P}_{\mathcal{X}}(\omega_{ij})\rangle \mathcal{P}_{\mathcal{X}}(\omega_{ij}).
\end{split}
\end{equation*}
Namely, $\mathcal{P}_{\mathcal{X}}\mathcal{R}'\mathcal{P}_{\mathcal{X}}=\sum_{ij}\kappa_{ij}\mathcal{P}_{\mathcal{X}}(\omega_{ij})\otimes\mathcal{P}_{\mathcal{X}}(\omega_{ij})$. Similarly, $\mathcal{P}_{\mathcal{X}}=\sum_{ij}\mathcal{P}_{\mathcal{X}}(\omega_{ij})\otimes\mathcal{P}_{\mathcal{X}}(\omega_{ij})$. So we obtain
\begin{equation*}
\begin{split}
\left\Vert p^{-1}\mathcal{P}_{\mathcal{X}}\mathcal{R}'\mathcal{P}_{\mathcal{X}}-\mathcal{P}_{\mathcal{X}}\right\Vert&=\left\Vert\sum_{ij}(p^{-1}\kappa_{ij}-1)\mathcal{P}_{\mathcal{X}}(\omega_{ij})\otimes\mathcal{P}_{\mathcal{X}}(\omega_{ij})\right\Vert\\&\triangleq\left\Vert\sum_{ij}X_{ij}\right\Vert,
\end{split}
\end{equation*}
where $X_{ij}=(p^{-1}\kappa_{ij}-1)\mathcal{P}_{\mathcal{X}}(\omega_{ij})\otimes\mathcal{P}_{\mathcal{X}}(\omega_{ij})$ is a zero-mean random variable.

To use Lemma \ref{theorem: Operator-Bernstein}, we need to work out $M$ and $L$ therein. Note that
\begin{equation*}
\begin{split}
\|X_{ij}\|&=\|(p^{-1}\kappa_{ij}-1)\mathcal{P}_{\mathcal{X}}(\omega_{ij})\otimes\mathcal{P}_{\mathcal{X}}(\omega_{ij})\|\\
&\le|p^{-1}\kappa_{ij}-1\||\mathcal{P}_{\mathcal{X}}(\omega_{ij})\otimes\mathcal{P}_{\mathcal{X}}(\omega_{ij})\|\\
&\le\max\{p^{-1}-1,1\}\|\mathcal{P}_{\mathcal{X}}(\omega_{ij})\|_F^2\\
&\le \frac{c\mu r}{n_{(2)}p}\\
&\triangleq L.
\end{split}
\end{equation*}
Furthermore,
\begin{equation*}
\begin{split}
\left\Vert\sum_{ij}\mathbb{E}\left[X_{ij}X_{ij}^*\right]\right\Vert&=\left\Vert\sum_{ij}\mathbb{E}\left[X_{ij}^*X_{ij}\right]\right\Vert\\
&=\left\Vert\sum_{ij}\hspace{-0.1cm}\mathbb{E}\hspace{-0.1cm}\left(\frac{\kappa_{ij}-p}{p}\right)^2\hspace{-0.15cm}\left[\mathcal{P}_{\mathcal{X}}(\omega_{ij})\hspace{-0.05cm}\otimes\hspace{-0.05cm}\mathcal{P}_{\mathcal{X}}(\omega_{ij})\right]\hspace{-0.05cm}\left[\mathcal{P}_{\mathcal{X}}(\omega_{ij})\hspace{-0.05cm}\otimes\hspace{-0.05cm}\mathcal{P}_{\mathcal{X}}(\omega_{ij})\right]\right\Vert\\
&=(p^{-1}-1)\left\Vert\sum_{ij}\left\Vert\mathcal{P}_{\mathcal{X}}(\omega_{ij})\right\Vert_F^2\mathcal{P}_{\mathcal{X}}(\omega_{ij})\otimes\mathcal{P}_{\mathcal{X}}(\omega_{ij})\right\Vert\\
&\le \frac{c\mu r}{n_{(2)}p}\left\Vert\sum_{ij}\mathcal{P}_{\mathcal{X}}(\omega_{ij})\otimes\mathcal{P}_{\mathcal{X}}(\omega_{ij})\right\Vert\\
&=\frac{c\mu r}{n_{(2)}p}\left\Vert\mathcal{P}_{\mathcal{X}}\right\Vert\\
&=\frac{c\mu r}{n_{(2)}p}\triangleq M.
\end{split}
\end{equation*}
Since $M/c=1>\epsilon$, by Lemma \ref{theorem: Operator-Bernstein}, we have
\begin{equation*}
\begin{split}
&\ \ \ \ \mathbb{P}\{\|p^{-1}\mathcal{P}_{\mathcal{X}}\mathcal{R}'\mathcal{P}_{\mathcal{X}}-\mathcal{P}_{\mathcal{X}}\|<\epsilon\}\\&\le 2mn\exp\left(-\frac{3\epsilon^2}{8M}\right)\\
&=2mn\exp\left(-\frac{3\epsilon^2n_{(2)}p}{8c\mu r}\right)\\
&\triangleq2mn\exp\left(-\frac{C\epsilon^2n_{(2)}p}{\mu r}\right)\\
&\le2mn\exp\left(-CC_0\log n_{(1)}\right)\\
&=2n^{-CC_0+2},
\end{split}
\end{equation*}
where the second inequality holds once we have $p\ge C_0\epsilon^{-2}(\mu r\log n_{(1)})/n_{(2)}$. So the proof is completed.
\end{proof}

\subsection{Proofs of Lemma \ref{lemma: Omega infty infty norm}}
\label{Sec: Proofs of Lemma 12}
We proceed to prove Lemma \ref{lemma: Omega infty infty norm}.
\begin{proof}
From the definition of operator $R'$, we know that
\begin{equation*}
\mathcal{R'}(Z)=\sum_{ij\in\mathcal{K}}\langle Z,\omega_{ij}\rangle\omega_{ij}=\sum_{ij}\delta_{ij}\langle Z,\omega_{ij}\rangle\omega_{ij},
\end{equation*}
where $\delta_{ij}$s are i.i.d. Bernoulli variables with parameter $p$. Notice that $Z\in\mathcal{\widetilde{T}}$, so we have
\begin{equation*}
Z-p^{-1}\mathcal{P}_{\mathcal{\widetilde{T}}}\mathcal{R'}Z=\sum_{ij}(1-p^{-1}\delta_{ij})\langle Z,\omega_{ij}\rangle\mathcal{P}_\mathcal{\widetilde{T}}\omega_{ij},
\end{equation*}
and
\begin{equation*}
\langle Z-p^{-1}\mathcal{P}_{\mathcal{\widetilde{T}}}\mathcal{R'}Z,\omega_{ab}\rangle=\sum_{ij}(1-p^{-1}\delta_{ij})\langle Z,\omega_{ij}\rangle\langle\mathcal{P}_\mathcal{\widetilde{T}}\omega_{ij},\omega_{ab}\rangle.
\end{equation*}
We now want to invoke the scalar Bernstein inequality. Let $X_{ij}=(1-p^{-1}\delta_{ij})\langle Z,\omega_{ij}\rangle\langle\mathcal{P}_\mathcal{\widetilde{T}}\omega_{ij},\omega_{ab}\rangle$ with zero mean.
\begin{equation*}
\begin{split}
|X_{ij}|&=|(1-p^{-1}\delta_{ij})\langle Z,\omega_{ij}\rangle\langle\mathcal{P}_\mathcal{\widetilde{T}}\omega_{ij},\omega_{ab}\rangle|\\
&\le|(1-p^{-1}\delta_{ij})|\max_{ij}\left|\langle Z,\omega_{ij}\rangle\right|\left\Vert\mathcal{P}_\mathcal{\widetilde{T}}\omega_{ij}\right\Vert_F\left\Vert\mathcal{P}_\mathcal{\widetilde{T}}\omega_{ab}\right\Vert_F\\
&\le\frac{2\mu r}{n_{(2)}p}\max_{ab}\left|\langle Z,\omega_{ab}\rangle\right|\\
&\triangleq L.
\end{split}
\end{equation*}
Furthermore,
\begin{equation*}
\begin{split}
\sum_{ij}\mathbb{E}X_{ij}^2&=\sum_{ij}\mathbb{E}(1-p^{-1}\delta_{ij})^2\langle Z,\omega_{ij}\rangle^2\langle\mathcal{P}_\mathcal{\widetilde{T}}\omega_{ij},\omega_{ab}\rangle^2\\
&=(p^{-1}-1)\sum_{ij}\langle Z,\omega_{ij}\rangle^2\langle\mathcal{P}_\mathcal{\widetilde{T}}\omega_{ij},\omega_{ab}\rangle^2\\
&=(p^{-1}-1)\max_{ij}\langle Z,\omega_{ij}\rangle^2\sum_{ij}\langle\omega_{ij},\mathcal{P}_\mathcal{\widetilde{T}}\omega_{ab}\rangle^2\\
&=(p^{-1}-1)\max_{ij}\langle Z,\omega_{ij}\rangle^2\left\Vert\mathcal{P}_\mathcal{\widetilde{T}}\omega_{ab}\right\Vert_F^2\\
&\le\frac{2\mu r}{n_{(2)}p}\max_{ab}\langle Z,\omega_{ab}\rangle^2\\
&\triangleq M.
\end{split}
\end{equation*}
Since $M/L=\max_{ab}\left|\langle Z,\omega_{ab}\rangle\right|>\epsilon\max_{ab}\left|\langle Z,\omega_{ab}\rangle\right|$, by scalar Bernstein inequality, we obtain
\begin{equation*}
\begin{split}
\mathbb{P}\left\{\max_{ab}\left|\langle Z-p^{-1}\mathcal{P}_{\mathcal{\widetilde{T}}}\mathcal{R'}Z,\omega_{ab}\rangle\right|<\varepsilon\max_{ab}\left|\langle Z,\omega_{ab}\rangle\right|\right\}&\le2\exp\left\{\frac{-3\epsilon^2\max_{ab}\langle Z,\omega_{ab}\rangle^2}{8M}\right\}\\
&=2\exp\left\{\frac{-3\epsilon^2n_{(2)}p}{16\mu r}\right\}\\
&\le n_{(1)}^{-10},
\end{split}
\end{equation*}
provided that $p\ge C_0\epsilon^{-2}\mu r\log n_{(1)}/n_{(2)}$ for some numerical constant $C_0$.
\end{proof}

\subsection{Proofs of Lemma \ref{lemma: Omega 2 infty norm}}
\label{Sec: Proofs of Lemma 13}
We are prepared to prove Lemma \ref{lemma: Omega 2 infty norm}.
\begin{proof}
From the definition of operator $\mathcal{R}'$, we know that
\begin{equation*}
\mathcal{R'}(Z)=\sum_{ij\in\mathcal{K}}\langle Z,\omega_{ij}\rangle\omega_{ij}=\sum_{ij}\delta_{ij}\langle Z,\omega_{ij}\rangle\omega_{ij},
\end{equation*}
where $\delta_{ij}$s are i.i.d. Bernoulli variables with parameter $p$. So
\begin{equation*}
Z-p^{-1}\mathcal{R'}Z=\sum_{ij}(1-p^{-1}\delta_{ij})\langle Z,\omega_{ij}\rangle\omega_{ij}.
\end{equation*}
Let $X_{ij}=(1-p^{-1}\delta_{ij})\langle Z,\omega_{ij}\rangle\omega_{ij}$. To use the matrix Bernstein inequality, we need to bound $X_{ij}$ and its variance. To this end, note that
\begin{equation*}
\begin{split}
\|X_{ij}\|&=|1-p^{-1}\delta_{ij}|\ |\langle Z,\omega_{ij}\rangle|\ \|\omega_{ij}\|\\
&\le p^{-1}\|\omega_{ij}\|_F\max_{ij}|\langle Z,\omega_{ij}\rangle|\\
&=p^{-1}\max_{ij}|\langle Z,\omega_{ij}\rangle|\\
&\triangleq L.
\end{split}
\end{equation*}
Furthermore,
\begin{equation*}
\begin{split}
\left\Vert\sum_{ij}\mathbb{E}X_{ij}X_{ij}^*\right\Vert&=\left\Vert\sum_{ij}\mathbb{E}(1-p^{-1}\delta_{ij})^2\langle Z,\omega_{ij}\rangle^2\omega_{ij}\omega_{ij}^*\right\Vert\\
&\le p^{-1}\max_{ij}\langle Z,\omega_{ij}\rangle^2\left\Vert\sum_{ij}\omega_{ij}\omega_{ij}^*\right\Vert\\
&=p^{-1}\max_{ij}\langle Z,\omega_{ij}\rangle^2\left\Vert nI_m\right\Vert\\
&=\frac{n}{p}\max_{ij}\langle Z,\omega_{ij}\rangle^2.
\end{split}
\end{equation*}
Similarly,
\begin{equation*}
\left\Vert\sum_{ij}\mathbb{E}X_{ij}^*X_{ij}\right\Vert\le\frac{m}{p}\max_{ij}\langle Z,\omega_{ij}\rangle^2.
\end{equation*}
We now let $M=n_{(1)}\max_{ij}\langle Z,\omega_{ij}\rangle^2/p$ and set $t$ as $C_0'\sqrt{p^{-1}n_{(1)}\log n_{(1)}}\max_{ij}|\langle Z,\omega_{ij}\rangle|$. Since $M/L=n_{(1)}\max_{ij}|\langle Z,\omega_{ij}\rangle|>t$, by the matrix Bernstein inequality, we obtain
\begin{equation*}
\begin{split}
\mathbb{P}\left\{\|Z-p^{-1}\mathcal{R'}Z\|<C_0'\sqrt{\frac{n_{(1)}\log n_{(1)}}{p}}\max_{ij}\left|\langle Z,\omega_{ij}\rangle\right|\right\}&=\mathbb{P}\left\{\|Z-p^{-1}\mathcal{R'}Z\|<t\right\}\\
&=(m+n)\exp\left\{\frac{-3t^2}{8M}\right\}\\
&\le n_{(1)}^{-10}.
\end{split}
\end{equation*}
\end{proof}

\subsection{Proofs of Lemma \ref{lemma: relation 2 norm}}
\label{Sec: Proofs of Lemma 14}
We proceed to prove Lemma \ref{lemma: relation 2 norm}.
\begin{proof}
For any fixed matrix $Z$, it can be seen that
\begin{equation*}
\mathcal{R'}\mathcal{P}_\mathcal{\widehat{V}}Z=\sum_{ij\in\Omega_{obs}}\langle\mathcal{P}_\mathcal{\widehat{V}}Z,\omega_{ij}\rangle\omega_{ij}=\sum_{ij}\kappa_{ij}\langle Z,\mathcal{P}_\mathcal{\widehat{V}}\omega_{ij}\rangle\omega_{ij}.
\end{equation*}
Note that the operators $\mathcal{R'}$ and $\mathcal{P}_\mathcal{I}$ are commutative according to \eqref{equ: commutative operators}, thus we have
\begin{equation*}
\mathcal{P}_\mathcal{\widehat{V}}\mathcal{R'}\mathcal{P}_\mathcal{I}\mathcal{R'}\mathcal{P}_\mathcal{\widehat{V}}Z=\sum_{j}\delta_j\sum_{i}\kappa_{ij}\langle Z,\mathcal{P}_\mathcal{\widehat{V}}\omega_{ij}\rangle\mathcal{P}_\mathcal{\widehat{V}}\omega_{ij}.
\end{equation*}
Similarly, $\mathcal{P}_\mathcal{\widehat{V}}\mathcal{R'}\mathcal{P}_\mathcal{\widehat{V}}Z=\sum_{j}\sum_{i}\kappa_{ij}\langle Z,\mathcal{P}_\mathcal{\widehat{V}}\omega_{ij}\rangle\mathcal{P}_\mathcal{\widehat{V}}\omega_{ij}$, and so
\begin{equation*}
\begin{split}
(a^{-1}\mathcal{P}_\mathcal{\widehat{V}}\mathcal{R'}\mathcal{P}_\mathcal{I}\mathcal{R'}\mathcal{P}_\mathcal{\widehat{V}}-\mathcal{P}_\mathcal{\widehat{V}}\mathcal{R'}\mathcal{P}_\mathcal{\widehat{V}})Z=\sum_{j}(a^{-1}\delta_j-1)\sum_{i}\kappa_{ij}\langle Z,\mathcal{P}_\mathcal{\widehat{V}}\omega_{ij}\rangle\mathcal{P}_\mathcal{\widehat{V}}\omega_{ij}.
\end{split}
\end{equation*}
Namely,
\begin{equation*}
\begin{split}
a^{-1}\mathcal{P}_\mathcal{\widehat{V}}\mathcal{R'}\mathcal{P}_\mathcal{I}\mathcal{R'}\mathcal{P}_\mathcal{\widehat{V}}-\mathcal{P}_\mathcal{\widehat{V}}\mathcal{R'}\mathcal{P}_\mathcal{\widehat{V}}=\sum_{j}(a^{-1}\delta_j-1)\sum_{i}\kappa_{ij}\mathcal{P}_\mathcal{\widehat{V}}\omega_{ij}\otimes\mathcal{P}_\mathcal{\widehat{V}}\omega_{ij}.
\end{split}
\end{equation*}

We now plan to use concentration inequality. Let $X_j\triangleq(a^{-1}\delta_j-1)\sum_{i}\kappa_{ij}\mathcal{P}_\mathcal{\widehat{V}}\omega_{ij}\otimes\mathcal{P}_\mathcal{\widehat{V}}\omega_{ij}$. Notice that $X_j$ is zero-mean and self-adjoint. Denote the set $g=\{\|C_1\|_F\le1,C_2=\pm C_1\}$. Then we have
\begin{equation*}
\begin{split}
\|X_j\|&=\sup_g\langle C_1, X_j(C_2)\rangle=\sup_g\left|a^{-1}(\delta_j-a)\right|\left|\sum_i\kappa_{ij}\langle C_1, \mathcal{P}_\mathcal{\widehat{V}}(\omega_{ij})\rangle\langle C_2, \mathcal{P}_\mathcal{\widehat{V}}(\omega_{ij})\rangle\right|\triangleq|a^{-1}(\delta_j-a)|\sup_g|f(\delta_j)|.
\end{split}
\end{equation*}
According to \eqref{equ: commutative operators},
\begin{equation*}
\begin{split}
\|\mathcal{P}_\mathcal{\widehat{V}}C_1\|_{2,\infty}^2&=\max_j\sum_i\langle C_1,\omega_{ij}\widehat{V}\widehat{V}^*\rangle^2\\
&=\max_j\sum_i\langle G_j^*e_i^*C_1,e_j^*\widehat{V}\widehat{V}^*\rangle^2\\
&\le\max_j\sum_i\|e_i^*C_1\|_2^2\|e_j^*\widehat{V}\widehat{V}^*\|_2^2\\
&=\max_j\|C_1\|_F^2\|e_j^*\widehat{V}\widehat{V}^*\|_2^2\\
&\le\frac{\mu r}{n},
\end{split}
\end{equation*}
where $G_j$ is a unitary matrix. So we have
\begin{equation*}
\begin{split}
|f(\delta_j)|&\le \sum_i|\langle C_1,\mathcal{P}_\mathcal{\widehat{V}}(\omega_{ij})\rangle|\ |\langle C_2,\mathcal{P}_\mathcal{\widehat{V}}(\omega_{ij})\rangle|\\
&=\sum_i\langle C_1,\mathcal{P}_\mathcal{\widehat{V}}(\omega_{ij})\rangle^2\\
&\le \sum_i\langle \mathcal{P}_\mathcal{\widehat{V}}C_1, \omega_{ij}\rangle^2\\
&\le \|\mathcal{P}_\mathcal{\widehat{V}}C_1\|_{2,\infty}^2\\&\le\frac{\mu r}{n},
\end{split}
\end{equation*}
where the first identity holds since $C_2=\pm C_1$. Thus $\|X_j\|\le\mu ra^{-1}n^{-1}\triangleq L$. We now bound $\sum_j\left\Vert\mathbb{E}_{\delta_j}X_j^2\right\Vert$. Observe that
\begin{equation*}
\begin{split}
\|\mathbb{E}_{\delta_j}X_j^2\|\le\mathbb{E}_{\delta_j}\|X_j^2\|=\mathbb{E}_{\delta_j}\|X_j\|^2=\mathbb{E}_{\delta_j}a^{-2}(\delta_j-a)^2\sup_gf(\delta_j)^2,
\end{split}
\end{equation*}
where the last identity holds because $C_1,\ C_2$ and $\delta_j$ are separable.
Furthermore,
\begin{equation*}
\begin{split}
\sup_gf(\delta_j)^2&=\sup_g\left(\sum_i\kappa_{ij}\langle C_1, \mathcal{P}_\mathcal{\widehat{V}}(\omega_{ij})\rangle\langle C_2, \mathcal{P}_\mathcal{\widehat{V}}(\omega_{ij})\rangle\right)^2\\
&\le \sup_g\left(\sum_i|\langle C_1, \mathcal{P}_\mathcal{\widehat{V}}(\omega_{ij})\rangle\langle C_2, \mathcal{P}_\mathcal{\widehat{V}}(\omega_{ij})\rangle|\right)^2\\
&= \left(\sum_i\langle C_1, \mathcal{P}_\mathcal{\widehat{V}}(\omega_{ij})\rangle^2\right)^2\\
&\le \left(\sum_i\langle \mathcal{P}_\mathcal{\widehat{V}}C_1, \omega_{ij}\rangle^2\right)\left(\sum_i\langle \mathcal{P}_\mathcal{\widehat{V}}C_1, \omega_{ij}\rangle^2\right)\\
&\le \|\mathcal{P}_\mathcal{\widehat{V}}C_1\|_{2,\infty}^2\sum_i\langle \mathcal{P}_\mathcal{\widehat{V}}C_1, \omega_{ij}\rangle^2\\
&\le \frac{\mu r}{n}\sum_i\langle \mathcal{P}_\mathcal{\widehat{V}}C_1, \omega_{ij}\rangle^2.
\end{split}
\end{equation*}
Therefore,
\begin{equation*}
\begin{split}
\sum_j\left\Vert\mathbb{E}_{\delta_j}X_j^2\right\Vert&\le\mathbb{E}_{\delta_j}a^{-2}(\delta_j-a)^2\frac{\mu r}{n}\sum_{ij}\langle \mathcal{P}_\mathcal{\widehat{V}}C_1, \omega_{ij}\rangle^2\\
&=\frac{\mu r(1-a)}{na}\|\mathcal{P}_\mathcal{\widehat{V}}C_1\|_F^2\\
&\le\frac{\mu r}{na}\\
&\triangleq M.
\end{split}
\end{equation*}
Since $M/L=1>\epsilon$, by the matrix Bernstein inequality,
\begin{equation*}
\begin{split}
\mathbb{P}\left\{\left\Vert a^{-1}\mathcal{P}_\mathcal{\widehat{V}}\mathcal{R'}\mathcal{P}_\mathcal{I}\mathcal{R'}\mathcal{P}_\mathcal{\widehat{V}}-\mathcal{P}_\mathcal{\widehat{V}}\mathcal{R'}\mathcal{P}_\mathcal{\widehat{V}}\right\Vert<\varepsilon\right\}&=\mathbb{P}\left\{\left\Vert \sum_j X_j\right\Vert<\varepsilon\right\}\\
&\le(m+n)\exp\left\{\frac{-3\epsilon^2}{8M}\right\}\\
&=(m+n)\exp\left\{\frac{-3\epsilon^2na}{8\mu r}\right\}\\
&\le n_{(1)}^{-10},
\end{split}
\end{equation*}
provided that $a\ge C_0\varepsilon^{-2}(\mu r\log n_{(1)})/n$ for some numerical constant $C_0>0$.
\end{proof}

\end{appendix}

\begin{IEEEbiographynophoto}{Hongyang Zhang}
received the Master's degree in computer science from Peking University, Beijing, China in 2015. He is now a Ph.D. student in Machine Learning Department, Carnegie Mellon University, Pittsburgh, USA. His research interests include machine learning, statistics, and numerical optimization.
\end{IEEEbiographynophoto}

\begin{IEEEbiographynophoto}{Zhouchen Lin}
(M'00-SM'08) received the Ph.D. degree in applied mathematics from Peking University in 2000. He is currently a Professor with the Key Laboratory of Machine Perception, School of Electronics Engineering and Computer Science, Peking University. He is also a Chair Professor with Northeast Normal University. He was a Guest Professor with Shanghai Jiao Tong University, Beijing Jiaotong University, and Southeast University. He was also a Guest Researcher with the Institute of Computing Technology, Chinese Academy of Sciences. His research interests include computer vision, image processing, machine learning, pattern recognition, and numerical optimization. He was an area chair of CVPR 2014, ICCV 2015, NIPS 2015, AAAI 2016, IJCAI 2016, and CVPR 2016. He is an Associate Editor of the IEEE Transactions on Pattern Analysis and Machine Intelligence and the International Journal of Computer Vision.
\end{IEEEbiographynophoto}

\begin{IEEEbiographynophoto}{Chao Zhang}
(M'06) received the Ph.D. degree in electrical engineering from Beijing Jiaotong University, Beijing, China, in 1995. He was a Post-Doctoral Research Fellow with the National Laboratory on Machine Perception, Peking
University, Beijing, from 1995 to 1997. He has been an Associate Professor with the Key Laboratory of Machine Perception, School of Electronics Engineering and Computer Science, Peking University, since 1997. His current research interests include image processing, statistical pattern recognition, and visual recognition.
\end{IEEEbiographynophoto}

\end{document}